\documentclass[11pt]{article}

\makeatletter
\newif\ifanonymous
\@ifclasswith{article}{anonymous}{\anonymoustrue}{\anonymousfalse}
\makeatother

\usepackage[hmargin=1in,vmargin=1in]{geometry}
\usepackage{hchang}
\hypersetup{hypertexnames=false}
\nolinenumbers
\usepackage{thm-restate}
\usepackage[nameinlink,capitalize,noabbrev]{cleveref}
\usepackage{comment}
\usepackage{tcolorbox}
\usepackage{float}
\usepackage{orcidlink}
\usepackage{fontawesome5}
\usepackage[thinlines]{easytable}
\usepackage{makecell}
\usepackage{graphicx}
\usepackage{algorithm}
\usepackage{algpseudocode}

\newcommand{\emaillink}[1]{\href{mailto:#1}{\raisebox{0.05em}{\textcolor{gray}{\fontsize{9}{12}\sffamily\bfseries\faIcon[regular]{envelope}}}}}

\newtheorem{theorem}{Theorem}[section]
\newtheorem{lemma}[theorem]{Lemma}

\newtheorem{corollary}[theorem]{Corollary}

\newtheorem{definition}[theorem]{Definition}

\newtheorem{proposition}[theorem]{Proposition}
\newtheorem{example}[theorem]{Example}

\crefname{theorem}{Theorem}{Theorems}
\Crefname{theorem}{Theorem}{Theorems}
\crefname{lemma}{Lemma}{Lemmas}
\Crefname{lemma}{Lemma}{Lemmas}
\crefname{claim}{Claim}{Claims}
\Crefname{claim}{Claim}{Claims}
\crefname{observation}{Observation}{Observations}
\Crefname{observation}{Observation}{Observations}
\crefname{problem}{Problem}{Problems}
\Crefname{problem}{Problem}{Problems}
\crefname{remark}{Remark}{Remarks}
\Crefname{remark}{Remark}{Remarks}
\crefname{corollary}{Corollary}{Corollaries}
\Crefname{corollary}{Corollary}{Corollaries}
\crefname{construction}{Construction}{Constructions}
\Crefname{construction}{Construction}{Constructions}
\crefname{definition}{Definition}{Definitions}
\Crefname{definition}{Definition}{Definitions}
\crefname{question}{Question}{Questions}
\Crefname{question}{Question}{Questions}
\crefname{conjecture}{Conjecture}{Conjectures}
\Crefname{conjecture}{Conjecture}{Conjectures}
\crefname{proposition}{Proposition}{Propositions}
\Crefname{proposition}{Proposition}{Propositions}
\crefname{example}{Example}{Examples}
\Crefname{example}{Example}{Examples}

\newcommand{\RR}{\mathbb{R}}

\newcommand{\NN}{\mathbb{N}}
\newcommand{\defn}[1]{\textbf{#1}}

\DeclareMathOperator{\EX}{\mathbb{E}}
\DeclareMathOperator{\OO}{\mathcal{O}}
\DeclareMathOperator{\oo}{o}

\newcommand{\remove}[1]{}

\title{A Robust and Unified SDP Approach for Vertex Separators:\\
Bounded-Genus Graphs, Minor-Free Graphs, and More}
\title{Reweighted Spectral Partitioning Works:\\A Simple Algorithm for Vertex Separators in Special Graph Classes}
\ifanonymous
    \author{Anonymous author(s)}
\else
    \author{Jack Spalding-Jamieson\,\orcidlink{0000-0002-1209-4345}\,\emaillink{jacksj@uwaterloo.ca}}
\fi
\date{}

\begin{document}

\maketitle

\begin{abstract}
We establish that a practical polynomial-time algorithm
obtains near-optimal $\frac23$-balanced vertex separators
for a number of graph classes,
including genus-$g$ graphs, $K_h$-minor-free graphs,
and geometric graphs.
For many of these classes,
our polynomial-time separator bounds improve on previous work.
Moreover, unlike previous work,
all these separators are produced by a single unified algorithm
that has no specialization for the graph class,
and this algorithm admits a scalable implementation with good empirical results.
Furthermore,
our approach is robust to
graphs falling just outside these classes,
making it appropriate for practical applications
like road networks.

A key distinction of our approach, compared to previous work,
is the disentanglement of our algorithm
and its analysis for specific graph classes.
Our algorithm is based on a newly improved form
of a recent Cheeger-style inequality
relating the value of a special semidefinite program
to vertex expansion.
Our analysis then only has to directly bound the value of this semidefinite program
for each graph class,
and constructions for these bounds need not be algorithmic.
Finally, we obtain several results of independent interest.
These include an improved (existential) separator theorem for intersection graphs of $k$-ply $d$-dimensional balls,
an improved bound on the Fiedler value of genus-$g$ graphs,
and a complete spectral proof of the planar separator theorem.
\end{abstract}

\section{Introduction}
\label{sec:introduction}

Let $G=(V,E)$ be a graph with $n$ vertices.  For
$\beta\in(1/2,1)$, a set $S\subseteq V$ is an
\defn{$\beta$-balanced vertex separator}, or just a \defn{separator},
if every connected component of
$G[V\setminus S]$ has at most $\beta n$ vertices.
Small separators are a basic primitive for divide-and-conquer algorithms,
dynamic programming on graphs, nested dissection, graph layout, and graph
decomposition methods~\cite{separatorsbook}.

Almost half a century ago,
Lipton and Tarjan showed that
planar graphs admit
$\OO(\sqrt n)$-sized separators~\cite{ungar1951theorem,lipton1979separator}.
A planar embedding may be found in linear time~\cite{hopcroft1974efficient},
so this gives a linear time algorithm for finding such separators.
This dependence on $n$ is tight for square grids.

Since this landmark result,
a large number of other separator theorems have been discovered.
In 1984, Gilbert, Hutchinson, and Tarjan
published a generalization to genus-$g$ graphs (planar graphs have genus $0$),
showing that they admit $\OO(\sqrt{gn})$-sized
separators~\cite{gilbert1984separator}.
More recently, in 2010,
a long line of work culminated in an almost complete generalization
to $K_h$-minor-free graphs (genus-$g$ graphs are $K_{O(\sqrt{g})}$-minor-free),
resulting in a proof that they admit $\OO(h\sqrt{n}+f(h))$-sized
separators~\cite{kawarabayashi2010separator}
for a function $f$.\footnote{$f$ involves multiple tower functions,
and a full proof of this result has not yet been published.}

Unfortunately, unlike the planar case,
neither of the corresponding algorithms
are efficient in general.
The result for genus-$g$ graphs
can be run in linear time once a surface embedding is computed,
but computing such an embedding requires
time exponential in $g$~\cite{kawarabayashi2008simpler}.
The result for $K_h$-minor-free graphs is even worse,
requiring a time complexity consisting of multiple tower functions in $h$~\cite{kawarabayashi2010separator}.
Both of these algorithms are fixed-parameter tractable,
but that remains undesirable in a number of scenarios,
and significant effort has been put into
attaining slightly weaker results in truly polynomial time
(see \cref{tab:minor-free-history-rough}).

\begin{table}[h]
\centering
\scriptsize
\setlength{\tabcolsep}{1.6pt}
\renewcommand{\arraystretch}{1.12}
\begin{tabular}{>{\centering\arraybackslash}p{0.08\textwidth}
                >{\raggedright\arraybackslash}p{0.08\textwidth}
                >{\raggedright\arraybackslash}p{0.31\textwidth}
                >{\raggedright\arraybackslash}p{0.14\textwidth}
                >{\raggedright\arraybackslash}p{0.29\textwidth}}
\hline
\textbf{Year}
& \textbf{Reference}
& \textbf{Separator size}
& \textbf{Running time}
& \textbf{Notes}\\
\hline
1990
& \cite{alon1990separatorJournal}
& $\OO(h^{3/2}\sqrt{n})$
& $\OO(h^{1/2}\sqrt{n}\cdot m)$
& \\
\hdashline
1994
& \cite{plotkin1994shallow}
& $\OO(h\sqrt{n\log n})$
& $\OO(h\sqrt{n\log n}\cdot m)$
& \\
\hdashline
2008
& \cite{biswal2010eigenvalue}
& $\OO(h^3\sqrt{n\log h})$
& $\widetilde{\OO}({\rm poly}(n))$
& \\
\hdashline
2008
& \cite{biswal2010eigenvalue}
& $\OO(\Delta h^3\sqrt{n\log h})$
& $\widetilde{\OO}(nm)$
& \\
\hdashline
2009
& \cite{reed2009linear}
& $2^{\OO(h^2)}n^{\frac23-\delta}$
& $2^{\OO(h^2)}n^{1+\delta}$
&  For a parameter $\delta\in\left[0,\frac12\right]$ \\
\hdashline
2010
& \cite{kawarabayashi2010separator}
& $\OO(h\sqrt{n}+f(h))$
& $\OO(f(h)n^{1+\epsilon})$
&
$\epsilon>0$ is an arbitrary constant, $f$ is a very fast-growing function. Full proof absent. \\
\hdashline
2011
& \cite{wulff2011separator}
& $\OO(\poly(h)n^{4/5+\epsilon})$
& $\OO(\poly(h)\,n)$
& $\epsilon>0$ is an arbitrary constant \\
\hdashline
2011
& \cite{wulff2011separator}
& $\OO(h\sqrt{n\log n})$
& $\OO(\poly(h)\,n^{5/4+\epsilon})$
& $\epsilon>0$ is an arbitrary constant\\
\hdashline
2026
& \cite{bonnet2026separator}
& $\OO(h^{11}\sqrt n)$
& $\OO(h^{11}n)$
& \\
\hdashline
\multicolumn{2}{c}{\textbf{This paper}}
& $\OO(\min\{\log h,\sqrt{\log\Delta}\}\cdot h\log h\log\log h\sqrt{n})$
& $\OO(n^4\cdot{\rm polylog}(n))$
& Robust to outliers
\\
\hdashline
\multicolumn{2}{c}{\textbf{This paper}}
& $\OO(\sqrt{\log\Delta}\cdot h\log h\log\log h\sqrt{n})$
& $\OO(n^2m\cdot{\rm polylog}(n))$
& Robust to outliers
\\
\hline
\multicolumn{5}{c}{\textbf{Genus-$g$ only} $(h=\Theta(\sqrt g))$}\\
\hdashline[2pt/2pt]
1984 \& 2008
& \cite{gilbert1984separator,kawarabayashi2008simpler}
& $\OO(\sqrt{gn})$
& $2^{\OO(\mathrm{poly}(g))}n$
& \\
\hdashline
2004
& \cite{kelner2006spectral}
& $\OO(\poly(\Delta)\sqrt{gn})$
& $\widetilde{\OO}(nm)$
& \\
\hdashline
2008
& \cite{biswal2010eigenvalue}
& $\OO(\Delta g^{3/2}\sqrt n)$
& $\widetilde{\OO}(nm)$
& \\
\hdashline
2010
& \cite{lee2010genus}
& $\OO(\log g\sqrt{gn})$
& $\widetilde{\OO}({\rm poly}(n))$
& Citation's claim differs. See \cref{sec:short-disproof}.
\\
\hdashline
2018
& \cite{amini2018transfer}
& $\OO(\Delta^{3/2}\sqrt{gn})$
& $\widetilde{\OO}(nm)$
& \\
\hdashline
2026
& \cite{bonnet2026separator}
& $\OO(\mathrm{poly}(g)\sqrt n)$
& $\OO(\mathrm{poly}(g)n)$
& \\
\hdashline
\multicolumn{2}{c}{\textbf{This paper}}
& $\OO(\min\{(\log g)^2,\log\Delta\}\sqrt{gn})$
& $\OO(n^4{\rm polylog}(n))$
& Robust to outliers
\\
\hdashline
\multicolumn{2}{c}{\textbf{This paper}}
& $\OO((\log \Delta)\cdot\sqrt{gn})$
& $\OO(n^3{\rm polylog}(n))$
& Robust to outliers
\\
\hline
\end{tabular}
\caption{History of separator algorithms for $K_h$-minor-free graphs and genus-$g$ graphs.
$n$, $m$, and $\Delta$ denote the counts of vertices, edges, and maximum degree, respectively.
}
\label{tab:minor-free-history-rough}
\label{tab:genus-history-rough}
\end{table}

Some remarkable results
have arisen from the study of these polynomial-time algorithms,
such as a polynomial-time algorithm
to compute $\OO(h^{3/2}\sqrt{n})$-sized
separators
for a $K_h$-minor-free graph~\cite{alon1990separator,alon1990separatorJournal}.
However,
for these weaker results,
and even for the case of planar graphs,
two significant issues remain
besides their separator sizes.
First,
the algorithms are quite brittle.
Second,
the algorithms are often quite complex
and highly specialized.
Both of these issues are well-demonstrated
by the practical example of road networks.
Road networks empirically admit good separators
that are critical to modern navigation~\cite{delling2011graph}.
However,
the best methods we have for computing these
separators are heuristics.
This is because road networks are only \emph{almost} planar,
due to the existence of \defn{outliers} like
bridges, overpasses, and similar structures.
Moreover, while it is entirely possible
that many road networks are, say,
$K_h$-minor-free for a small $h$,
the algorithms we have for this class
are complex and would be impractical to implement.

Our main result overcomes
these challenges,
and simultaneously
provides
improved polynomial-time computable separator sizes
in genus-$g$, $K_h$-minor-free, and geometric graph classes:

\begin{theorem}[Informal main theorem]
\label{thm:informal-main}
There is a
polynomial-time algorithm
for computing
vertex separators
with the following sizes depending on the graph class
(where $\Delta$ denotes maximum degree):
\begin{itemize}[leftmargin=2.2em,itemsep=1pt,topsep=3pt]
  \item planar: $\OO(\sqrt n)$
  \item genus-$g$:
  $\OO(\min\{(\log g)^2,\log\Delta\}\sqrt{gn})$
  \item $K_h$-minor-free:
  $\OO(\min\{\log h,\sqrt{\log\Delta}\}\,h\log h\log\log h\sqrt n)$
  \item $d$-dimensional geometric graph (either $k$-NN graph or $k$-ply ball graph):
  $\OO(\sqrt{\log\Delta}\,k^{1/d}n^{1-1/d})$
  \item every subgraph with $n'$ vertices admits separator of size $\OO\left(({n'})^{a}\right)$: $\OO(\sqrt{\log\Delta}\cdot n^{(1+a)/2})$ (or better)
\end{itemize}
Moreover, the same bounds hold even if
the graph only falls into one of these classes
after a small number of simple operations.
\end{theorem}

Furthermore, the algorithm
used by
\cref{thm:informal-main}
is practical,
able to run on graphs with a few hundred thousand vertices and edges
in a matter of seconds,
or faster with GPU acceleration
(see the larger graphs in
\cref{fig:separator-examples}).

\begin{figure}[H]
\centering
\begin{minipage}{0.18\textwidth}
\centering
\includegraphics[width=\linewidth,height=0.095\textheight,keepaspectratio]{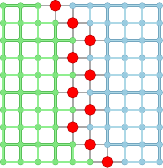}
\end{minipage}
\hfill
\begin{minipage}{0.18\textwidth}
\centering
\includegraphics[width=\linewidth,height=0.095\textheight,keepaspectratio]{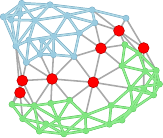}
\end{minipage}
\hfill
\begin{minipage}{0.18\textwidth}
\centering
\includegraphics[width=\linewidth,height=0.095\textheight,keepaspectratio]{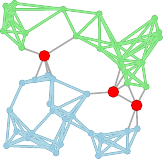}
\end{minipage}
\hfill
\begin{minipage}{0.18\textwidth}
\centering
\includegraphics[width=\linewidth,height=0.095\textheight,keepaspectratio]{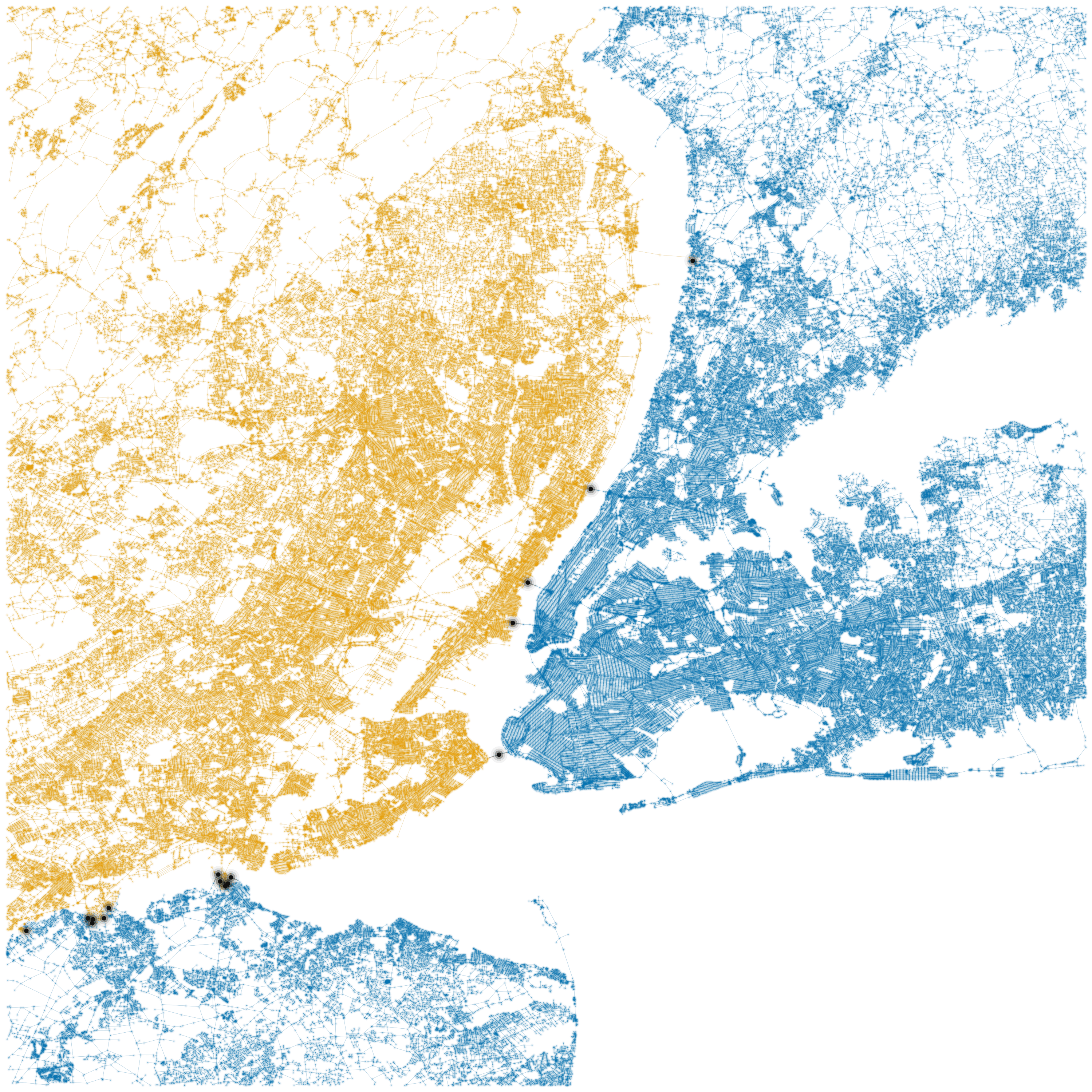}
\end{minipage}
\hfill
\begin{minipage}{0.18\textwidth}
\centering
\includegraphics[width=\linewidth,height=0.095\textheight,keepaspectratio]{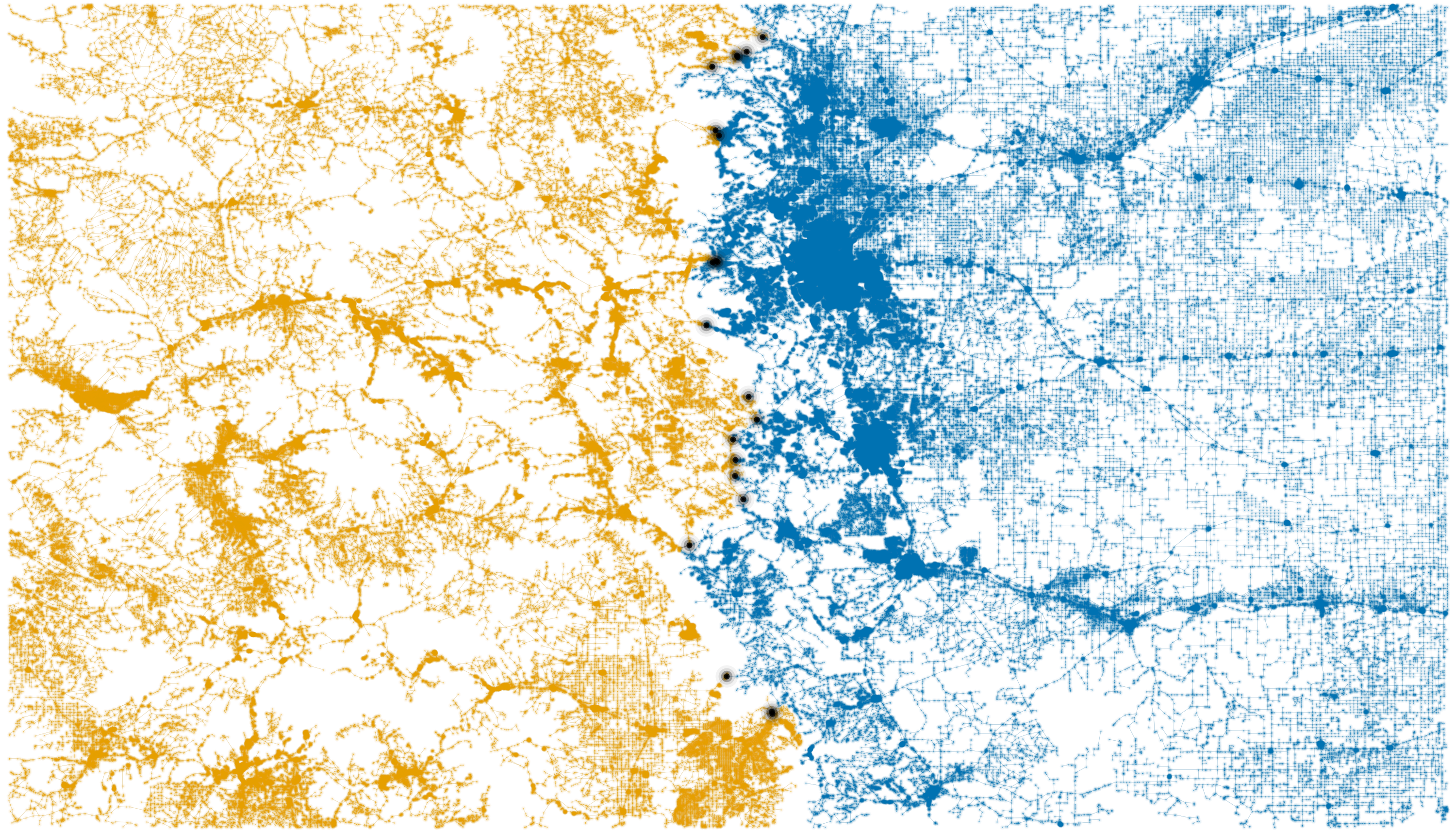}
\end{minipage}
\caption{Example separators produced by our algorithm.
The first three panels show small grid, planar, and $k$-nearest-neighbour
examples, respectively.
The last two use the DIMACS shortest-path road-network
instances for New York
($264{,}346$ vertices, $365{,}050$ edges)
and Colorado
($435{,}666$ vertices, $521{,}200$ edges)%
~\cite{demetrescu2009shortest}.
Separator vertices are highlighted.
}
\label{fig:separator-examples}
\end{figure}

Our proof involves two components.
First, we use a semidefinite programming relaxation of
``vertex expansion'',
a quantity closely related to vertex separators.
We denote the SDP by
$\gamma^{(n)}(G)$,
and round it through a new Cheeger-style inequality.
Second, we analyze the performance of this algorithm for special graph classes
by giving specialized bounds on the value of $\gamma^{(n)}(G)$.

Our analysis of this algorithm applied to special graph classes
uses two different kinds of methods.
First, we show that, for planar graphs, geometric graphs, and genus-$g$ graphs, $\gamma^{(n)}(G)$ can be bounded
\emph{geometrically} via the theory of circle packings.
Second, we bound $\gamma^{(n)}(G)$ in terms of some well-understood combinatorial quantities.
Both methods are generalizations/improvements of prior methods bounding Laplacian eigenvalues of graphs.
Specifically, the first method builds on and improves (and in some places fixes)
works of Spielman and Teng~\cite{spielman1996spectral,spielman2007spectral}
and Kelner~\cite{kelner2004spectral,kelner2006spectral,kelner2006new}.
The second method builds on works of Biswal, Lee, and Rao~\cite{biswal2008eigenvalue}.

\subsection{SDP Relaxation of Vertex Expansion: An Improved Cheeger-style Inequality}
\label{subsec:reweighted-eigenvalues-separator-bounds}

Throughout this subsection, let $G=(V,E)$ be a graph with $n=|V|$ vertices,
$m=|E|$ edges,
and maximum degree $\Delta$, and write $a\lesssim b$ to mean $a=\OO(b)$.

For a parameter $d\in\NN$, we define
the optimization problem
\begin{equation}
\label{eq:gamma-d}
\everymath{\displaystyle}
\begin{array}{r c c r c l l}
\gamma^{(d)}(G) & := & \min_{\substack{f:V\to\RR^d\\y:V\to\RR_{\geq0}}}
& \dfrac{\sum_{v\in V}y(v)}{\sum_{x\in V}||f(x)||_2^2} \\
&& \textup{subject to} & \sum_{v\in V} f(v) & = & \overline{0}\\
&& & y(u)+y(v) & \geq & ||f(u)-f(v)||_2^2 & \forall uv\in E.
\end{array}
\end{equation}
Our SDP relaxation is exactly $\gamma^{(n)}$.
Note that $\gamma^{(n)}\leq\gamma^{(d)}$ and $\gamma^{(1)}\leq d\gamma^{(d)}$
for any $d\leq n$.

Roch observed that $\gamma^{(n)}(G)$
is a relaxation of the \defn{vertex expansion} of $G$~\cite{roch2005bounding}:

\begin{definition}
For a subset $A\subset V$ with $|A|\leq n/2$,
let $\partial A$ denote the set of vertices in $V\setminus A$ adjacent to some vertex in $A$.
The \defn{vertex expansion} of a set $A$ is the ratio $\psi(A):=\frac{|\partial A|}{|A|}$.
The \defn{vertex expansion} of a graph $G$ is the minimum vertex expansion over all possible sets $|A|\leq\frac n2$,
and it is denoted as $\psi(G)$.
\end{definition}

Specifically, Roch observed that $\gamma^{(d)}\lesssim\psi(G)$ for any $d\geq1$.
It wasn't until 2022 that a partial inverse was found,
when Olesker-Taylor and Zanetti
showed that $\psi(G)^2/\log n\lesssim\gamma^{(n)}$~\cite{olesker2022geometric,olesker2024geometric}.
This inequality is also algorithmic,
in the sense that a feasible solution to $\gamma^{(n)}$
can be converted to a set $A$
so that $\psi(A)$ matches this bound.
Later in 2022,
two different groups independently improved this dependence on $\log n$ to $\log\Delta$~\cite{kwok2022cheeger,kwok2025cheeger,jain2022dimension}.

We further improve the $\log\Delta$ factor
to $\min\{\log\Delta,\alpha(G)^2\}$,
where we use $\alpha(G)$ to denote the \defn{extremal padding parameter} of $G$,
a quantity that we will define shortly.
This is a significant improvement for a number of graph classes,
including planar graphs, genus-$g$ graphs, and $K_h$-minor-free graphs,
where $\alpha(G)$ has values of $\OO(1)$,
$\OO(\log g)$, and $\OO(\log h)$, respectively~\cite{conroy2025protect}.
The history of these Cheeger-style inequalities is summarized in
\cref{tab:cheeger-history-rough},
and our full Cheeger-style inequality follows.

\begin{theorem}[Improved Cheeger-Style Inequality for Vertex Expansion]
\label{thm:refined-cheeger}
For a graph $G$ with $n$ vertices and maximum degree $\Delta$,
$$\frac{\psi(G)^2}{\min\{\log\Delta,\alpha(G)^2\}}
\lesssim\frac{\gamma^{(1)}(G)}{\min\{\log\Delta,\alpha(G)^2\}}
\lesssim\gamma^{(n)}(G)\lesssim\gamma^{(1)}(G)\lesssim\psi(G).$$
\end{theorem}
Note that the earlier, weaker inequality $\psi(G)^2/\log\Delta\lesssim\gamma^{(n)}$
corresponds to a similar analytic guarantee.
The algorithmic version of \cref{thm:refined-cheeger},
including the SDP-solving and padded-decomposition subroutines,
is stated in
\cref{thm:algorithmic-refined-cheeger,thm:gamma-sdp-solvers},
or discussed briefly in \cref{sec:rsp}.

We will prove this result in
\cref{sec:refined-cheeger}.
Note that
if for every $n'$-vertex-subgraph of $G$,
one can find such a set $S$ with vertex expansion $\psi(S)\leq a/(n')^b$
for values $a,b>0$ (where $b=\Theta(1)$),
then a standard recursive application (see e.g.~\cite{FeigeHL05,FeigeHL08,spielman2007spectral})
obtains a separator of size
$O(an^{1-b})$.
Thus,
\cref{thm:refined-cheeger}
captures the analytic part of our algorithm;
the remaining algorithmic aspects are isolated in \cref{sec:rsp}.

\begin{table}[t]
\centering
\normalsize
\setlength{\tabcolsep}{0pt}
\renewcommand{\arraystretch}{1.45}
\begin{tabular}{c@{\hspace{2.25em}}l@{\hspace{2.75em}}l}
\hline
\textbf{Year}
& \textbf{Reference}
& \textbf{Inequality}\\
\hline
2005
& \cite{roch2005bounding}
& $\gamma^{(n)}(G)\lesssim\psi(G)$\\
\hdashline
2022
& \cite{olesker2022geometric}
& $\psi(G)^2/\log n\lesssim\gamma^{(n)}(G)\lesssim\psi(G)$\\
\hdashline
2022
& \cite{jain2022dimension,kwok2022cheeger}
& $\psi(G)^2/\log\Delta\lesssim\gamma^{(n)}(G)\lesssim\psi(G)$\\
\hdashline
\multicolumn{2}{c@{\hspace{2.75em}}}{\textbf{This paper}}
& $\psi(G)^2/\min\{\log\Delta,\alpha(G)^2\}
  \lesssim\gamma^{(n)}(G)\lesssim\psi(G)$\\
\hline
\end{tabular}
\caption{History of Cheeger-style inequalities for vertex expansion and $\gamma^{(n)}(G)$.
$\alpha(G)$ denotes the extremal padding parameter of $G$.}
\label{tab:cheeger-history-rough}
\end{table}

\paragraph*{Padded Decompositions}
We quickly discuss the significant aspects of padded decompositions.

\begin{definition}
Let $(X,d)$ be a metric space.
The (weak) \defn{padding parameter} of $(X,d)$ is defined as the smallest value of $\alpha$
so that for every $\Delta>0$,
there exists a distribution $\mu$ of partitions of $X$
with parts of diameter at most $\Delta$,
and
\begin{equation*}
\label{eq:padded-decompose-padding}
\everymath{\displaystyle}
\begin{array}{c c c c}
\inf_{x\in X}
&
\Pr_{P\sim\mu}\left[B\left(x,\frac\Delta\alpha\right)\subset P(x)\right] & \geq & \Omega(1)
\end{array}
\end{equation*}
where $B(x,r)$ denotes the metric ball of radius $r$ centred at $x$.
We denote this quantity as $\alpha(X,d)$,
and we call the distribution $\mu$ a \defn{padded decomposition} of $(X,d)$.
\end{definition}

We define the \defn{extremal padding parameter on $G$}
to be the maximum value of $\alpha(V,d_\omega)$
for vertex-weight functions $\omega:V\to\RR_{\geq0}$
inducing shortest-path metrics $d_\omega$ on $G$.
If $G$ is not connected, we define its extremal padding parameter to be
the maximum such value over its connected components.
As previously noted, we denote this quantity as $\alpha(G)$.%
\footnote{A similar quantity can be defined for the more general case of edge-weighted shortest-path metrics, but this suffices for us.}

For the algorithmic version of
\cref{thm:refined-cheeger},
we use generic routines for producing padded decompositions.
The details are deferred to \cref{sec:rsp}.

\subsection{Class-Specific SDP Value Bounds}

As previously discussed,
the other component
required to prove
our main theorem
is to bound the value of the SDP $\gamma^{(n)}$
for each individual graph class we study.
We will prove all of the below bounds.

\begin{theorem}[Bounds on $\gamma^{(d)}$]
\label{thm:all-base-gamma-bounds}
Let $G$ be a graph with $n$ vertices and maximum degree $\Delta$.
Then the following conditional bounds hold:
\begin{itemize}[leftmargin=*]
  \item If $G$ is planar, then $\gamma^{(1)}(G)\lesssim\frac{1}{n}$.
  \item If $G$ has (oriented) genus $g$, then $\gamma^{(1)}(G)\lesssim\frac{g\cdot\min\{(\log g)^2,\log\Delta\}}{n}$.
  \item If $G$ is $K_h$-minor-free, then $\gamma^{(1)}(G)\lesssim\frac{(h\log h\log\log h)^{2}}{n}$.
  \item If $G$ is a spanning subgraph of a $d$-dimensional
    $k$-ply ball-intersection graph,\footnote{Sometimes called a neighbourhood-system.}
    then $\gamma^{(d+1)}(G)\lesssim\left(\frac kn\right)^{\frac2d}$.
  \item If $G$ is a spanning subgraph of a $d$-dimensional $k$-nearest-neighbour graph,
    then $\gamma^{(d+1)}(G)\lesssim\left(\frac kn\right)^{\frac2d}$.
\end{itemize}
\end{theorem}

\subsection{Robustness}
\label{subsec:robustness-rough}

A \defn{vertex merge} identifies two vertices,
which may or may not be adjacent,
and then deletes any resulting loops and suppresses parallel edges.
Vertex merges capture ordinary edge contractions,
but also allow us to contract a ``fake'' edge that is not present in $G$,
such as an overpass vertex in a road network.

\begin{lemma}[Robustness under vertex merges]
\label{lem:gamma-vertex-merges}
Let $H$ be obtained from $G$ by a sequence of vertex merge operations.
Let $b$ be the size of the largest collection of vertices in $G$
merged into a single vertex in $H$.
Then, for every $d\geq1$,
\[
\gamma^{(d)}(G)\leq b\gamma^{(d)}(H).
\]
In particular, if $H$ is obtained by $q$ vertex merge operations, then
\[
\gamma^{(d)}(G)\leq(q+1)\gamma^{(d)}(H).
\]
\end{lemma}

Thus,
if there is a sequence of $O(1)$ vertex merge operations that turns $G$
into a graph in one of the classes in
\cref{thm:all-base-gamma-bounds},
then the corresponding bound on $\gamma^{(n)}(G)$
is unchanged up to constants.

\begin{lemma}[Padded decompositions under vertex merges]
\label{lem:alpha-vertex-merges}
Let $H$ be obtained from $G$ by $q$ vertex merge operations.
Suppose every induced subgraph of $H$
has extremal padding parameter at most $\alpha_0$.
Then
\[
\alpha(G)\lesssim \alpha_0+\log(q+2).
\]
\end{lemma}

Proofs of the statements above are given in
\cref{sec:robustness-proofs}.

\subsection{Main Result}

Our main result is the combination of
\cref{thm:refined-cheeger},
\cref{thm:all-base-gamma-bounds},
and
\cref{lem:gamma-vertex-merges,lem:alpha-vertex-merges},
together with the algorithmic form of the refined Cheeger inequality in
\cref{thm:algorithmic-refined-cheeger}.

\begin{theorem}[{Main Theorem, formal version of \cref{thm:informal-main}}]
\label{cor:oracle-constructible-separators}
Let $G$ be a graph with $n$ vertices, $m$ edges, and maximum degree $\Delta$.
Assume that there is a sequence of $r\le n/2$ vertex merge operations
that turns $G$ into a graph $H$.
Then there exist polynomial-time randomized algorithms
(succeeding w.h.p.)
producing $\frac23$-balanced vertex separators of $G$
with several kinds of per-class guarantees.%
\footnote{Note that, like some prior works, we are slightly abusing
notation in the case of $\log a$, which can be interpreted more literally to mean $\max\{\log a,1\}$
so that the bounds are not negative.
Similarly, the bounds also assume $r\geq1$.}

\begin{enumerate}[leftmargin=*]
\item
\label{thmalg:1}
There is an algorithm running in
$\OO(n^2m\cdot{\rm polylog}(n))$ time
producing the following separator sizes:
\begin{itemize}[leftmargin=*]
\item If $H$ is planar, the separator has size
$$
\OO\left(\sqrt{rn\log\Delta}\right).
$$
\item If $H$ has genus $g$, the separator has size
\[
\OO\left(\sqrt{rgn\log\Delta\min\{(\log g)^2,\log\Delta\}}\right).
\]
\item If $H$ is $K_h$-minor-free, the separator has size
\[
\OO\left((h\log h\log\log h)\sqrt{rn\log\Delta}\right).
\]
\item If $H$ is either a $d$-dimensional $k$-ply ball-intersection graph or a
$k$-nearest-neighbour graph, the separator has size
\[
\OO\left(\sqrt{r\log\Delta}k^{1/d}n^{1-1/d}\right).
\]
\end{itemize}

\item\label{thmalg:2} There is an
$\OO(n^4\cdot{\rm polylog}(n))$-time randomized algorithm with the following
stronger guarantees, in addition to the guarantees in \cref{thmalg:1}:
\begin{itemize}[leftmargin=*]
\item If $H$ is planar, the separator has size
\[
\OO\left(\min\left\{\sqrt{\log\Delta},\log r\right\}
  \sqrt{rn}\right).
\]
\item If $H$ has genus $g$, the separator has size
\[
\OO\left(\min\left\{\sqrt{\log\Delta},\log g+\log r\right\}
  \sqrt{rgn\min\{(\log g)^2,\log\Delta\}}\right).
\]
\item If $H$ is $K_h$-minor-free, the separator has size
\[
\OO\left(\min\left\{\sqrt{\log\Delta},\log h+\log r\right\}
  (h\log h\log\log h)\sqrt{rn}\right).
\]
\end{itemize}

\item
\label{thmalg:3}
If there is a fixed value $a<1$ so that
every induced subgraph of $H$ with $n'$ vertices admits an
unknown $\frac23$-balanced vertex separator of size $\OO((n')^a)$,
then there is an algorithm running in
$\OO(nm^2\cdot{\rm polylog}(n))$ time
producing a separator of size
$$
\OO\left(\sqrt{r\log\Delta}n^{(1+a)/2}\right).
$$
\item
\label{thmalg:4}
Under the same separator premise as in \cref{thmalg:3},
assume further that every induced subgraph of $H$
has extremal padding parameter at most $\alpha_0$.
Then there is an algorithm running in
$\OO((nm^2+n^4)\cdot{\rm polylog}(n))$ time
producing a separator of the smaller size
$$
\OO\left(\min\left\{\sqrt{\log\Delta},\alpha_0+\log r\right\}
  \sqrt r n^{(1+a)/2}\right).
$$
\end{enumerate}
\end{theorem}

\noindent We note that our practical implementation specifically corresponds to
\cref{thmalg:3} above.

\subsection{Additional Results}
\label{subsec:additional-results-rough}

We also leverage our methods to obtain some additional results of independent interest
for genus-$g$ graphs and geometric graphs.

Our bounds on $\gamma^{(n)}(G)$
also imply a new bound on the so-called ``Fiedler value''.
For a graph $G$ with adjacency matrix $A$ and degree vector $d$,
the Laplacian is $L(G)={\rm diag}(d)-A$.
Its eigenvalues are ordered as
$0=\lambda_1(G)\leq\lambda_2(G)\leq\cdots\leq\lambda_n(G)$,
and $\lambda_2(G)$ is the \defn{Fiedler value} of $G$,
also called its \defn{algebraic connectivity}~\cite{fiedler1973algebraic}.
Bounds on $\lambda_2(G)$ for genus-$g$ graphs have been the subject of
much prior work~\cite{kelner2004spectral,kelner2006spectral,kelner2006new,
biswal2008eigenvalue,biswal2010eigenvalue,
kelner2009higher,kelner2011metric,
lee2010genus,
amini2018transfer}.

Roch gives a dual form of $\gamma^{(n)}$ that
can be used to show that $\lambda_2(G)\leq\Delta\gamma^{(n)}(G)$~\cite{roch2005bounding},
so \cref{thm:all-base-gamma-bounds} implies a bound for genus-$g$ graphs.

\begin{corollary}
\label{cor:fiedler-genus-bound}
If $G$ is an $n$-vertex genus-$g$ graph with maximum degree $\Delta$, then
\[
\lambda_2(G)\lesssim\frac{\Delta g\cdot\min\left\{\log\Delta,(\log g)^2\right\}}{n}.
\]
\end{corollary}

Compared to earlier work,
this is an improved bound on $\lambda_2(G)$.
See
\cref{tab:lambda2-bounds}
for a history.

\begin{table}[hb]
\centering
\normalsize
\setlength{\tabcolsep}{4pt}
\renewcommand{\arraystretch}{1.25}
\begin{tabular}{c l l l}
\hline
\textbf{Year} & \textbf{Reference} & \textbf{Bound on $\lambda_2(G)$} & \textbf{Scope}\\
\hline
2004 & \cite{kelner2006spectral}
& $\mathrm{poly}(\Delta)g/n$
& triangulated genus-$g$\footnotemark
\\
\hdashline
2008 & \cite{biswal2008eigenvalue,biswal2010eigenvalue}
& $\Delta g^3/n$
& genus-$g$\\
\hdashline
2010 & \cite{lee2010genus}
& $\Delta g(\log g)^2/n$
& genus-$g$\\
\hdashline
2018 & \cite{amini2018transfer}
& $\Delta^2g/n$
& genus-$g$\\
\hdashline
\multicolumn{2}{c}{\textbf{This paper}}
& $\Delta g\min\{(\log g)^2,\log\Delta\}/n$
& genus-$g$\\
\hline
\end{tabular}
\caption{History of upper bounds on the Fiedler value for genus-$g$ graphs.}
\label{tab:lambda2-bounds}
\end{table}
\footnotetext{Kelner originally claimed that his bound applied to general
genus-$g$ graphs,
but the presented argument is incorrect.
We discuss this briefly in the next section,
and in more detail in \cref{subsec:genus-graphs-geometric-bound}.}

The paper also gives direct separator algorithms for represented geometric
inputs.
For $d$-dimensional ball-intersection graphs of ply $k$
and $d$-dimensional $k$-nearest-neighbour graphs,
Miller, Teng, Thurston, and Vavasis proved the existence of
$\OO(d)$-balanced vertex separators of size
$\OO(k^{1/d}n^{1-1/d})$~\cite{miller1997separators},
which in turn imply the existence of
$2/3$-balanced vertex separators of size
$\OO(dk^{1/d}n^{1-1/d})$.
We significantly improve on this loss.

\begin{corollary}
\label{cor:new-neighbourhood-separator-algorithm}
Let $G$ be the intersection graph of $d$-dimensional balls with ply $k$ and
maximum degree $\Delta$,
provided by the center and radius of each ball.
Then $G$ admits a $\frac23$-balanced vertex separator of size
\[
\OO(\sqrt{\min\{d,\log\Delta\}}k^{1/d}n^{1-1/d}).
\]
\end{corollary}

\begin{corollary}
\label{cor:new-knn-separator-algorithm}
Let $G$ be a $d$-dimensional $k$-nearest-neighbour graph with maximum degree
$\Delta$,
provided by the vertex coordinates and the value $k$.
Then $G$ admits a $\frac23$-balanced vertex separator of size
\[
\OO(\sqrt{\min\{d,\log\Delta\}}k^{1/d}n^{1-1/d}).
\]
\end{corollary}

\section{Technical Ideas}
\label{sec:technical-overview-rough}

\subsection{Bounds}
\label{subsec:overview-cheeger}

We use two different mechanisms for proving bounds on $\gamma^{(n)}(G)$,
and more generally $\gamma^{(d)}(G)$.

The first is geometric: A feasible solution of $\gamma^{(d)}(G)$
is interpreted as a family of balls,
and the objective is bounded by a surface-area calculation after mapping those
balls to sphere caps.
This is an adaptation of a framework used by Spielman and Teng for planar graphs
and geometric graphs~\cite{spielman2007spectral}, as well as Kelner for genus-$g$ graphs~\cite{kelner2006new}.
This method is quite involved, and requires a careful handling of the geometry in many cases.
It has also proved quite tricky: We discovered notable errors in the proofs of both
papers that we had to fix in order to adapt their methods.
Our method to fix Kelner's proof in particular required us to devise completely new methods
that resulted in a significantly improved bound.

The second method more directly leverages combinatorial bounds and metric embeddings:
We design an adapter based on ``non-expansive embeddings'' and padded decompositions
that allows us to leverage part of a pipeline for bounds on $\lambda_2(G)$
developed by Biswal, Lee, and Rao~\cite{biswal2008eigenvalue,biswal2010eigenvalue}.

The improved Cheeger-style inequality in \cref{thm:refined-cheeger}
is what converts our solution to $\gamma^{(n)}$ into a partition of the graph.
Similar to the most recent improvement of this kind of Cheeger-style inequality,
the improvement lies in a new dimension-reduction method to relate $\gamma^{(n)}$ and $\gamma^{(1)}$.
In our case, we apply padded decompositions over vertex-weighted shortest path metrics
to create a ``partially'' non-expansive embedding into the line.
This partial property is critical to our correctness.

\subsubsection{Geometric Bounds For Planar and Geometric Graphs}
\label{subsec:overview-gamma-bounds}
The geometric methods primarily leverage one key fact:
For a $d$-dimensional unit ball, we know its area exactly
with an analytic formula.

In order to leverage this,
we first show that $\gamma^{(d)}(G)$
is equivalent to a problem that
embeds $G$ into a spanned $d$-dimensional ball-intersection supergraph,
at the loss of a constant factor.
Then, we restrict ourselves to the special case where these balls
have their centres on the surface of a unit $(d-1)$-sphere embedded in $\RR^d$.
Ultimately, for each class, we construct a feasible solution to $\gamma^{(d)}$ in
the form of a circle packing
of some spanned supergraph of $G$
with some kind of bound on its ply.

Spielman and Teng used a very similar construction
to bound $\lambda_2(G)$ for planar graphs and geometric graphs.
However, they had an error for geometric graphs in dimension $d>2$.
Their error was with respect to
the unintuitive behaviour of higher-dimensional geometry.
Fortunately, we are able to derive a replacement step
that also recovers their bound.

We will show the following bounds with these methods:

\begin{proposition}
\label{prop:planar-lambda-2-star}
Let $G$ be a planar graph with $n$ vertices.
Then $\gamma^{(1)}(G)\lesssim\gamma^{(3)}(G)\le 8/n$.
\end{proposition}

\begin{proposition}
\label{prop:lambda-2-star-neighbourhood-system}
Let $G$ be a spanning subgraph of the intersection graph of $d$-dimensional balls with ply $k$.
Then $\gamma^{(d+1)}(G)\lesssim(k/n)^{2/d}$.
\end{proposition}

\begin{corollary}
\label{cor:knn-lambda-2-star-bound}
Let $G$ be a spanning subgraph of a $d$-dimensional $k$-nearest-neighbour graph with $n$ vertices.
Then $\gamma^{(d+1)}(G)\lesssim(k/n)^{2/d}$.
\end{corollary}

For planar graphs, the required ball representation is a circle packing on the
sphere, so the ply is $1$.
For $k$-ply ball graphs, a homeomorphism used by
Spielman and Teng
\cite{spielman1996spectral,spielman2007spectral}
maps Euclidean balls to sphere caps and preserves ply.
This is where the dimension increases by one.
At this point, in either case,
an inequality is needed relating the area of the sphere caps,
and the radii of the balls they represent.
There is a natural inequality for the $2$-sphere embedded in $\RR^3$,
but it does not extend to higher dimensions,
which is exactly the error in Spielman and Teng's work.
We replace this step with \cref{lemma:area-ineq},
which is enough for the $k$-ply surface-area charging argument.

\subsubsection{Bounded-Genus Graphs}
\label{subsec:overview-genus-reductions}

The techniques we perform for genus-$g$ graphs
are much more involved.
First, we develop a theory of \defn{uniform shallow minors}:
A graph $G$ is a uniform shallow minor of
another graph $H$ if, loosely speaking,
$H$ \emph{expands} every vertex of $G$
into a uniform number of new connected vertices.
See
\cref{fig:overview-uniform-shallow-minors}
for an example.

\begin{figure}[ht]
\centering
\begin{tabular}{@{}cc@{}}
\includegraphics[
page=1,
width=0.22\textwidth
]{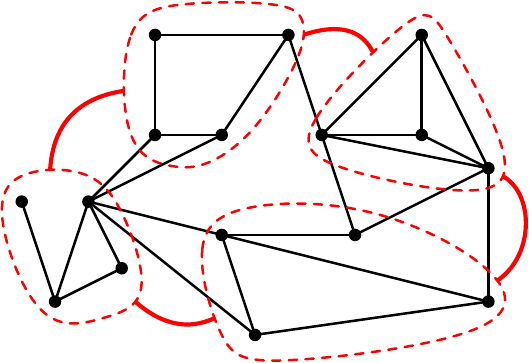}
&
\includegraphics[
page=2,
width=0.22\textwidth
]{graphics/uniform-shallow-minors.pdf}
\end{tabular}
\caption{Two examples of uniform shallow minors:
$C_4$ with depth $2$ on the left,
and $P_4$ with depth $1$ on the right.}
\label{fig:overview-uniform-shallow-minors}
\end{figure}

For those who work with minors, this definition may appear backwards,
but it is the right way around for our purposes:
In our constructions, we start with a graph $G$,
and then construct a \emph{larger} graph $H$ for which $G$ is a uniform shallow minor.
We then show how to relate the values of $\gamma^{(1)}(G)$ and $\gamma^{(1)}(H)$,
along with the depth (or ``shallowness'') of the minor.
In particular, to do this, we randomly \emph{embed} $G$ into $H$.
The vertices of $G$ get mapped to a subset of vertices of $H$,
and the edges of $G$ get mapped to paths in $H$.
We then show that the needed relationship between the values of
the corresponding feasible solutions for
$\gamma^{(1)}$
holds with positive probability.

We apply this technique with carefully constructed larger graphs in order to show both of the following results.

\begin{lemma}
\label{lemma:genus-degree-reduction}
Let $G$ be a graph with genus $g$, $n$ vertices, and maximum degree $\Delta$.
Then there exists a graph $H$ with $n\Delta$ vertices, maximum degree $4$, and
genus $g$, such that
\[
\gamma^{(1)}(G)\lesssim\gamma^{(1)}(H)\Delta\log\Delta.
\]
\end{lemma}

\begin{lemma}
\label{lemma:genus-triangulated-reduction}
Let $G$ be a graph with genus $g$, $n$ vertices, and maximum degree $\Delta$.
Then there is a triangulated genus-$g$ graph $H$ with $(\Delta+1)n$ vertices
and maximum degree $\OO(\Delta)$, such that
\[
\gamma^{(1)}(G)\lesssim\Delta\gamma^{(1)}(H).
\]
\end{lemma}

These constitute our first two steps:
First, we reduce to a surface graph with constant maximum degree.
Second, we reduce to a triangulated surface graph.
The constructions for each are visualized in
\cref{fig:overview-genus-local-reductions}.

\begin{figure}[ht]
\centering
\begin{tabular}{@{}cc@{}}
\includegraphics[
page=3,
width=0.39\textwidth,
keepaspectratio
]{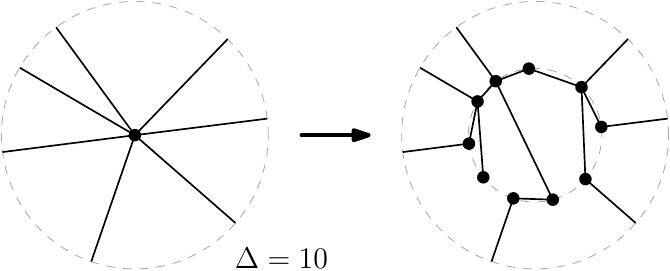}
&
\includegraphics[
page=6,
width=0.22\textwidth,
keepaspectratio
]{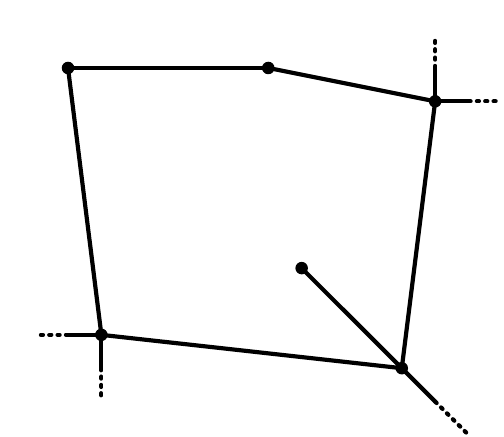}
\hspace{0.25em}
\includegraphics[
page=10,
width=0.22\textwidth,
keepaspectratio
]{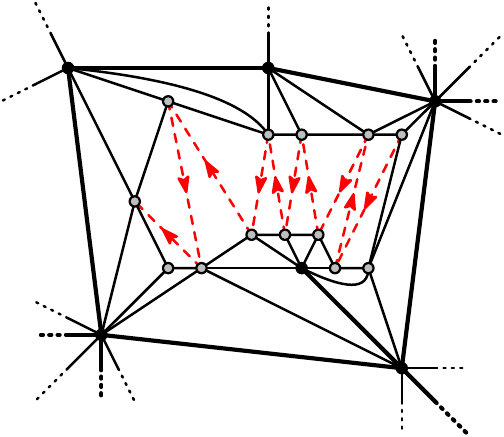}
\end{tabular}
\caption{The two local transformations used for genus-$g$ graphs:
degree reduction near a vertex on the left,
and degree-preserving triangulation of a face on the right.}
\label{fig:overview-genus-local-reductions}
\end{figure}

The degree reduction step replaces each high-degree vertex by a small binary tree
whose leaves respect the rotation system.
This leaves the genus unchanged, and reduces the maximum degree to a constant,
while using depth $\OO(\log\Delta)$ in the uniform shallow minor construction.

The triangulation step
first adds buffer vertices around each face so that
the remaining non-triangular faces can be triangulated independently by the
zig-zag ear-cut construction of Kant and Bodlaender
\cite{kant1992triangulating,kant1997triangulating}.
The buffer vertices are very important:
Without them, it is not possible to even construct a cellular embedding,
much less a triangulation,
since the same edges may appear along the same input face twice
in either opposing or identical orientations.

Our third step starts with a triangulated genus-$g$ graph
of constant maximum degree,
and then
performs repeated ``subdivisions''
until we obtain a sufficiently refined graph
that it can be embedded onto the
unit sphere so that it has ply at most $O(g)$
almost everywhere.

Kelner performed this third step
in order to bound $\lambda_2(G)$,
but incorrectly claimed that a much simpler
construction sufficed for the second step
of triangulation: Just add edges until the graph becomes a triangulated genus-$g$ graph.
Unfortunately, this construction
is not just missing a proof;
it is impossible for genus $g>0$.
Moreover, even if it were possible, it could increase the maximum degree,
which is very problematic for Kelner's bound,
which skips our first degree-reduction step entirely
and consequently depended heavily
on the maximum degree.
We give a concrete counterexample in
\cref{subsec:genus-graphs-geometric-bound}.

Overall, our techniques allow us to obtain the following bound.

\begin{proposition}
\label{prop:genus-geometric-bound}
Let $G$ be a graph with genus at most $g$, $n$ vertices, and maximum degree
$\Delta$.
Then $\gamma^{(1)}(G)\lesssim g\log\Delta/n$.
\end{proposition}

\subsection{Combinatorial Bounds}
\label{subsec:overview-spread}

Our other method builds on the work of Biswal, Lee, and Rao~\cite{biswal2008eigenvalue,biswal2010eigenvalue}.
They work with a quantity called \defn{extremal spread}, denoted $\overline s_2(G)$,
defined as an optimization problem over a certain kind of normalized vertex-weighted shortest-path metrics
in the graph $G$.
They ultimately derive bounds for this optimization problem
via highly combinatorial methods like the crossing number lemma for planar graphs,
hence why we call these bounds ``combinatorial''.

Biswal, Lee, and Rao
gave a method based on embedding
shortest-path metrics into the line
in such a way that they recovered bounds on $\lambda_2(G)$
by leveraging padded decompositions.
We use an analogous method,
although we require more intermediate quantities to
obtain our eventual bounds $\gamma^{(1)}(G)$
(and consequently $\gamma^{(n)}(G)$).

Biswal, Lee, and Rao essentially showed the following bounds on $\overline s_2(G)$:

\begin{proposition}[{\cite[Theorem 3.1]{biswal2010eigenvalue}}]
\label{prop:genus-s2-bound}
For a genus-$g$ graph $G$ with $n\ge 3\sqrt g$ vertices,
\[
\overline s_2(G)\gtrsim n^2/\sqrt g.
\]
\end{proposition}

\begin{proposition}[{\cite{biswal2010eigenvalue} combined with \cite{delcourt2025reducing}}]
\label{prop:minor-free-s2-bound}
For a $K_h$-minor-free graph $G$ on $n\gtrsim h\log\log h$ vertices,
\[
\overline s_2(G)\gtrsim n^2/(h\log\log h).
\]
\end{proposition}

In order to leverage these bounds, we relate
$\gamma^{(1)}(G)$ and $\overline s_2(G)$ as follows.
\begin{proposition}
\label{prop:gamma-s2-relationship}
For a graph $G$ with maximum degree $\Delta$,
\[
\gamma^{(1)}(G)\lesssim
\alpha(G)^2\cdot\frac{n^3}{\overline s_2(G)^2}.
\]
\end{proposition}

\subsection{Cheeger-Style Dimension-Reduction}
\label{subsec:overview-dimred}

Our improved Cheeger-style inequality
relies on a new dimension-reduction step
relating $\gamma^{(n)}$ and $\gamma^{(1)}$.
Specifically,
we will show that $\gamma^{(1)}(G)\lesssim\alpha(G)^2\gamma^{(n)}(G)$.

To do this,
we leverage two things.
First, Biswal, Lee, and Rao give a useful metric embedding result
into the line via padded decompositions that can be applied to shortest-path
metrics on $G$~\cite[Theorem 4.4]{biswal2010eigenvalue}.
The metric embeddings they construct are \defn{non-expansive} (or \defn{Lipschitz}),
meaning no pairwise distance grows,
while the squared distances do not decrease by too much on average.
Second, since the metric of interest defined by a feasible solution to $\gamma^{(n)}(G)$
is not defined by a shortest-path metric over $G$,
but rather by Euclidean embeddings,
it is critical that we do \emph{not} need
the strong non-expansive
property guaranteed by their embedding,
but rather a weaker ``partially non-expansive''
property,
where only distances between adjacent vertices must not grow,
and non-adjacent distances have more freedom.

This construction is the topic of
\cref{sec:refined-cheeger}.

\subsection{Robustness}
\label{subsec:overview-robustness}

For $\gamma^{(n)}(H)$ and $\gamma^{(n)}(G)$,
we use the most natural lift
from the feasible solution of the former to the latter.

For padded decompositions,
we use a linear-programming relaxation
and rounding algorithm of Krauthgamer and Roughgarden~\cite{krauthgamer2011metric},
for which we merge two solutions.
One of the solutions corresponds to $\alpha(H)$, except without the problematic vertices,
while the other captures a limited radius around the nontrivially mapped vertices
in $G$.

\section{A Refined Cheeger Inequality}
\label{sec:refined-cheeger}

In this section, we will prove \cref{thm:refined-cheeger}.
Fix a graph $G=(V,E)$ with maximum degree $\Delta$.
Several of the inequalities stated in \cref{thm:refined-cheeger}
are already known:

\begin{theorem}[{\cite[Theorem 2.10]{olesker2022geometric}}]
\label{thm:cheeger-vertex-one-dim}
For a graph $G$,
$$\psi(G)^2\lesssim\gamma^{(1)}(G)\lesssim\psi(G).$$
\end{theorem}

The result that $\psi(G)^2\lesssim\gamma^{(1)}(G)$
was re-proven
by Kwok, Lau, and Tung~\cite{kwok2022cheeger},
who gave a slightly simpler algorithm
that takes the best of a set of candidate cuts
via a sweep.
Our \cref{alg:sweep}
further simplifies their algorithm
and considers a superset of the cuts that they consider.

\begin{theorem}[{\cite[Proposition 3.14]{kwok2022cheeger}}]
\label{thm:degree-dim-red}
For a graph $G$ with maximum degree $\Delta$,
$$\gamma^{(n)}(G)\lesssim\gamma^{(1)}(G)\lesssim\gamma^{(n)}(G)\cdot\log\Delta.$$
\end{theorem}

The proof that $\gamma^{(1)}(G)\lesssim\gamma^{(n)}(G)\cdot\log\Delta$
uses a random projection and succeeds in expectation,
so \cref{alg:proj-red} replicates this guarantee with high probability.

Thus, the only remaining inequality in \cref{thm:refined-cheeger} left for us to prove is:
$$\gamma^{(1)}(G)\lesssim\gamma^{(n)}(G)\cdot[\alpha(G)]^2.$$
In order to prove the result, we must also do so in a way that
shows
\cref{alg:rp-red} produces a feasible solution for $\gamma^{(1)}(G)$ obtaining this bound
with high probability.

We will do this in several steps.
First, we define the following quantity:

\begin{definition}
\label{def:embedded-spread}
For a graph $G$, we define the $d$-dimensional $L^p$-extremal embedded spread as:
\begin{equation*}
\label{eq:embedded-spread}
\everymath{\displaystyle}
\begin{array}{r l}
    Q^{(d)}_p(G) := &
    \begin{array}{c l}
    \max_{\substack{y:V(G)\to\mathbb{R}_{\geq0},\\f:V(G)\to\mathbb{R}^d}}
    &
    \sum_{u,v\in V(G)} \left|\left|f(u)-f(v)\right|\right|_p^p\\
    \textup{subject to} &
    \begin{array}{r c l l}
    ||y||_1 & \leq & 1\\
    y(u) + y(v) & \geq & \left|\left|f(u)-f(v)\right|\right|_p^p & \forall uv\in E(G).
    \end{array}
\end{array}
\end{array}
\end{equation*}
\end{definition}

For the purposes of this section, $Q^{(1)}_2$ will be a bit easier to work with than $\gamma^{(1)}$,
and the more general definition will also be useful in a later section when we
try to relate $\gamma^{(1)}$ and the $L^2$-extremal spread $\overline s_2$.
$Q^{(d)}_2(G)$ and $\gamma^{(d)}(G)$ are related as follows:

\begin{lemma}
\label{lemma:qd-to-gammad}
For a graph $G$,
$\gamma^{(d)}(G)=\frac{2n}{Q^{(d)}_2(G)}$.
\end{lemma}

\begin{proof}
We will make a sequence of transformations to each optimization problem.

We claim that
for any function $f:V\to\RR^d$ for which $\sum_{v\in V}f(v)=\overline0$,
it holds that
$$2n\sum_{v\in V}||f(v)||_2^2=\sum_{u,v\in V}||f(u)-f(v)||_2^2.$$
To prove this claim, let $f_i(v)$ denote the $i$th coordinate of the vector $f(v)$.
Then,
\begin{equation*}
\everymath{\displaystyle}
\begin{aligned}
    \sum_{u,v\in V}||f(u)-f(v)||^2
    &=\sum_{u,v\in V}\sum_{i\in[d]}(f_i(u)-f_i(v))^2\\
    &=\sum_{u,v\in V}\sum_{i\in[d]}\left[f_i(u)^2+f_i(v)^2-2f_i(u)f_i(v)\right]\\
    &=\sum_{i\in[d]}\sum_{u\in V}\sum_{v\in V}
    \left[f_i(u)^2+f_i(v)^2-2f_i(u)f_i(v)\right]\\
    &=\sum_{i\in[d]}\sum_{u\in V}
    \left[\sum_{v\in V}\left(f_i(u)^2+f_i(v)^2\right)-2f_i(u)\sum_{v\in V}f_i(v)\right]\\
    &=\sum_{i\in[d]}\sum_{u\in V}
    \sum_{v\in V}\left[f_i(u)^2+f_i(v)^2\right]\\
    &=\sum_{i\in[d]}2n\sum_{x\in V}f_i(x)^2\\
    &=2n\sum_{x\in V}||f(x)||^2.
\end{aligned}
\end{equation*}

Therefore,
\begin{equation*}
\everymath{\displaystyle}
\begin{array}{r c c r c l l}
\gamma^{(d)}(G) & =
& 2n\cdot\min_{\substack{f:V\to\RR^d\\y:V\to\RR_{\geq0}}} & \frac{\sum_{v\in V}y(v)}{\sum_{u,v\in V}||f(u)-f(v)||_2^2} \\
&& \textup{subject to} & \sum_{v\in V} f(v) & = & \hspace{1em}\overline{0}\\
&& & y(u)+y(v) & \geq & \hspace{1em}||f(u)-f(v)||_2^2 & \forall uv\in E.
\end{array}
\end{equation*}

No vector $f(u)-f(v)$ is affected by translation,
so the constraint $\sum_{v\in V}f(v)=\overline 0$ is unnecessary,
and thus
the reciprocal of the inner fractional program here is exactly
$Q_2^{(d)}$.
\end{proof}

The key lemma is as follows:

\begin{lemma}
\label{lemma:dim-red-Q-alpha}
For any graph $G$, and any $d,p\geq1$,
$Q_p^{(d)}(G)\leq C_p\cdot\left[\alpha(G)\right]^p\cdot Q_p^{(1)}(G),$
where $C_p$ is a constant dependent only on $p$.
Moreover, given a feasible solution to $Q_p^{(d)}(G)$,
a corresponding feasible solution to $Q_p^{(1)}(G)$ satisfying this inequality
can be found using \cref{alg:rp-red} with high probability.
\end{lemma}

\begin{proof}
Let $f,y$ be the optimal solution to $Q_p^{(d)}$.
Let $\omega(v):=y(v)^{\frac1p}$.
Let $d_\omega$ be the shortest-path metric over $G$
that uses $\omega$ as vertex weights,
so a path $v_1,\dots,v_k$ has weight $\frac12\omega(v_1)+\omega(v_2)+\cdots+\omega(v_{k-1})+\frac12\omega(v_k)$.
Biswal, Lee, and Rao~\cite[Theorem 4.4]{biswal2010eigenvalue} showed that
applying
\cref{alg:rp-red} (or a generalization of this algorithm for arbitrary $p$)
embeds the metric $(V,d_\omega)$ into $\RR$
with a function $f':V\to\RR$
so that:
\begin{itemize}
    \item For every pair of vertices $u$ and $v$, $|f'(u)-f'(v)|\leq d_\omega(u,v)$.
    \item With high probability, $C_p'\cdot\alpha(V,d_\omega)^p\sum_{u,v\in V}|f'(u)-f'(v)|^p\geq\sum_{u,v\in V}d_\omega(u,v)^p$,
        for some positive constant $C_p'$ depending only on $p$.
\end{itemize}
To verify that $f',y$ forms a feasible solution to $Q_p^{(1)}$,
consider an edge $uv\in E$.
For such an edge,
the edge itself gives the shortest path,
so $d_\omega(u,v)=\frac12\left[\omega(u)+\omega(v)\right]$.
Therefore,
$|f'(u)-f'(v)|^p\leq d_\omega(u,v)^p
\leq\left(\frac{\omega(u)+\omega(v)}2\right)^p
\leq\frac{\omega(u)^p+\omega(v)^p}2
\leq y(u)+y(v)$,
where the penultimate inequality follows from Jensen on the convex mapping $x\mapsto x^p$.

We now bound the objective value.
For each $u,v\in V$,
there is some path $u=w_1,w_2,\dots,w_k=v\in V$
with
\begin{align*}
d_\omega(u,v)
&=\frac12\omega(w_1)+\omega(w_2)+\cdots+\frac12\omega(w_k)\\
&=\frac12\left[\left(\omega(w_1)+\omega(w_2)\right)+\left(\omega(w_2)+\omega(w_3)\right)
+\cdots+\left(\omega(w_{k-1})+\omega(w_k)\right)\right].
\end{align*}
For each $i\in\{1,\dots,k-1\}$ we have $w_iw_{i+1}\in E$,
so $\left[y(w_i)+y(w_{i+1})\right]^{\frac1p}\geq||f(w_i)-f(w_{i+1})||_p$.
Moreover, since $x\mapsto x^{\frac1p}$ is subadditive for $x\geq0$,
$$\omega(w_i)+\omega(w_{i+1})
=y(w_i)^{\frac1p}+y(w_{i+1})^{\frac1p}
\geq\left[y(w_i)+y(w_{i+1})\right]^{\frac1p}\geq||f(w_i)-f(w_{i+1})||_p.$$
The triangle inequality then gives
$d_\omega(u,v)\geq\frac12||f(u)-f(v)||_p$, and so
$$C_p'\cdot\alpha(V,d_\omega)^p\sum_{u,v\in V}|f'(u)-f'(v)|^p
\geq\sum_{u,v\in V}d_\omega(u,v)^p
\geq\frac1{2^p}\sum_{u,v\in V}||f(u)-f(v)||^p.
$$
Picking $C_p:=2^pC_p'$ then gives the stated inequality.
\end{proof}

The combination of
\cref{lemma:qd-to-gammad} and \cref{lemma:dim-red-Q-alpha}
therefore proves
\cref{thm:refined-cheeger}
up to the algorithmic aspects,
which are supplied by
\cref{thm:gamma-sdp-solvers}
and an algorithm for padded decompositions briefly explained in
\cref{sec:rsp}.
These are
combined in
\cref{thm:algorithmic-refined-cheeger}.

\section{Bounds via Geometric Intersection Graphs}
\label{sec:bounds-geometric}

In this section, we prove
\cref{prop:planar-lambda-2-star-restatement,%
prop:lambda-2-star-neighbourhood-system,%
cor:knn-lambda-2-star-bound,%
prop:genus-geometric-bound}.
Each of these is a bound on $\gamma^{(d)}$ for a different graph class (and a few different values of $d$).
In particular, each of these bounds uses the relationships these classes have
with a few special kinds of $d$-dimensional ball-intersection graphs,
and how such ball-intersection graphs relate directly to $\gamma^{(d)}$.

We start by defining a modified form of $\gamma^{(d)}$:

\begin{equation*}
\label{eq:overline_gamma_d}
\everymath{\displaystyle}
\begin{array}{r c c r c l l}
\dot\gamma^{(d)}(G) & := & \min_{\substack{f:V\to\RR^d\\s:V\to\RR_{\geq0}}}
& \frac{\sum_{v\in V}\left(s(v)\right)^2}{\sum_{x\in V}||f(x)||_2^2} \\
&& \textup{subject to} & \sum_{v\in V} f(v) & = & \hspace{1em}\overline{0}\\
&& & s(u)+s(v) & \geq & \hspace{1em}||f(u)-f(v)||_2 & \forall uv\in E
\end{array}
\end{equation*}

\begin{lemma}
\label{lemma:gamma-to-dot-gamma}
For any graph $G$, and any value $d\geq1$, $\frac12\gamma^{(d)}\leq\dot\gamma^{(d)}(G)\leq\gamma^{(d)}(G)$.
\end{lemma}

\begin{proof}
For any $a,b\in\RR_{\geq0}$,
Jensen's inequality and superadditivity on $x\mapsto x^2$ ($x\geq0$)
imply that
$\frac12(a+b)^2\leq a^2+b^2\leq(a+b)^2$.
Therefore,
the first inequality in the statement holds if we choose $y(v):=2s(v)^2$,
which satisfies the constraints since for any edge $uv\in E$,
$y(u)+y(v)\geq 2\left[s(u)^2+s(v)^2\right]
\geq\left[s(u)+s(v)\right]^2\geq||f(u)-f(v)||_2^2$.
Similarly,
the second inequality in the statement holds if we choose
$s(v):=\sqrt{y(v)}$,
in which case for any edge $uv\in E$,
$\left[s(u)+s(v)\right]^2\geq\left[\sqrt{y(u)}+\sqrt{y(v)}\right]^2\geq y(u)+y(v)\geq||f(u)-f(v)||_2^2$.
\end{proof}

Importantly, $\dot\gamma^{(d)}(G)$ gives us a new geometric interpretation of the problem:
A solution satisfying the constraints corresponds
to a representation of $G$ as a subgraph of a \emph{geometric intersection graph}:
For each vertex $v$, create a $d$-dimensional ball $B_v$ of radius $s(v)$ centered
at $f(v)$.
Construct a new graph $H=(V,F)$ so that $uv\in F$ if and only if $B_u\cap B_v\neq\emptyset$.
Then $H\supset G$.
As a consequence, if we can bound the sum of squared radii for such a geometric representation,
with the normalization constraints $\sum_{x\in V}||f(x)||_2^2=1$ and $\sum_{x\in V}f(x)=\overline0$,
then we can bound $\gamma^{(d)}(G)$.

The special conditions we need for our ball-intersection graphs all involve the concept of \emph{ply}:
For a set of balls $\mathcal B$ in $\RR^d$,
a point $p\in\RR^d$ is said to have \defn{ply} $k$ if there are exactly $k$ balls in $\mathcal B$
that contain $p$.
Then, the \defn{ply} of $\mathcal B$ itself is said to be the maximum
ply value of any point not contained in the boundary of some ball.
Note that, due to this boundary condition, sets of balls with even ply $1$
can still contain intersections.
For planar graphs and the geometric graph classes,
we will be able to work with bounded ply for the full set of balls.
However, for bounded-genus graphs, we instead will
use a set of balls containing a very large subset that has bounded ply.
See \cref{fig:planar-circle-packing-example} for some examples of balls
with bounded ply, and their corresponding geometric intersection graphs.
A set of $2$-dimensional balls with ply $1$
is also called a (univalent) \defn{circle packing}.

\begin{figure}
    \centering
    \includegraphics[page=1,scale=0.4]{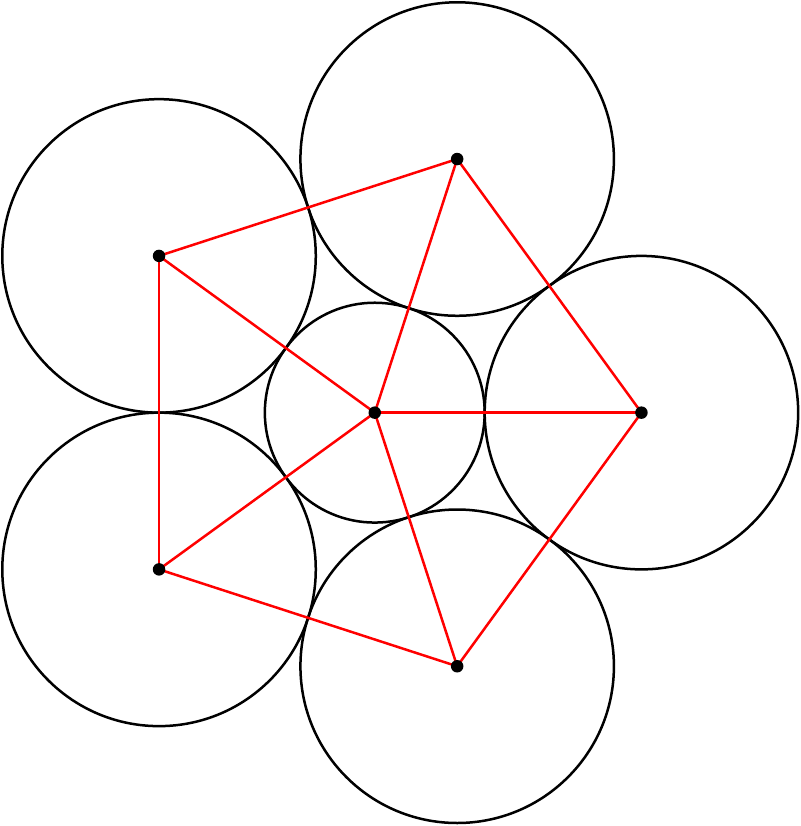}
    \hspace{4em}
    \includegraphics[page=2,scale=0.4]{graphics/touchingcircles-example.pdf}
    \caption{Two example sets of balls, and the corresponding geometric intersection graphs.
        The left example has ply $1$, while the right example has ply $2$.}
    \label{fig:planar-circle-packing-example}
\end{figure}

It will also be helpful for the results in this section to note the following:
\begin{lemma}[{\cite[Proof of Proposition 2.9]{olesker2022geometric}}]
\label{lemma:easy-dimred}
For a graph $G$ and any value $d\geq1$,
$\gamma^{(1)}(G)\lesssim d\gamma^{(d)}(G)$.
Moreover, given a feasible solution to $\gamma^{(d)}(G)$,
a feasible solution to $\gamma^{(1)}(G)$
satisfying this inequality
can be found in polynomial time.
\end{lemma}

\subsection{Planar Graphs}
\label{subsec:geometric-planar-graphs}

Planar graphs admit
a very well-structured representation as a ball-intersection graph.
This comes from the celebrated
circle packing theorem:

\begin{theorem}[Planar Circle Packing Theorem~\cite{koebe1936kontaktprobleme,andreev1970convex,thurston1979geometry}]
\label{thm:planar-circle-packing}
Let $G$ be a simple undirected graph.
Then $G$ is the geometric intersection graph of $2$-dimensional balls with ply $1$ (a circle packing)
if and only if $G$ is a planar graph.
Moreover, if $G$ is a maximal planar graph,
then this representation is unique up to M\"{o}bius transformations.
\end{theorem}

Spielman and Teng~\cite{spielman1996spectral, spielman2007spectral} made use of this theorem to bound $\lambda_2(G)$
for planar graphs $G$, hence giving a ``spectral'' proof of a weaker version of the planar separator theorem
via the usual Cheeger inequality for edge expansion.
We will use a similar argument to bound $\gamma^{(1)}(G)$,
hence also implying a ``spectral''-type of proof of the planar separator theorem in full generality
via
\cref{thm:refined-cheeger}.
We will make use of one key result of Spielman and Teng to accomplish this:

\begin{theorem}[{\cite[Theorem 4.5]{spielman1996spectral, spielman2007spectral}}]
\label{thm:balls-to-surface}
Let $B_1,\dots,B_n$ be a collection of balls in $\RR^d$
with centres $c_1,\dots,c_n$,
so that no point $x\in\RR^d$ is contained in $\left\lceil\frac n2\right\rceil$ of the balls.
Then there is a homeomorphism $\alpha$
from $\RR^d$ to a subset of the sphere $S^d$,%
\footnote{We use the standard topological convention of $S^d$ as the $d$-dimensional surface of the $(d+1)$-dimensional unit ball embeddable in $\RR^{d+1}$.
This differs from the notation of Spielman and Teng, who used
$S^d$ to denote the $(d-1)$-dimensional surface of the $d$-dimensional unit ball.}
so that $\alpha(B_i)$ is exactly a geodesic ball (or sphere cap) in $S^d$
with metric centre $c_i'\in S^d$,
and moreover so that the centroid of the values $c_1',\dots,c_n'$,
in the natural representation of $S^d$ as the surface of the unit ball of $d+1$ dimensions,
is exactly the origin.
\end{theorem}

This theorem statement is slightly weaker than the statement used by Spielman and Teng, who also described the structure of the homeomorphism as a stereographic projection.
However, this statement is sufficient for our purposes.
In particular, it gives us a method of normalizing a geometric intersection graph of balls in $\RR^d$
to be a geometric intersection graph of (geodesic) balls in the sphere $S^d$.
Importantly, since the theorem gives a homeomorphism, it preserves the ply of each individual point.
From this, we can now obtain the desired bound on $\gamma^{(1)}$ for planar graphs:

\begin{proposition}[Extended restatement of \cref{prop:planar-lambda-2-star}]
\label{prop:planar-lambda-2-star-restatement}
Let $G$ be a planar graph with $n$ vertices.
Then,
$$\gamma^{(1)}(G)\lesssim\gamma^{(3)}(G)\leq\frac 8n.$$
Hence, $\psi(G)^2\lesssim\frac1n$.
\end{proposition}

The method to prove this will be analogous to a proof of
Spielman and Teng~\cite[proof of Theorem 3.3]{spielman1996spectral, spielman2007spectral}.

\begin{proof}
If we can show that
$\gamma^{(3)}(G)\leq\frac8n$,
the other inequalities in the statement follow from previous theorem statements.

By \cref{thm:planar-circle-packing} and \cref{thm:balls-to-surface},
there exists a representation of $G$ as a circle packing on the sphere $S^2\subset\RR^3$.
Let the centers be given by $f:V\to S^2\subset\RR^3$
(i.e., $||f(v)||_2^2=1$ for each $v\in V$).
For each $v\in V$, let $s(v)$
be the longest Euclidean distance from $f(v)$
to a point in its geodesic ball on $S^2$, embedded in $\RR^3$.
See \cref{fig:sphere-cap} for an example of these values.
Note that the surface area of the geodesic ball is given exactly by $\pi s(v)^2$.
We claim first that $f,s$ form a feasible solution to $\dot\gamma^{(3)}(G)$:
The constraint $\sum_{v\in V}f(v)=\overline 0$ is exactly
the statement that the centroid of the centres is the origin,
which is given by \cref{thm:balls-to-surface}.
The other constraints follow from the fact that
two balls intersect if and only if they share an edge.
Next, we bound the objective value:
It follows from the statement of \cref{thm:planar-circle-packing}
that the ply of all points in $S^2$ except for a set of measure $0$ is at most $1$.
Hence, the total area of all balls is bounded by the area of $S^2$ itself,
which is $4\pi$.
That is,
$\sum_{v\in V}\pi s(v)^2\leq4\pi$.
Therefore, we get a bound on the objective value of
$$
\frac{\sum_{v\in V}s(v)^2}{\sum_{x\in V}||f(x)||_2^2}
\leq
\frac{4}{n}.
$$
Finally, the result follows from
\cref{lemma:gamma-to-dot-gamma}.
\end{proof}

\begin{figure}[h]
\centering
\begin{tikzpicture}[scale=1]
\def\R{2}

\pgfmathsetmacro\h{0.8}

\pgfmathsetmacro\viewangle{25}

\pgfmathsetmacro\capwidth{sqrt(\R*\R - \h*\h/(cos(\viewangle)*cos(\viewangle)))}
\pgfmathsetmacro\capheight{\capwidth*sin(\viewangle)}

\pgfmathsetmacro\ytangent{\h/(cos(\viewangle)*cos(\viewangle))}
\pgfmathsetmacro\xtangent{sqrt(\R*\R - \ytangent*\ytangent)}
\pgfmathsetmacro\tangentangle{atan2(\ytangent-\h,\xtangent)}

\shade[ball color=gray!10,opacity=0.8] (0,0) circle (\R);

\pgfmathsetmacro\lefttangentcircleangle{atan2(\ytangent,-\xtangent)}
\pgfmathsetmacro\righttangentcircleangle{atan2(\ytangent,\xtangent)}

\fill[red!50,opacity=0.6]
  (0,\h) + ({180-\tangentangle}:{\capwidth} and {\capheight})
  arc ({180-\tangentangle}:{\tangentangle}:{\capwidth} and {\capheight})
  -- ({\xtangent},{\ytangent})  %
  arc ({\righttangentcircleangle}:{\lefttangentcircleangle}:\R)
  -- cycle;

\begin{scope}
\clip (0,0) circle (\R);
\fill[red!50,opacity=0.6] (0,\h) ellipse ({\capwidth} and {\capheight});
\end{scope}

\draw[thick,red] (0,\h) + ({(180-\tangentangle)}:{\capwidth} and {\capheight})
  arc ({180-\tangentangle}:{360+\tangentangle}:{\capwidth} and {\capheight});
\draw[thick,red,dashed,line cap=round] (0,\h) + ({180-\tangentangle}:{\capwidth} and {\capheight})
  arc ({180-\tangentangle}:{\tangentangle}:{\capwidth} and {\capheight});

\coordinate (fv) at (0,\R);
\fill (fv) circle (3pt);
\node[above=2pt] at (fv) {$f(v)$};

\draw[very thick,dashed,dash pattern=on 6pt off 3pt,line cap=round] (fv) -- (\capwidth,\h) node[midway,below,sloped,yshift=1pt] {$s(v)$};

\draw[thick] (0,0) circle (\R);

\end{tikzpicture}
\caption{Illustration of a geodesic ball (sphere cap) on the unit sphere $S^2$ embedded in $\RR^3$,
    along with its corresponding centre $f(v)$. The distance $s(v)$ is the longest Euclidean
    distance from $f(v)$ to any point in the ball.}
\label{fig:sphere-cap}
\end{figure}

\subsection{Ball-Intersection Graphs and Nearest-Neighbour Graphs}
\label{subsec:intersection-graphs}

Spielman and Teng~\cite{spielman1996spectral, spielman2007spectral} also considered higher-dimensional geometric graphs
with similar properties to planar graphs,
and bounded $\lambda_2$ for these classes.
In particular, they obtained bounds for geometric intersection graphs of $d$-dimensional balls with ply $k$,
and $d$-dimensional $k$-nearest-neighbour graphs.
We can use similar techniques to prove
\cref{prop:lambda-2-star-neighbourhood-system,cor:knn-lambda-2-star-bound}.

Unfortunately, their proof also contains a minor error.
Specifically, they used an inequality that does not hold for $d\geq3$ dimensions:
Let $A_d$ denote the area of the unit $d$-sphere $S^d$,
and let $V_d$ denote the volume of the unit $d$-ball $B^d$ (both of these are $d$-dimensional measures).%
\footnote{We again use the topological convention of the $d$-sphere $S^d$ as the $d$-dimensional surface
of the $(d+1)$-dimensional unit ball.}
Let $C$ be a sphere cap (geodesic ball) on the unit $d$-sphere $S^d$ embedded in $\RR^{d+1}$.
Let $f$ be the centre of $C$ on $S^d$.
Let $r$ be the maximum Euclidean distance from $f$ to a point in $C$.
They claimed that $V_d r^d\leq\text{volume}(B)$.
This claim dates back to an earlier work on separators for ball-intersection graphs~\cite[Section 3]{eppstein1993deterministic}.
This is true for $d\leq2$.
In fact, we implicitly make use of this fact in the previous subsection.
However, we will now provide a counterexample for $d=3$:

\begin{example}
For $d=3$ dimensions,
any sphere cap $C$ that covers strictly less
than half the hypersurface of $S^3$
has $3$-dimensional surface area strictly less
than 
the volume of the $3$-dimensional ball $V_3$
of radius $r$,
where $r$ is the maximum Euclidean
distance through $\RR^4$
from the centre of the sphere cap
to any other point inside it.
\end{example}

\begin{proof}
Define $h,r,\phi$
as in the $2$-dimensional cross-section of the $4$-dimensional
unit ball given in \cref{fig:r-h-relationship}.

\begin{figure}[ht]
\centering
\includegraphics{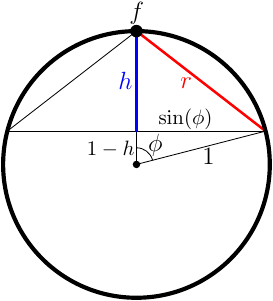}
\caption{Cross-section of a unit ball, and the different associated quantities.}
\label{fig:r-h-relationship}
\end{figure}

The $3$-dimensional area of the $3$-dimensional sphere cap
with colatitude angle $\phi$ (equivalently, height $h$)
is
$\int_0^\phi A_2(\sin\theta)^2d\theta
=\frac{A_2}2(\phi-\cos\phi\sin\phi)$.
Note that $A_2=4\pi$ is the surface area of the unit $2$-sphere,
and that $\cos\phi=1-h$, so $\phi=\arccos(1-h)$.
For the volume,
$r=\sqrt{2h}$,
so $V_3\cdot r^3=V_3\cdot(2h)^{\frac32}$.
Using, say, $h=0.4$,
we can compare the values,
and verify that the $3$-dimensional surface area of the sphere cap
is $<2.81$,
while $V_3\cdot r^3>2.99$.
\end{proof}

A weaker form of their bound still holds, fortunately.
We will prove and use the following inequality:

\begin{lemma}
\label{lemma:area-ineq}
Let $d\geq2$,
and let $C$ be a sphere cap on the $d$-dimensional unit sphere embedded in $\RR^{d+1}$.
Let $f$ be its centre
and
let $r$ be the maximum Euclidean distance in $\RR^{d+1}$
from $f$ to a point in $C$.
Assume $C$ contains at most half the surface area of the sphere.
Then the $d$-dimensional
surface area of the sphere cap
satisfies
$\text{area}(C)\geq\frac{A_dr^d}{4^d}$.
\end{lemma}

We defer the proof of this inequality to
\cref{sec:deferred-area-proof},
since it is a fairly straightforward
expansion of some Gamma and Beta functions.
Using it, we can now prove the main result of this subsection:
\begin{proposition}[Restatement of \cref{prop:lambda-2-star-neighbourhood-system}]
\label{prop:lambda-2-star-neighbourhood-system-full}
Let $G$ be the intersection graph of $d$-dimensional balls with ply $k$
and maximum degree $\Delta$.
Then,
$$
\gamma^{(d+1)}(G)
\lesssim\left(\frac kn\right)^{\frac2d}.
$$
\end{proposition}

\begin{proof}
The proof is very similar to that of
\cref{prop:planar-lambda-2-star},
and again is based on the methods of Spielman and Teng~\cite[proof of Theorem 5.1]{spielman1996spectral, spielman2007spectral}.
By \cref{thm:balls-to-surface},
we may realize $G$ as the geometric intersection graph of a set of geodesic balls on $S^d$,
so that every point in $S^d$ is contained in at most $k$ of these geodesic balls almost surely,
and so that the centroid of the centres of the geodesic balls in $\RR^{d+1}$ is exactly the origin.
Denote the geodesic balls (sphere caps) as $C_1,\dots,C_n$.
For the geodesic ball $C_i$,
let $f_i$ denote its centre
and let $r_i$
denote the maximum Euclidean distance in $\RR^{d+1}$
from $f_i$ to a point in $C_i$.
By
\cref{lemma:area-ineq},
$\text{area}(C_i)\geq\frac{A_d r_i^d}{4^d}$.
The total sum of all such areas is
bounded above by $kA_d$,
so $\sum_{i=1}^n r_i^d\leq k4^d$.
By the power-mean inequality,
we get
$\frac{\sum_{i=1}^n r_i^2}n\leq\left(\frac{\sum_{i=1}^n r_i^d}n\right)^{\frac2d}
\leq\left(\frac{k4^d}n\right)^{\frac2d}\leq16\left(\frac kn\right)^{\frac2d}$.
\end{proof}

As previously mentioned, we will also apply this result to $k$-nearest neighbour graphs, again in a similar manner to
Spielman and Teng~\cite[Corollary 5.2]{spielman1996spectral, spielman2007spectral}.
Let $\tau_d$ denote the \defn{kissing number} in $d$ dimensions,
which is the maximum number of non-overlapping unit balls in $\RR^d$ arranged to all touch a central unit ball.
It is known that as $d\to\infty$, $2^{0.2075d(1+\oo(1))}\leq\tau_d\leq2^{0.401d(1+\oo(1))}$~\cite{wyner65,katabiansky1978bounds}.
Hence, $\log\tau_d\gtrsim d$ and $\tau_d^{\frac1d}\lesssim 1$.
Miller, Teng, Thurston, and Vavasis~\cite{miller1997separators} observed that
every $k$-nearest neighbour graph is the subgraph of an intersection graph of balls with ply $\tau_dk$,
and moreover that every $k$-nearest neighbour graph has maximum degree bounded by $\tau_dk$.
Hence, we obtain the following corollary:
\begin{corollary}[Restatement of \cref{cor:knn-lambda-2-star-bound}]
Let $G$ be a $d$-dimensional $k$-nearest neighbour graph
with $n$ vertices.
Then,
$$\gamma^{(d+1)}(G)\lesssim\left(\frac kn\right)^{\frac2d}.$$
\end{corollary}
Note that the weaker bound on $\frac{|\partial S|}{|S|}$
stated in
\cref{cor:oracle-constructible-separators} for this class
comes from the inequality $\log\Delta\lesssim\log\tau_dk\lesssim d+\log k$.

We will now give specialized algorithms for these graph classes,
since the implied separator results are new.
\begin{corollary}[Restatement of \cref{cor:new-neighbourhood-separator-algorithm}]
Let $G$ be the intersection graph of $d$-dimensional balls with ply $k$ and maximum degree $\Delta$,
provided as a set of $d$-dimensional coordinates and radii for each vertex.
Then, in polynomial time, we can compute a balanced $2/3$-vertex-separator of size
$$\OO\left(\sqrt{\min\{d,\log\Delta\}}\cdot\left(\frac kn\right)^{\frac1d}\right).$$
\end{corollary}

\begin{proof}
Instead of applying \cref{thm:refined-cheeger},
perform the stereographic projection implicit in the proof of
\cref{thm:balls-to-surface},
and then directly apply
\cref{thm:degree-dim-red}
with
\cref{alg:proj-red},
as well as
\cref{lemma:easy-dimred}
(taking the best of the two).
Finally, apply
\cref{thm:cheeger-vertex-one-dim}
with
\cref{alg:sweep}.
\end{proof}

\begin{corollary}[Restatement of \cref{cor:new-knn-separator-algorithm}]
Let $G$ be a $d$-dimensional $k$-nearest neighbour graph with maximum degree $\Delta$,
provided as a set of $d$-dimensional coordinates and a value $k$.
Then, in polynomial time, we can compute a balanced $2/3$-vertex-separator of size
$$\OO\left(\sqrt{\min\{d,\log\Delta\}}\cdot\left(\frac kn\right)^{\frac1d}\right).$$
\end{corollary}

\begin{proof}
We start by computing the graph itself.
By using a simple construction of Kwok, Lau, and Tung~\cite[Proof of Lemma 3.12]{kwok2022cheeger},
we can use the coordinates
to obtain radii for a representation
of the graph as a ball-intersection graph,
while optimizing the objective function up to constant factors
(relative to the fixed choice of coordinates).
Then, instead of applying \cref{thm:refined-cheeger},
perform the stereographic projection implicit in the proof of
\cref{thm:balls-to-surface},
and directly apply
\cref{thm:degree-dim-red}
with
\cref{alg:proj-red},
as well as
\cref{lemma:easy-dimred}
(taking the best of the two).
Finally, apply
\cref{thm:cheeger-vertex-one-dim}
with
\cref{alg:sweep}.
\end{proof}

Another result of Spielman and Teng
for a graph class they call $\alpha$-overlap graphs
also relied on the incorrect inequality for the area of sphere caps~\cite[Theorem 6.4]{spielman2007spectral}.
\Cref{lemma:area-ineq} may also be applied to
fix that result in essentially the same manner.
The result could be generalized to a bound on $\gamma^{(1)}$ as well.

\subsection{Bounded-Genus Graphs}
\label{subsec:genus-graphs-geometric-bound}

In this subsection, we will prove
\cref{prop:genus-geometric-bound},
which states that a genus-$g$ graph $G$ with maximum degree $\Delta$ has $\gamma^{(1)}(G)\lesssim\frac{g\log\Delta}n$.

This will involve several steps.
At a high-level, we will use two different groups of techniques:
\begin{itemize}
    \item Techniques involving \defn{uniform shallow-minors}, which we will define shortly.
    \item Techniques based on the theory of circle packings for genus-$g$ graphs,
        adapting constructions used by Kelner to give a bound on the Fiedler value $\lambda_2$
        for triangulated genus-$g$ graphs~\cite{kelner2004spectral, kelner2006spectral, kelner2006new}.
\end{itemize}

We will now briefly outline some important definitions and challenges related to genus-$g$ graphs.
Recall that a genus-$g$ graph is one that can be drawn on a genus-$g$ (oriented) surface
without crossing edges.
Such a drawing is called a \defn{surface embedding},
and it implies the existence of a \defn{rotation system}:
A (circular) ordering of the edges around each vertex.
Like a plane graph, we can define a \defn{face} of a genus-$g$ graph
as any cycle that is formed by a tour
induced by always moving counter-clockwise in the rotation system at a vertex to choose
the next edge, starting from the previous.
A rotation system (and a corresponding embedding) is called \defn{cellular} if each edge
is a part of exactly two such tours.
Not every surface embedding is cellular, even for plane graphs:
A face on a plane graph
can include an edge twice, each time traversed in opposite directions.
For instance, this is always true of cut edges in any embedding of a planar graph
(see \cref{fig:non-cellular-embeddings}).
In a surface embedding for a genus-$g$ graph with $g\geq1$, it can also be the case that an edge
is included twice in \emph{the same} direction (see \cref{fig:non-cellular-embeddings}).

\begin{figure}[h]
\centering
\includegraphics[page=1,width=0.25\textwidth]{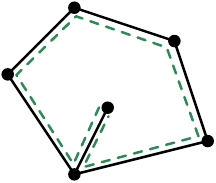}
\hspace{5em}
\includegraphics[page=2,width=0.25\textwidth]{graphics/non-cellular-embeddings.pdf}
\caption{Non-cellular embeddings: planar graph (left) and $K_5$ in the flat torus (right).
A non-cellular face in each is traced.}
\label{fig:non-cellular-embeddings}
\end{figure}

For simple graphs, we call a graph embedded in a surface with cellular faces
a \defn{triangulation} if all of its faces consist of exactly $3$ edges.
Planar triangulations correspond exactly to \defn{maximal} planar graphs,
which are graphs for which the addition of any edge would enforce the invalidation of planarity, regardless of embedding.
For graphs of genus-$g$, this is no longer true:

\begin{example}
$K_5$ is a maximal graph of genus $1$,
and $K_5$ admits no triangulated embedding
into a surface of genus $1$.
\end{example}

See \cref{fig:non-cellular-embeddings} for an embedding of $K_5$ in the torus.

\begin{proof}
The proof is by contradiction.
Suppose this was not the case, and that it had
a triangulated embedding in the torus.
$K_5$ has $5$ vertices and $10$ edges.
By the Euler characteristic, an embedding of $K_5$ in the torus must have $5$ faces.
We will now count the number of edge-face incidences in two different ways.
Since the embedding is triangulated (and hence cellular), each edge is incident to exactly two faces, and there are $10$ edges and hence $2\cdot10=20$ incidences in total.
Since the embedding is triangulated, there are exactly $3$ edges incident to each face, and hence there are $3\cdot5=15$ incidences in total, giving a contradiction.
\end{proof}

Kelner claimed to have bounded $\lambda_2$ for \emph{all} genus-$g$ graphs,
rather than just triangulations.
The method he used to try and prove this
was to reduce the case of general genus-$g$ graphs to the case of triangulated genus-$g$ graphs.
However, the reduction used for this is actually incorrect.
The claimed argument is as follows:
\begin{itemize}
    \item Adding edges to a graph can only increase $\lambda_2$ (this is also true for $\gamma^{(1)}$).
    \item Add edges until the graph is maximal.
    \item Apply the bound on triangulated genus-$g$ graphs that Kelner proves in the rest of his paper.
        This bound relies on a generalized circle packing theorem of He and Schramm~\cite{circlepackinggenus} that applies only to simple triangulations of genus-$g$ surfaces.
\end{itemize}
This argument has two separate issues:
\begin{itemize}
    \item As mentioned, a maximal genus-$g$ graph is not always a triangulation.
    \item Even in a planar graph,
        naively adding edges may increase the maximum degree, possibly up to as large as $\Omega(n)$,
        and the bound that Kelner obtains for triangulations
        has a dependence on the maximum degree.
\end{itemize}
It is worth noting that both of these issues can be sidestepped for planar graphs:
Maximal planar graphs are always planar triangulations,
and
there is a known result that triangulates a planar graph while asymptotically preserving the maximum degree~\cite{kant1992triangulating,kant1997triangulating}.
It is possible that these issues could be amended for genus-$g$ graphs
by allowing augmentation to obtain a multigraph,
but such a construction is not known,
and it would need to be checked that Kelner's construction can be made to work with a multigraph.
Instead, we have both generalized and improved Kelner's bound via our other techniques.

We now move onto the proof.
This will involve several phases.
Henceforth, we will fix a rotation system for our initial graph $G$.
This rotation system need not be cellular, but we will assume that $G$ is simple (as we do throughout this paper).
Importantly for our purposes,
so long as we carefully respect this rotation system,
we
will be able to make certain local transformations
to a graph without changing its genus.

\subsubsection{Uniform Shallow Minors}
\label{subsubsec:uniform-shallow-minors}

For the first phase of our transformations,
we will make use of the following structure:

\begin{definition}
Let $G=(V,E)$ be a graph with $n$ vertices,
and let $H=(V',E')$ be a graph with $r\cdot n$ vertices.
Suppose there is a function $p:V'\to V$
so that $|p^{-1}(v)|=r$ for each $v\in V$,
and $p^{-1}(v)$ is a connected subset of vertices
with diameter $L$.
Suppose furthermore that for each edge $uv\in E'$,
$p(u)p(v)\in E$.
Then we say that $G$ is a \defn{uniform shallow minor} of $H$ with \defn{depth} $L$.
\end{definition}

See \cref{fig:uniform-shallow-minors-examples} for some examples.

\begin{figure}
    \centering
    \includegraphics[page=1,width=0.35\textwidth]{graphics/uniform-shallow-minors.pdf}
    \hspace{4em}
    \includegraphics[page=2,width=0.35\textwidth]{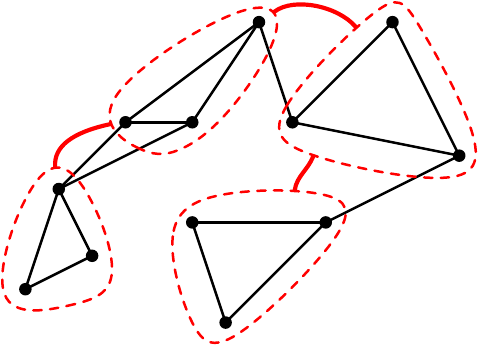}
    \caption{Two examples of uniform shallow minors:
        A graph that admits $C_4$ as a uniform shallow minor with depth $2$ (left),
        and a graph that admits $P_4$ as a uniform shallow minor with depth $1$ (right).}
    \label{fig:uniform-shallow-minors-examples}
\end{figure}

Usually when studying minors,
we start with a graph, and then consider which minors it contains
or does not contain.
Often, we even ``forbid'' minors, resulting in classes like $K_h$-minor-free graphs.
However, we will use uniform shallow minors in the exact opposite sense:
Given a graph $G$, we will aim to find a graph $H$ for which $G$ is a uniform shallow minor.
In particular, their use is characterized by the following technical lemma:

\begin{lemma}
\label{lemma:uniform-shallow-minor-application}
Let $G=(V,E)$ be a graph with $n$ vertices,
and let $H=(V',E')$ be a graph with $|V'|=r\cdot n$ vertices.
Suppose $G$ is a uniform shallow minor of $H$ with depth $L$.
Then $\gamma^{(1)}(G)\lesssim r\cdot L\cdot\gamma^{(1)}(H)$.
\end{lemma}

This lemma will be useful in the following sense:
If we have a graph $G$ of genus $g$,
we may be able to find a graph $H$ of genus $g$
by performing local transformations at each vertex,
so that $G$ is a uniform shallow minor of $H$.
In particular, we may be able to find an $H$ with useful properties,
such as a reduced maximum degree or faces that are easier to triangulate.
Then, a bound on $\gamma^{(1)}(H)$ will imply a bound on $\gamma^{(1)}(G)$,
with ``loss'' $L$.

\begin{figure}
    \centering
    \includegraphics[page=3,width=0.35\textwidth]{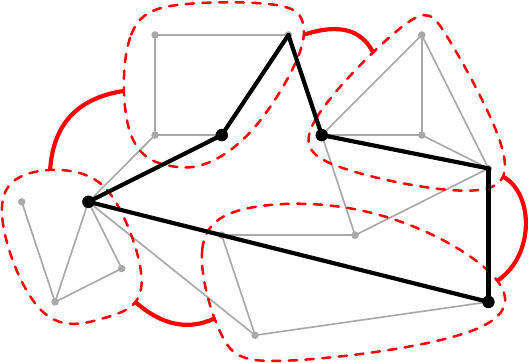}
    \hspace{4em}
    \includegraphics[page=4,width=0.35\textwidth]{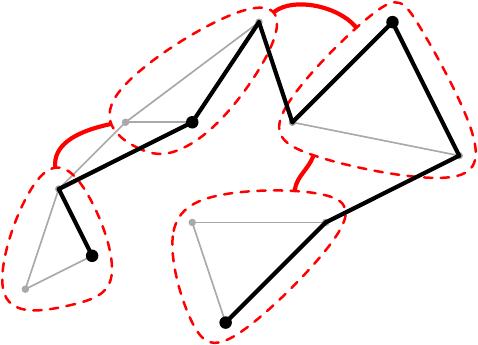}
    \caption{Examples of the construction in \cref{lemma:uniform-shallow-minor-application}
        applied to the uniform shallow minors in \cref{fig:uniform-shallow-minors-examples}.}
    \label{fig:uniform-shallow-minors-construction}
\end{figure}

\begin{proof}
We start with a trivial case whose elimination we will use later:
If $\gamma^{(1)}(H)\cdot r\cdot L\geq1$, then the result holds since $\gamma^{(1)}(G)\lesssim1$:
We can obtain a solution validating this bound by picking two vertices $u,v$ of the graph,
and setting $y(u)=y(v)=2$, $f(u)=1,f(v)=-1$, and $y(w)=f(w)=0$ for any $w\not\in\{u,v\}$.

At a high-level, our construction will randomly sample a \emph{representative} vertex in $H$
for each vertex in $G$,
and then look at the edges along the shortest-paths between these representatives.
This is visualized in \cref{fig:uniform-shallow-minors-construction}.

Let $p:V'\to V$ be the function that certifies the depth-$L$ uniform shallow minor.
Let $f_H,y_H$ be a solution to $\gamma^{(1)}(H)$.
We will construct a solution $f_G,y_G$ to $\gamma^{(1)}(G)$ with random sampling as follows:
Let $\pi(v)$ be a random variable that is uniformly chosen over $p^{-1}(v)$,
so that each $\pi(v)$ is independent.
Let $m(v)$ be samples of all such $\pi(v)$ so that
$$\sum_{u\in V}\sum_{v\in V}|f_H(m(u))-f_H(m(v))|^2
\geq\EX\left[\sum_{u\in V}\sum_{v\in V}|f_H(\pi(u))-f_H(\pi(v))|^2\right].$$
Let $f_G(v):=f_H(m(v))-\frac{1}{|V|}\sum_{x\in V}f_H(m(x))$,
and let
$y_G(v):=2(2L+1)\cdot\sum_{v'\in p^{-1}(v)}y_H(v')$.

We check that the constraints of $\gamma^{(1)}(G)$ are satisfied by $f_G,y_G$.
The normalization constraint is satisfied since $\sum_{v\in V}f_G(v)=\sum_{v\in V}f_H(m(v))-\sum_{x\in V}f_H(m(x))=\overline 0$.
For the edge constraints, let $e=uv\in E$.
There is a path of length $k\leq2L+1$ in $H$ from $m(u)$ to $m(v)$
using only vertices in $p^{-1}(\{u,v\})$.
Denote this path as $m(u)=v_1,\dots,v_{k+1}=m(v)$.
Then,
\begin{center}
\everymath{\displaystyle}
\(
\begin{array}{r c l}
y_G(u)+y_G(v) & = & 2(2L+1)\left[\sum_{u'\in p^{-1}(u)}y_H(u')+\sum_{v'\in p^{-1}(v)}y_H(v')\right]\\[8pt]
              & \geq & 2(2L+1)\left[y_H(v_1)+\cdots+y_H(v_{k+1})\right]\\[8pt]
              & \geq & (2L+1)\sum_{i=1}^{k}\left[y_H(v_i)+y_H(v_{i+1})\right]\\[8pt]
              & \geq & (2L+1)\sum_{i=1}^{k}\left|f_H(v_i)-f_H(v_{i+1})\right|^2\\[8pt]
              & \geq & \frac{2L+1}{k}\left[\sum_{i=1}^{k}\left|f_H(v_i)-f_H(v_{i+1})\right|\right]^2\\[8pt]
              & \geq & \frac{2L+1}{k}\left|f_H(v_1)-f_H(v_{k+1})\right|^2\\[8pt]
              & = & \frac{2L+1}{k}\left|f_G(u)-f_G(v)\right|^2\\[8pt]
              & \geq & \left|f_G(u)-f_G(v)\right|^2,
\end{array}
\)
\end{center}
where
the preantepenultimate step follows from Cauchy-Schwarz,
the antepenultimate step follows from the triangle inequality,
and the last step uses $k\leq2L+1$.
Thus, the edge constraints are satisfied.

It remains to verify the objective.
To obtain a lower bound on $\sum_{v\in V}|f_G(v)|^2$, we first bound the squared distances within each $p^{-1}(v)$.
For any vertex $v\in V$ and any two vertices $u',v'\in p^{-1}(v)$,
there is a path $u'=v_1',\dots,v_{l+1}'=v'$ of length $l\leq L$ contained entirely within $p^{-1}(v)$.
By the triangle inequality, Cauchy-Schwarz, and the constraint from $\gamma^{(1)}(H)$,
\begin{center}
\everymath{\displaystyle}
\(
\begin{array}{r c l}
|f_H(u')-f_H(v')|^2
& \leq & \left[|f_H(v_1')-f_H(v_2')|+\cdots+|f_H(v_l')-f_H(v_{l+1}')|\right]^2\\[4pt]
& \leq & l\left[|f_H(v_1')-f_H(v_2')|^2+\cdots+|f_H(v_l')-f_H(v_{l+1}')|^2\right]\\[4pt]
& \leq & L\sum_{i=1}^l\left[y_H(v_i')+y_H(v_{i+1}')\right].
\end{array}
\)
\end{center}
Summing over all pairs $u',v'\in p^{-1}(v)$ and noting that each $y_H(w')$ for $w'\in p^{-1}(v)$ appears in at most $2|p^{-1}(v)|^2$ terms,
\begin{center}
\everymath{\displaystyle}
\(
\sum_{u'\in p^{-1}(v)}\sum_{v'\in p^{-1}(v)}
|f_H(u')-f_H(v')|^2
\leq
2Lr^2
\sum_{w'\in p^{-1}(v)}y_H(w').
\)
\end{center}

Using this bound, we obtain a lower bound on $\sum_{v\in V}|f_G(v)|^2$:
\begin{center}
\everymath{\displaystyle}
\(
\begin{array}{r c l}
2|V|\cdot\sum_{v\in V}|f_G(v)|^2
& = & \sum_{u,v\in V}|f_G(u)-f_G(v)|^2\\[12pt]
& = & \sum_{u,v\in V}|f_H(m(u))-f_H(m(v))|^2\\[12pt]
& \geq & \EX\left[\sum_{u,v\in V}|f_H(\pi(u))-f_H(\pi(v))|^2\right]\\[12pt]
& = & \frac1{r^2}
\sum_{u\neq v\in V}
\sum_{u'\in p^{-1}(u)}
\sum_{v'\in p^{-1}(v)}
|f_H(u')-f_H(v')|^2\\[12pt]
& = & \frac1{r^2}
\left[
\sum_{u',v'\in V'}
|f_H(u')-f_H(v')|^2
-
\sum_{v\in V}
\sum_{u',v'\in p^{-1}(v)}
|f_H(u')-f_H(v')|^2
\right]\\[12pt]
& \geq &
\frac1{r^2}
\left[
2r|V|
\sum_{v'\in V'}|f_H(v')|^2
-
2Lr^2
\sum_{v'\in V'}y_H(v')
\right]\\[12pt]
& = &
\frac{2\sum_{v'\in V'}|f_H(v')|^2
}{r}
\left[
|V|
-
Lr\gamma^{(1)}(H)
\right]\\[12pt]
& \geq &
\frac{2(|V|-1)\sum_{v'\in V'}|f_H(v')|^2
}{r},
\end{array}
\)
\end{center}
where the antepenultimate step uses $\sum_{u',v'\in V'}|f_H(u')-f_H(v')|^2=2|V'|\sum|f_H(v')|^2$ and the bound on within-set distances,
the penultimate step uses $\sum y_H(v')=\gamma^{(1)}(H)\sum|f_H(v')|^2$,
and the last step uses the assumption $Lr\gamma^{(1)}(H)\leq1$ from the beginning of the proof.
Thus, $\sum_{v\in V}|f_G(v)|^2\gtrsim\frac{\sum_{v'\in V'}|f_H(v')|^2}{r}$.

Finally, we verify the objective:
\begin{center}
\everymath{\displaystyle}
\(
\begin{array}{r c l}
\frac{\sum_{v\in V} y_G(v)}{\sum_{v\in V} |f_G(v)|^2}
& = &
\frac{2(2L+1)\sum_{v'\in V'} y_H(v')}{\sum_{v\in V} |f_G(v)|^2}\\[12pt]
& \lesssim &
\frac{2(2L+1)\sum_{v'\in V'} y_H(v')}{\sum_{v'\in V'} |f_H(v')|^2 / r}\\[12pt]
& = &
2r(2L+1) \cdot \gamma^{(1)}(H)\\[12pt]
& \lesssim &
r \cdot L \cdot \gamma^{(1)}(H).
\end{array}
\)
\end{center}
\end{proof}

We will now apply this result to prove the following lemma:
\begin{lemma}[Restatement of \cref{lemma:genus-degree-reduction}]
\label{lemma:genus-degree-reduction-full}
Let $G$ be a graph with genus $g$, $n$ vertices, and maximum degree $\Delta$.
Then there exists a graph $H$ with $n\Delta$ vertices, maximum degree $4$,
and genus $g$,
so that
$\gamma^{(1)}(G)\lesssim\gamma^{(1)}(H)\cdot\Delta\cdot\log\Delta$.
\end{lemma}

\begin{proof}
We create a set $V'$ that contains $\Delta$ vertices
for each vertex in $V(G)$.
Let $p:V'\to V(G)$ denote the mapping of these vertices.
Likewise, for each edge-vertex pair $(e,v)\in E\times V$ with $e$ incident to $v$,
map $(e,v)$ to some vertex in $p^{-1}(v)$,
in such a manner that no two distinct edge-vertex pairs are mapped to the same vertex in $V'$.
Use this second mapping to construct a set of endpoint-distinct edges $E'$.
For each $v\in V(G)$, construct a near-perfect binary tree $T_v$
over the $\Delta$ vertices in $p^{-1}(v)$
so that some Euler tour of $T_v$ contains the vertices mapped to the rotation system around $v$
as a subsequence
(it suffices to construct a near-perfect binary search tree over an arbitrary ordering of $p^{-1}(v)$ containing this subsequence).
Note that $T_v$ will have
max-degree $3$ and diameter at most $2\log_2\Delta$.
Let $E''$ be the union over all such trees.
Let $H:=(V',E'\cup E'')$.
Note that the maximum degree of $H$ is $4$, since each vertex in $V'$
is incident to at most one edge in $E'$ and at most $3$ edges in $E''$.

We claim that $H$ is of genus exactly $g$:
$G$ is a minor of $H$ (by contracting all edges in $E''$), so clearly the genus of $H$ is at least $g$.
Furthermore, a rotation system of $G$ can be extended to a rotation system of $H$
since the trees forming $E''$ are planar respecting the ordering induced by the rotation system
(see \cref{fig:degree-reduction}),
so the genus of $H$ is also at most $g$.

\begin{figure}[h]
    \centering
    \includegraphics[scale=0.5,page=2]{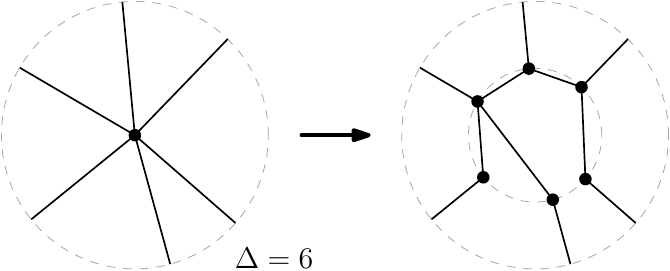}
    \hspace{4em}
    \includegraphics[scale=0.5,page=3]{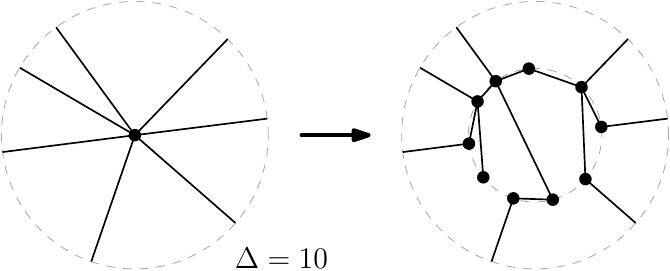}
    \caption{Two examples of the degree-reduction construction for a genus-$g$ graph, respecting a rotation system, local to a vertex. The vertex is replaced with $\Delta$ vertices, at least one vertex for each incident edge, and these new vertices are connected with a near-perfect binary tree. In the left example, the displayed vertex is of maximum degree in its graph, while in the right example it is not (resulting in a degree-$1$ vertex).}
    \label{fig:degree-reduction}
\end{figure}

It remains to show that
$\gamma^{(1)}(G)\lesssim\gamma^{(1)}(H)\cdot\Delta\cdot\log\Delta$.
In fact, $G$ is exactly a uniform shallow minor of $H$ with depth $2\log_2\Delta$,
so we simply apply
\cref{lemma:uniform-shallow-minor-application}.
\end{proof}

We will be able to leverage this lemma in order to turn a bound on
$\gamma^{(1)}$
for bounded-genus graphs
of \emph{constant} maximum degree into a bound on
$\gamma^{(1)}$
for bounded-genus graphs of \emph{arbitrary} maximum degree
at only a $\log\Delta$ loss.
We will now use a similar technique to further reduce to the case of
triangulated bounded-genus graphs with constant maximum degree:

\begin{lemma}[Restatement of \cref{lemma:genus-triangulated-reduction}]
\label{lemma:genus-triangulated-reduction-full}
Let $G$ be a graph of genus $g$ with $n$ vertices and maximum degree $\Delta$.
Then there is a triangulated genus-$g$ graph $H$ with $(\Delta+1)\cdot n$ vertices and maximum degree $\OO(\Delta)$, so that
$\gamma^{(1)}(G)\lesssim(\Delta+1)\cdot\gamma^{(1)}(H)$.
\end{lemma}

\begin{proof}
We will assume for simplicity that $G$ is connected.
$O(1)$ edges can be added to each vertex to make this the case if it is not true.

\begin{figure}
    \centering
    \includegraphics[scale=0.53,page=6]{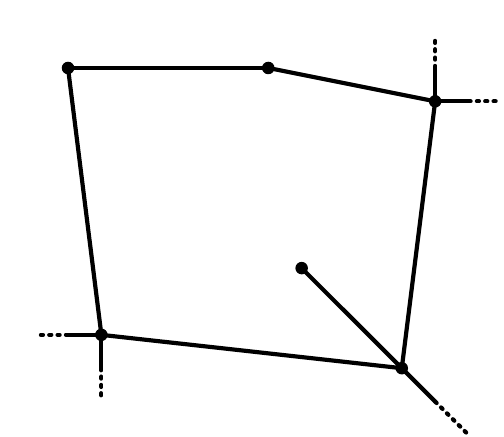}
    \hfill
    \includegraphics[scale=0.53,page=7]{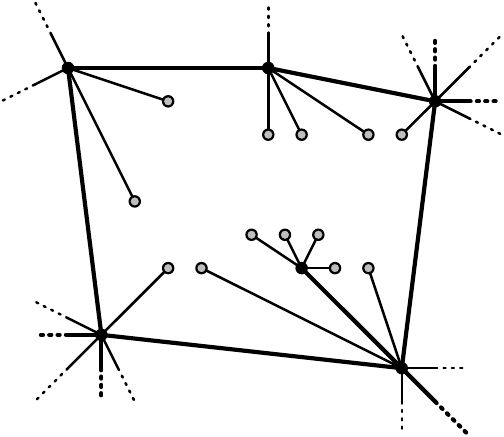}
    \hfill
    \includegraphics[scale=0.53,page=8]{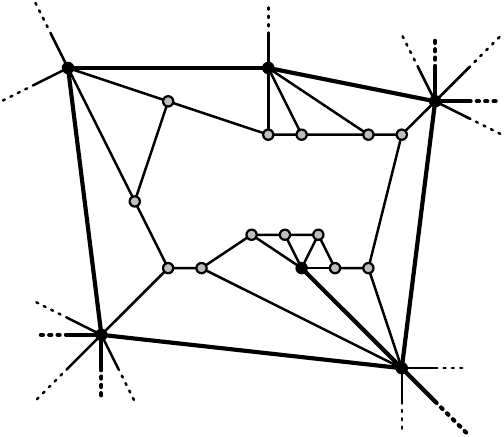}\\
    \includegraphics[scale=0.53,page=9]{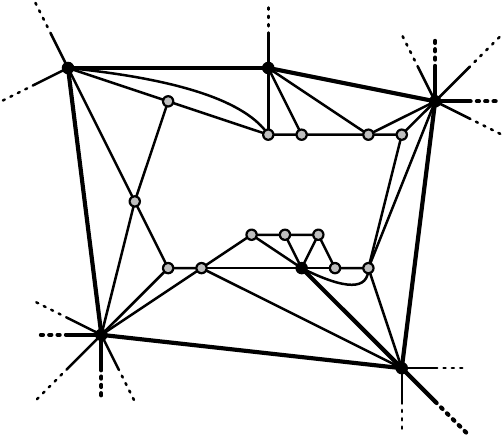}
    \hspace{3em}
    \includegraphics[scale=0.53,page=10]{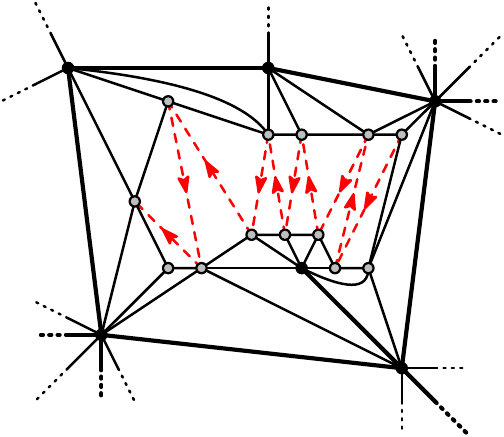}
    \caption{A demonstration of the triangulation steps on a face.
        From left to right, top to bottom:
        $G$, $H_1$, $H_2$, $H_3$, $H$.
        The same steps can also be performed on a non-cellular face.
        }
    \label{fig:zig-zag-ear-cuts-steps}
\end{figure}

Let $H_1$ be the graph that adds $\Delta$ vertices around each vertex to $G$,
and ensure that between any two edges in the rotation system around a vertex in $G$,
at least one such edge is added.
Observe that $H_1$ is a genus-$g$ graph with $(\Delta+1)n$ vertices and maximum degree $2\Delta$.
Observe also that $G$ is a uniform shallow minor of $H_1$
of depth $2$,
and hence
\cref{lemma:uniform-shallow-minor-application}
shows that $\gamma^{(1)}(G)\lesssim(\Delta+1)\gamma^{(1)}(H_1)$.

Now, we will add edges to $H_1$.
In particular, we will construct a sequence of graphs
that add edges to $H_1$, but never add vertices,
and hence $\gamma^{(1)}$ will only increase for this sequence of graphs.
Specifically, within each face of $H_1$,
we create a cycle of all the newly added vertices
in order around the face.
The result is a graph we call $H_2$
(with genus $g$, $(\Delta+1)n$ vertices, and maximum degree $\max\{3,2\Delta\}$)
that has no adjacent pair of
faces both with size $>4$.
Moreover, any vertex incident
to a face of size $>4$ (i.e., those in the created cycles)
is incident exactly two other faces, and has degree exactly $3$.
For each face of exactly size $4$ in $H_2$,
triangulate it arbitrarily to get another new graph $H_3$
(with genus $g$, $(\Delta+1)n$ vertices, and maximum degree $\max\{5,4\Delta\}$).
See
\cref{fig:zig-zag-ear-cuts-steps}
for examples of each of these graphs around a particular face.
It's worth noting that for non-cellular embeddings of higher genus graphs, it is possible that a non-simple face may include the same edge twice
in the same ``direction'' during a traversal around the face (as opposed to the example in the figure, where the doubly-included edge is included in opposite directions).
Fortunately, this does not cause issues for the construction described here either.
It is also worth noting that these steps are necessary even if all the faces are simple:
Specifically, we would like to handle the remaining non-triangular faces
independently, and the separation ensures that we are not able to later produce parallel edges
(see~\cite[Figure 6.3(b)]{kant1997triangulating} for a demonstration of the issues of adjacent simple faces).

It remains to triangulate the remaining non-triangular faces of $H_3$,
none of which are adjacent,
and all of which are simple.
This can be accomplished by adding at most $2$ edges per vertex.
In particular,
Kant and Bodlaender~\cite{kant1992triangulating,kant1997triangulating} gave an algorithm called ``zig-zagging''
that adds edges to a single face using a sequence of ear-cuts,
triangulating the face.
This algorithm applies in the case of faces on \emph{any} rotation system,
not just a rotation system on a planar graph.
Call the final graph $H$,
which is a triangulated genus-$g$ graph with maximum degree $\max\{7,4\Delta\}$.
See the last diagram in 
\cref{fig:zig-zag-ear-cuts-steps}
for a demonstration of this algorithm.

Since $H$ has the same vertex set as $H_1$, and a strict superset of its edges,
$\gamma^{(1)}(H)\geq\gamma^{(1)}(H_1)$,
so the lemma follows from the earlier argument.
\end{proof}

\subsubsection{Circle Packings for Bounded-Genus Graphs}
\label{subsubsec:circle-packings-genus}

Because of the lemmas in
\cref{subsubsec:uniform-shallow-minors}
we henceforth need only consider triangulated genus-$g$ graphs of constant degree.
We will now begin generalizing the arguments of Kelner~\cite{kelner2004spectral,kelner2006spectral,kelner2006new}
to obtain bounds on $\gamma^{(1)}(G)$ for this class.

A genus-$g$ graph naturally embeds into a genus-$g$ surface.
Moreover, a genus-$g$ graph in fact has a circle packing on a genus-$g$ surface.
Both the surface and packing is unique up to a certain set of transformations~\cite{circlepackinggenus}.
However, to bound $\gamma^{(3)}$
in a similar manner to 
\cref{prop:planar-lambda-2-star},
we would need a circle packing on the unit sphere $S^2$,
likely of ply $O(g)$.
Unfortunately,
as Kelner~\cite{kelner2004spectral, kelner2006spectral, kelner2006new} pointed out,
such circle packings need-not exist.
Instead, in order to bound $\lambda_2$, Kelner made use of two steps: First, a specific subdivision procedure is applied to the graph,
and it is shown that a bound for the graph obtained by a sequence of subdivision procedures also induces a bound on the initial graph.
Second, by applying this procedure a sufficient number of times, it is shown that a continuous analogue of the circle packing theory can be used to obtain a bound.
We will apply analogues of each of these steps for $\gamma^{(3)}$
(and $\gamma^{(1)}$, which differs only by a constant factor of at most $3$)
instead of $\lambda_2$.

The subdivision procedure used by Kelner is the \defn{hexagonal subdivision} of a triangulation $G$,
which produces a graph $G'$ by bisecting every edge in $E(G)$, and then connecting all three bisection vertices around each resulting face to form a triangle with new edges.
See \cref{fig:triforcing-example} for an example of this operation.

\begin{figure}
    \centering
    \raisebox{4em}{\includegraphics[scale=0.5,page=6]{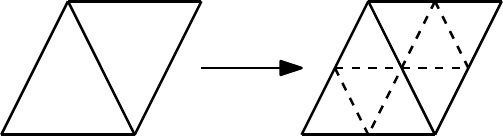}}
    \hspace{2em}
    \raisebox{3em}{\includegraphics[scale=0.25,page=3]{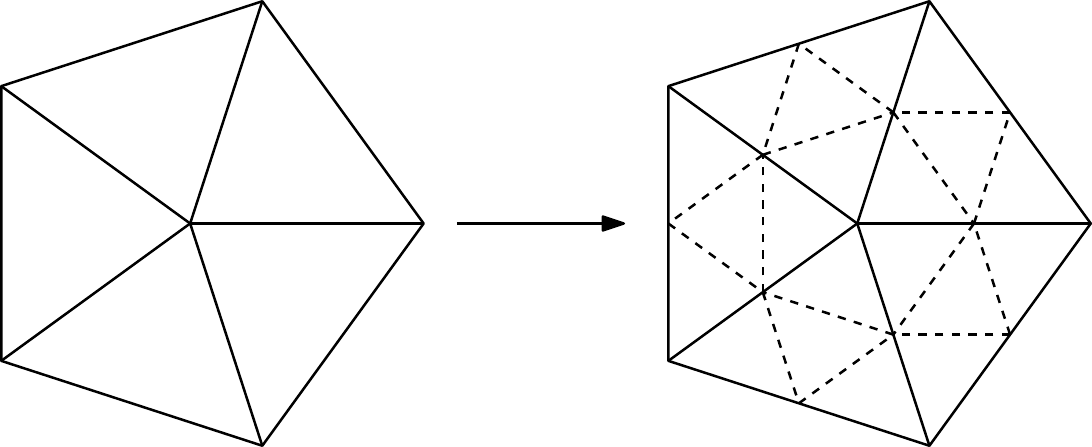}}
    \hspace{2em}
    \includegraphics[scale=0.2,page=9]{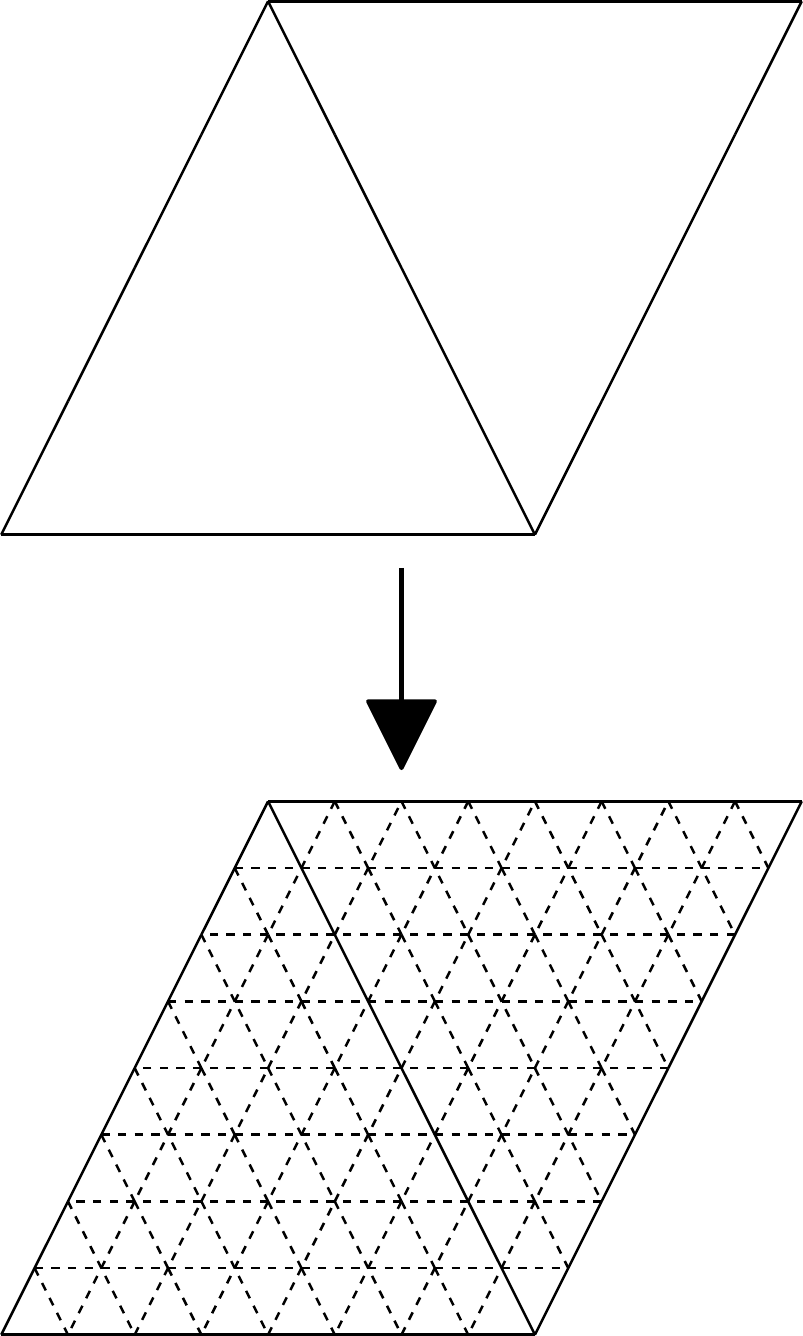}
    \caption{The hexagonal subdivision procedure applied to small example triangulation subgraphs.
    The left and centre drawings show a single application
    of the procedure, while the right drawing shows many.}
    \label{fig:triforcing-example}
\end{figure}

For a triangulation $G$, let $G^{(k)}$ denote the triangulation that results from $k$ successive applications of hexagonal subdivision.
The first step is to show that we can relate $\gamma^{(1)}(G)$ and $\gamma^{(1)}(G^{(k)})$:

\begin{lemma}
\label{lemma:subdivision-bound-propagation}
Let $G$ be a triangulation of genus $g$
with maximum degree $\Delta$.
Then there is some (universal) constant $c$
so that $|V(G)|\cdot\gamma^{(1)}(G)\lesssim \Delta^c\cdot|V(G^{(k)})|\cdot\gamma^{(1)}(G^{(k)})$.
\end{lemma}

We defer the proof of 
\cref{lemma:subdivision-bound-propagation} to \cref{sec:deferred-subdivision-proof}
since it is similar to that of
\cref{lemma:uniform-shallow-minor-application},
and the details are quite technical.
Essentially, it will make use of an argument of Kelner%
~\cite{kelner2006new,kelner2004spectral,kelner2006spectral}
that randomly ``embeds'' $G$ into $G^{(k)}$.
The main distinguishing feature compared to the proof of
\cref{lemma:uniform-shallow-minor-application} is that
this randomized argument includes not just random vertices, but also random paths
between vertices.

The second step is to show that a sufficient number of subdivisions allows us to obtain a bound.
This step will make use of the following lemma
essentially proved by Kelner:

\begin{lemma}[{\cite[Lemma 5.3, Lemma 5.4, proof of Theorem 2.3]{kelner2006spectral}}]
\label{lemma:kelner-extracted}
Let $G$ be a graph of genus $g$.
For each $k$, let $F^{(k)}$ denote a genus-$g$ surface on which $G^{(k)}$ admits a circle packing.
For each vertex $v\in G^{(k)}$, let $C_v$ denote the disk in the packing corresponding to $v$,
and let $r_v$ denote its radius.
Let $A(C)$ and $D(C)$ denote the area (measure) and diameter (longest geodesic path between a point pair),
respectively, of a compact and connected region $C$ either in $F^{(k)}$ or $S^2$.
Then there is a sequence
$f^{(k)}:F^{(k)}\to S^2$
of analytic maps
so that,
for any $\varepsilon>0$,
there exists a threshold $N$
such that for any $k\geq N$
a partition of $S^{(k)}$
into $S^{(k)}_1\cup S^{(k)}_2$
exists with the following properties:
\begin{itemize}
    \item For any vertex $v\in G^{(k)}$
        where $C_v\subset S^{(k)}_1$,
        $$r_v^2\lesssim D\left(f^{(k)}(C_v)\right)^2\lesssim 
        A\left(f^{(k)}(C_v)\right).$$
    \item The projection of $S^{(k)}_1$
        with $f^{(k)}$ has ply $O(g)$.
        That is: For each point $p\in S^2$,
        $$\left|\left(f^{(k)}\right)^{-1}(p)\right|\lesssim g,$$
        almost surely.
    \item
        $\sum_{C_v:C_v\cap S^{(k)}_2\neq\emptyset}D\left(f^{(k)}(C_v)\right)^2
        \lesssim\varepsilon$.
    \item The mapping of the centres of each $C_v$ ($v\in V(G^{(k)})$)
        under $f^{(k)}$,
        for $S^2$ embedded in $\RR^3$ in the standard manner,
        is exactly the origin.
\end{itemize}
\end{lemma}

This is a significant simplification of the sequence of results used by Kelner,
tailored for our purposes.
We note, for the interested reader, that Kelner gives a complete exposition of the required steps
that can be used to prove the above statement
only in the journal version of his work~\cite{kelner2006spectral}
and his thesis~\cite[Chapter 11]{kelner2006new}.
Only a high-level outline is given in the original paper~\cite{kelner2004spectral}.

This lemma gives us an analogous approximate circle packing construction for bounded-genus graphs
after a sufficient number of hexagonal subdivisions.
We will now use it to prove the following lemma:

\begin{lemma}
\label{lemma:subdivision-bound-deep}
Let $G$ be a triangulation of genus $g$.
Then, there is some threshold $N\geq0$
so that for all $k\geq N$,
$\gamma^{(1)}(G^{(k)})\lesssim\gamma^{(3)}(G^{(k)})\lesssim\frac{g}{|V(G^{(k)})|}$.
\end{lemma}

\begin{proof}
Apply \cref{lemma:kelner-extracted}
with $\varepsilon=O(1)$
to get $N$.
Fix any $k\geq N$.
Denote $n^{(k)}:=|V(G^{(k)})|$.
We will now choose the values $y,f$ for the formulation of $\dot\gamma^{(3)}$
to obtain a similar argument
to the proof of
\cref{prop:planar-lambda-2-star}.
For each $v\in V(G^{(k)})$,
choose $s(v):=D(f^{(k)}(C_v))$.
Choose $f(v)$
as the image of the centre of $C_v$
under $f^{(k)}$.
Note that under this
choice, $\sum_{v\in V(G^{(k)})}f^{(k)}(v)$ is exactly the origin.

For each
$uv\in E(G^{(k)})$,
$$||f(v)-f(u)||_2\leq D(f^{(k)}(C_v))+D(f^{(k)}(C_u))=s(u)+s(v).$$
Hence, all the constraints are satisfied.

Note that $\sum_{v\in V(G^{(k)})}|f^{(k)}(v)|^2=n^{(k)}$.
Let $V_1=\{v\in V(G^{(k)}):C_v\subset S_1\}$
and $V_2=\{v\in V(G^{(k)}):C_v\cap S_2\neq\emptyset\}$.
Then,
$$n^{(k)}\cdot\dot\gamma^{(3)}(G^{(k)})
=\sum_{v\in V(G^{(k)})}D(f^{(k)}(C_v))^2
=\sum_{v\in V_1}D(f^{(k)}(C_v))^2
+\sum_{v\in V_2}D(f^{(k)}(C_v))^2$$
$$\lesssim\sum_{v\in V_1}A(f^{(k)}(C_v))
+\varepsilon
\lesssim g\cdot A(S^2)
+\varepsilon
\lesssim g,
$$
so $\dot\gamma^{(3)}(G^{(k)})\lesssim\frac g{n^{(k)}}$.
\end{proof}

\subsubsection{Combining the Lemmas}

The combination of
\cref{lemma:genus-degree-reduction},
\cref{lemma:genus-triangulated-reduction},
\cref{lemma:subdivision-bound-propagation}
and
\cref{lemma:subdivision-bound-deep}
together imply the main result:

\begin{proposition}[Restatement of \cref{prop:genus-geometric-bound}]
\label{prop:genus-geometric-bound-restate}
Let $G$ be a graph with $n$ vertices,
maximum degree $\Delta$, and genus at most $g$.
Then,
$$\gamma^{(1)}(G)\lesssim\frac{g\log\Delta}n.$$
\end{proposition}

\begin{proof}
Let $G'$ be the genus-$g$ graph on $n\Delta$ vertices with maximum degree at most $4=O(1)$
produced by
\cref{lemma:genus-degree-reduction}.
Since the maximum degree of $G'$ is constant,
a triangulated genus-$g$ graph $G''$ with $O(n\Delta)$ vertices and constant maximum degree
can be obtained by
\cref{lemma:genus-triangulated-reduction}
so that
$\gamma^{(1)}(G')\lesssim\gamma^{(1)}(G'')$.
By 
\cref{lemma:subdivision-bound-propagation},
and
\cref{lemma:subdivision-bound-deep}
$\gamma^{(1)}(G'')\lesssim\frac{g}{n\Delta}$,
since the degree of $G''$ is constant.
Hence,
$$\gamma^{(1)}(G)
\lesssim\gamma^{(1)}(G')\cdot\Delta\log\Delta
\leq\gamma^{(1)}(G'')\cdot\Delta\log\Delta
\lesssim\frac{g\log\Delta}n.$$
\end{proof}

\section{Bounds via \texorpdfstring{$1$-dimensional Extremal $L_p$-Embedded Spread}{1-dimensional Extremal Lp-Embedded Spread}}
\label{sec:bounds-metric}

In this section, we will prove
\cref{prop:gamma-s2-relationship}, which we restate here:

\begin{proposition}[Restatement of \cref{prop:gamma-s2-relationship}]
For a graph $G$ with maximum degree $\Delta$,
$$\gamma^{(1)}(G)
\lesssim
\left[\alpha(G)\right]^2
\cdot\frac{n^3}{\left[\overline s_2(G)\right]^2}.$$
\end{proposition}

When combined with
\cref{prop:genus-s2-bound}
and
\cref{prop:minor-free-s2-bound},
this will imply the remaining bounds in
\cref{thm:all-base-gamma-bounds}.

To do this, we will use the definition of
$L^p$-embedded spread, $\overline{Q}_p^{(d)}$,
given in
\cref{def:embedded-spread}.
Recall the lemma relating $\overline{Q}_2^{(d)}$ and $\gamma^{(d)}$:

\begin{lemma}[Restatement of \cref{lemma:qd-to-gammad}]
For a graph $G$,
$\gamma^{(d)}(G)=\frac{2n}{Q^{(d)}_2(G)}$.
\end{lemma}

Given this equivalence,
\cref{prop:gamma-s2-relationship}
is equivalent to the following lemma:
\begin{lemma}
\label{lemma:s2-to-q21}
Let $G$ be a graph,
$$\left[\alpha(G)\right]^2\cdot Q_2^{(1)}(G)\gtrsim\frac{\left[\overline s_2(G)\right]^2}{n^2}.$$
\end{lemma}

To prove this, we will
apply Cauchy-Schwarz directly to $\overline s_2(G)$,
and then apply a metric embedding to the resulting quantity.
The specific metric embedding result we will use was proven by
Biswal, Lee, and Rao~\cite{biswal2010eigenvalue}:
\begin{proposition}[{\cite[Theorem 4.4]{biswal2010eigenvalue}}]
\label{prop:blr-dimred}
For a metric $(X,d)$ with padding parameter $\alpha(X,d)$,
there exists an embedding $f:X\to\RR$ such that:
\begin{itemize}
    \item For every $x,y\in X$, $|f(x)-f(y)|\leq d(x,y)$ (i.e., $f$ is \defn{non-expansive}).
    \item $\alpha(X,d)^2\sum_{x,y\in X}|f(x)-f(y)|^2\gtrsim\sum_{x,y\in X}d(x,y)^2$.
\end{itemize}
\end{proposition}
We previously used a more general algorithmic form of this result in the proof of \cref{lemma:dim-red-Q-alpha}.

The first step is given by the following lemma:

\begin{lemma}
\label{lemma:s2-to-s22}
For a graph $G$,
$$\frac{\left[\overline s_2(G)\right]^2}{n^2}\leq
\sup_{\omega:V(G)\to\RR_{\geq0},||\omega||_2\leq1}\sum_{u,v\in V(G)}\left[d_G^\omega(u,v)\right]^2.$$
\end{lemma}

\begin{proof}
The result follows directly from Cauchy-Schwarz.
\end{proof}

The second step is then given by the following lemma:

\begin{lemma}
\label{lemma:s22-to-q21}
For a graph $G$,
$$
\left[\alpha(G)\right]^2
\cdot Q_2^{(1)}(G)\gtrsim\sup_{\omega:V(G)\to\RR_{\geq0},||\omega||_2\leq1}\sum_{u,v\in V(G)}\left[d_G^\omega(u,v)\right]^2.$$
\end{lemma}

\begin{proof}
Apply \cref{prop:blr-dimred} with $p=2$
to get a non-expansive embedding $f':V\to\RR$
from $(V,d_G^{\omega})$ to $\RR$.
Using $y(v):=\omega(v)^2$
and $f(v):=f'(v)/2$,
we claim the constraints of $Q_2^{(1)}(G)$ are satisfied:
For the normalization constraint, $||y||_1=\sum_{v\in V}\omega(v)^2=||\omega||_2^2\leq1$.
For the edge constraints,
$$y(u)+y(v)=\omega(u)^2+\omega(v)^2\geq\frac12\left[\omega(u)+\omega(v)\right]^2\geq\frac12|f'(u)-f'(v)|^2
\geq|f(u)-f(v)|^2.$$
Thus, $y,f$ forms a feasible solution to $Q_2^{(1)}(G)$.
Finally,
for the objective value,
it follows that
$$\sum_{u,v\in V(G)}\left[d_G^\omega(u,v)\right]^2\lesssim
\left[\alpha(G)\right]^2\cdot\sum_{u,v\in V(G)}|f'(u)-f'(v)|^2\lesssim
\left[\alpha(G)\right]^2\cdot\sum_{u,v\in V(G)}|f(u)-f(v)|^2.$$
\end{proof}

Finally, we can prove
\cref{lemma:s2-to-q21}:
\begin{proof}[Proof of \cref{lemma:s2-to-q21}]
Combine
\cref{lemma:s2-to-s22} and \cref{lemma:s22-to-q21}.
\end{proof}

Consequently, we have completed the proof of
\cref{prop:gamma-s2-relationship},
and therefore we have proved
the remaining bounds stated in
\cref{thm:all-base-gamma-bounds}.

\section{Robustness Proofs}
\label{sec:robustness-proofs}

\begin{proof}[Proof of \cref{lem:gamma-vertex-merges}]
Let $\pi:V(G)\to V(H)$ be the quotient map.
By \cref{lemma:qd-to-gammad},
it is enough to prove that
\[
Q_2^{(d)}(G)\geq
\frac{|V(G)|}{b|V(H)|}Q_2^{(d)}(H).
\]

Let $(f,y)$ be a feasible solution for $Q_2^{(d)}(H)$.
For $x\in V(H)$,
write $r_x=|\pi^{-1}(x)|$.
Pull $(f,y)$ back to $G$ by setting
\[
f'(v)=\frac{1}{\sqrt{b}}f(\pi(v))
\]
and
\[
y'(v)=\frac{1}{b}y(\pi(v)).
\]
Then
\[
||y'||_1
=
\frac{1}{b}\sum_{x\in V(H)}r_x y(x)
\leq
||y||_1
\leq1.
\]
The edge constraints are preserved,
since
if $uv\in E(G)$ and $\pi(u)\neq\pi(v)$,
then $\pi(u)\pi(v)\in E(H)$,
while if $\pi(u)=\pi(v)$, then
the Euclidean distance evaluates to zero.
Thus,
$(f',y')$ is feasible for $Q_2^{(d)}(G)$.

Its objective value is exactly
\[
\frac{1}{b}\sum_{x,z\in V(H)}r_x r_z||f(x)-f(z)||_2^2.
\]
By the weighted variance identity,
and using $r_x\geq1$ and $\sum_x r_x=|V(G)|$,
\[
\sum_{x,z\in V(H)}r_x r_z||f(x)-f(z)||_2^2
\geq
2|V(G)|\min_{a\in\RR^d}\sum_{x\in V(H)}||f(x)-a||_2^2
=
\frac{|V(G)|}{|V(H)|}
\sum_{x,z\in V(H)}||f(x)-f(z)||_2^2.
\]
Therefore
\[
Q_2^{(d)}(G)\geq
\frac{|V(G)|}{b|V(H)|}
\sum_{x,z\in V(H)}||f(x)-f(z)||_2^2.
\]
Finally, this gives us
\[
Q_2^{(d)}(G)\geq
\frac{|V(G)|}{b|V(H)|}Q_2^{(d)}(H).
\]
Using \cref{lemma:qd-to-gammad} again
proves the first claim.
The second claim follows because $q$ merge operations can put at most
$q+1$ original vertices into a single vertex of $H$.
\end{proof}

\begin{proof}[Proof of \cref{lem:alpha-vertex-merges}]
If $q=0$,
then the claim is trivial,
so assume $q\geq1$.

Let $\pi:V(G)\to V(H)$ be the quotient map.
Let $R$ be the set of vertices of $G$
whose image under $\pi$ has more than one preimage.
There are at most $q$ such images in $H$.
The sum over them of one less than their preimage size is $q$.
Hence $|R|\leq2q$.
Outside $R$,
the map $\pi$ identifies $G-R$ with an induced subgraph of $H$.
Thus every component of $G-R$
has padding parameter at most $\alpha_0$.

Fix a vertex-weighted shortest-path metric $d$ on $G$
and a scale $\Delta$.
We use the LP formulation and rounding theorem of Krauthgamer and
Roughgarden~\cite[(3.4), Figure 2, and Theorem 3.9]{krauthgamer2011metric}.
For a metric space $(X,\delta)$,
their LP at scale $\Delta$ and parameter $\beta$
has nonnegative variables $x_{av},y_{av}$ for $a,v\in X$ and constraints
\[
\begin{array}{rcll}
\sum_{a\in X}x_{av} & \leq & 1 & \forall v\in X,\\
y_{av} & \leq & x_{au} &
\forall a,v\in X, u\in B_{\delta}(v,\Delta/\beta),\\
x_{av} & = & 0 &
\forall v\in X, a\notin B_{\delta}(v,\Delta),\\
\sum_{a\in X}y_{av} & \geq & c & \forall v\in X,
\end{array}
\]
where $c>0$ is a sufficiently small absolute constant.
It is enough to make this LP feasible on $(V(G),d)$ with
\[
\beta=C(\alpha_0+\log(q+2))
\]
for a sufficiently large constant $C$.
Write
\[
\rho=\Delta/\beta.
\]

We use two partial LP solutions.
First, on each component of $G-R$,
use the LP solution supplied by the $\alpha_0$ bound,
extend it by zero to the rest of $G$,
and keep its $y$-variables only at vertices farther than $\rho$ from $R$.
These vertices are covered because their $\rho$-balls stay in one component.
The third LP constraint remains true because intrinsic distances in a component
are at least the ambient distances.
The second LP constraint is checked entirely inside that component,
again because the relevant $\rho$-balls avoid $R$.

Second,
choose a nearest point $p(v)\in R$ for every vertex within distance $2\rho$
of $R$.
The metric on $R$ has padding parameter $O(\log(q+2))$
by the standard finite-metric bound~\cite{conroy2025protect}.
Use the corresponding LP solution on $R$ at scale $\Delta/2$.
For $a\in R$ and vertices $u$ within distance $2\rho$ of $R$,
copy the $x$-variable for $a$ and $p(u)$ from the LP on $R$.
For $a\in R$ and vertices $v$ within distance $\rho$ of $R$,
copy the $y$-variable for $a$ and $p(v)$ from the LP on $R$.
All remaining variables in this piece are zero.
If $C$ is large enough,
the third LP constraint holds:
a nonzero $x_{au}$ has
$d(a,u)\leq d(a,p(u))+2\rho\leq\Delta$.
For the second LP constraint,
suppose $y_{av}$ is nonzero and $u\in B_d(v,\rho)$.
Then $d(v,R)\leq\rho$,
so $d(u,R)\leq2\rho$ and $p(u)$ is defined.
Moreover,
\[
d(p(u),p(v))\leq4\rho
\]
is within the ball radius in the second constraint of the LP on $R$,
after increasing $C$ by an absolute factor.
Thus the needed inequality follows from the LP on $R$.

Average the two partial solutions.
The first three LP constraints hold in both pieces,
so they hold after averaging.
For the fourth constraint,
each vertex has total $y$-value at least $c$ in one of the two pieces
and nonnegative total $y$-value in the other.
Thus the average satisfies the fourth constraint with $c/2$ in place of $c$.
After decreasing the absolute constant $c$ by a factor of two,
this is a feasible LP solution on $G$.
The rounding theorem of Krauthgamer and Roughgarden gives a padded
decomposition with the same asymptotic parameter.
Since $d$ and $\Delta$ were arbitrary,
$\alpha(G)\lesssim \alpha_0+\log(q+2)$.
\end{proof}

\section{Algorithmic Details and Implementation}
\label{sec:rsp}

Here we detail our algorithmic form of the refined Cheeger inequality,
then discuss some small experiments.

\subsection{Algorithmic Details}

Many pieces of the rounding algorithm
are implicit in the proofs of Kwok, Lau, and Tung~\cite{kwok2022cheeger,kwok2025cheeger},
although alternative algorithms are present in many other works
for several steps.

\begin{theorem}[Algorithmic form of \cref{thm:refined-cheeger}]
\label{thm:algorithmic-refined-cheeger}
Let $G$ be a graph with $n$ vertices, $m$ edges, and maximum degree $\Delta$.
There is a randomized algorithm (succeeding w.h.p.)
that runs in $\OO((m^2+n^3)\cdot{\rm polylog}(n))$ time
and outputs a set $S$ with $|S|\leq n/2$ and
\[
\psi(S)\lesssim
\sqrt{\gamma^{(n)}(G)\cdot\min\{\log\Delta,\alpha(G)^2\}}.
\]
There is also a randomized algorithm, succeeding with high probability,
that runs in $\OO((nm+n^3)\cdot{\rm polylog}(n))$ time
and outputs a set $T$ with $|T|\leq n/2$ and
\[
\psi(T)\lesssim
\sqrt{(\gamma^{(n)}(G)+O(1/n))\cdot
\min\{\log\Delta,\alpha(G)^2\}}.
\]
\end{theorem}

As in previous work,
there are two sub-steps in the algorithm:
A dimension-reduction step and a threshold rounding step.

\paragraph{Initial SDP solving.}
The input to the rounding step is a feasible and approximately optimal solution to
$\gamma^{(n)}(G)$,
consisting of an embedding $f:V\to\RR^n$ and weights
$y:V\to\RR_{\ge0}$ satisfying
\[
\begin{array}{ll}
y(u)+y(v)\ge ||f(u)-f(v)||_2^2 & \forall uv\in E.
\end{array}
\]

\begin{theorem}[Solving the vertex-expansion SDP]
\label{thm:gamma-sdp-solvers}
Let $G$ be an $n$-vertex, $m$-edge graph.
For every $\varepsilon>0$,
one can compute a feasible solution for $\gamma^{(n)}(G)$
with objective value at most $\gamma^{(n)}(G)+\varepsilon$
in $O(m\cdot{\rm polylog}(n)/\varepsilon)$ time.
One can also compute a feasible solution with objective value
$O(\gamma^{(n)}(G))$
in $O(m^2\cdot{\rm polylog}(n))$ time.
\end{theorem}

The additive error bound above follows by applying the Matrix Multiplicative Weights
Update method~\cite{arora2016combinatorial}
to Roch's formulation
(see also \cite[Section 5.1]{lau2024fast})
and using approximate fractional matchings,
which can be obtained using
a positive LP solver (e.g. Young's~\cite{young2014nearly}).
The multiplicative guarantee follows from the positive SDP solver of
Jambulapati, Lee, Li, Padmanabhan, and
Tian~\cite{jambulapati2021positivesemidefiniteprogrammingmixed}.%
\footnote{The original paper \cite{jambulapati2020positive}
corresponding to \cite{jambulapati2021positivesemidefiniteprogrammingmixed}
claimed a faster time complexity that would imply
$O(m\cdot{\rm polylog}(n))$ time.
This appears to match our empirical results,
but a flaw was later found in their proof.
A tight analysis remains open.}
Our implementation uses
the latter approach.

To round, we perform two dimension-reduction steps,
and take the best result of the two.
One of these dimension-reduction methods is quite slow in practice, so our implementation
only uses the faster one.

\paragraph{Dimension reduction.}
The first dimension-reduction method is a Gaussian projection.

\begin{algorithm}[H]
\caption{ProjectionReduce($f:V\to\RR^n$, $G=(V,E)$)}
\begin{algorithmic}
\label{alg:proj-red}
\For{$i=1$ to $\Theta(\log n)$}
  \State Sample $r_i\sim\mathcal{N}(0,I_n)$
  \State $f'_i(v)\gets\langle f(v),r_i\rangle$ for all $v\in V$
\EndFor
\State\Return $f'_j$ maximizing $\sum_{u,v\in V}[f'_j(u)-f'_j(v)]^2$
\end{algorithmic}
\end{algorithm}

The second dimension-reduction method
uses an average-squared-distortion embedding routine implicit in
Biswal, Lee, and Rao~\cite[Theorem 4.4]{biswal2010eigenvalue}, itself building
on Rabinovich~\cite{Rabinovich03}.
Sampling the random partition uses
a linear programming relaxation of
Krauthgamer and Roughgarden~\cite{krauthgamer2011metric},
which can be solved efficiently
using a positive LP solver~\cite{young2014nearly}.
Note that we do not implement this option.

\begin{algorithm}[H]
\caption{DistanceReduce($y:V\to\RR_{\ge0}$, $G=(V,E)$)}
\begin{algorithmic}
\label{alg:rp-red}
\State Let $d_y(u,v)$ be the shortest-path metric in $G$ with vertex weights $y$
\State $\Delta_2\gets
  \sqrt{\frac1{n^2}\sum_{u,v\in V}d_y(u,v)^2}$
\If{there exists $x_0\in V$ with
  $\left|B_{d_y}(x_0,\Delta_2/4)\right|\ge n/10$}
  \State $S_0\gets B_{d_y}(x_0,\Delta_2/4)$
  \State\Return $f'_0(u)=d_y(u,S_0)$ for all $u\in V$
\EndIf
\For{$i=1$ to $\Theta(\log n)$}
  \State $\mathcal{P}_i\gets$
    \texttt{RandomPartition}$(V,d_y,\Delta_2/4)$
  \State Let $S_i$ be the union of parts $P\in\mathcal{P}_i$ selected
    independently with probability $1/2$
  \State $f'_i(u)\gets d_y(u,S_i)$ for all $u\in V$
\EndFor
\State\Return $f'_j$ maximizing $\sum_{u,v\in V}[f'_j(u)-f'_j(v)]^2$
\end{algorithmic}
\end{algorithm}

The output of the better of the two methods is used.

\begin{algorithm}[H]
\caption{BestReduce($f:V\to\RR^n$, $y:V\to\RR_{\ge0}$, $G$)}
\begin{algorithmic}
\label{alg:best-red}
\State $f'_1\gets$ ProjectionReduce($f$, $G$)
\State $f'_2\gets$ DistanceReduce($y$, $G$)
\State\Return $f'_j$ maximizing $\sum_{u,v\in V}[f'_j(u)-f'_j(v)]^2$
\end{algorithmic}
\end{algorithm}

\paragraph{Threshold rounding.}
Given a one-dimensional embedding, we use a sweep over the vertices in sorted order
to find the best cut
induced by the ordering.

\begin{algorithm}[H]
\caption{Sweep($f':V\to\RR$, $y:V\to\RR_{\ge0}$, $G=(V,E)$)}
\begin{algorithmic}
\label{alg:sweep}
\State Sort vertices by $f'$ value: $v_1,\dots,v_n$
\For{$k=1$ to $n-1$}
  \State $A_k\gets\{v_1,\dots,v_k\}$ and
    $B_k\gets\{v_{k+1},\dots,v_n\}$
  \State $S_k\gets\emptyset$
  \For{each edge $uv\in E$ with one endpoint in $A_k$ and one endpoint in $B_k$}
    \If{$y(u)\ge y(v)$}
      \State $S_k\gets S_k\cup\{u\}$
    \Else
      \State $S_k\gets S_k\cup\{v\}$
    \EndIf
  \EndFor
  \State $\psi_k\gets|S_k|/\min\{|A_k|,|B_k|\}$
\EndFor
\State\Return the partition
  $(A_k\setminus S_k,B_k\setminus S_k,S_k)$ minimizing $\psi_k$
\end{algorithmic}
\end{algorithm}

We can then attain sparse cuts by combining these steps.
\begin{algorithm}[H]
\caption{ReweightedSpectralPartitioning($G=(V,E)$)}
\begin{algorithmic}
\label{alg:rsp}
\State Solve $\gamma^{(n)}(G)$, obtaining $f:V\to\RR^n$ and
  $y:V\to\RR_{\ge0}$
\State $f'\gets$ BestReduce($f$, $y$, $G$)
\State $(A,B,S)\gets$ Sweep($f'$, $y$, $G$)
\State\Return $(A,B,S)$
\end{algorithmic}
\end{algorithm}

One call gives a set of vertices with bounded vertex expansion.
Repeating on the larger residual component gives a balanced separator.

\begin{algorithm}[H]
\caption{BalancedRSP($G=(V,E)$)}
\begin{algorithmic}
\label{alg:rsp-balanced}
\State $S_{\mathrm{union}}\gets\emptyset$
\State $V_{\mathrm{residual}}\gets V$
\While{$|V_{\mathrm{residual}}|>\frac23|V|$}
  \State Let $G_{\mathrm{residual}}$ be the induced subgraph of $G$ on
    $V_{\mathrm{residual}}$
  \State $(A,B,S)\gets$ ReweightedSpectralPartitioning($G_{\mathrm{residual}}$)
  \State $S_{\mathrm{union}}\gets S_{\mathrm{union}}\cup S$
  \State $V_{\mathrm{residual}}\gets\arg\max_{X\in\{A,B\}}|X|$
\EndWhile
\State\Return $S_{\mathrm{union}}$
\end{algorithmic}
\end{algorithm}

\subsection{Experiments}
\label{subsec:implementation-figures}

We ran some very limited experiments with our prototype implementation,
using the DIMACS shortest-path road-network
dataset~\cite{demetrescu2009shortest},
as well as some synthetic graphs.
We hope that a better-tuned implementation
with more expansive experiments will be the subject of future work.
The CPU experiments were run using an AMD Ryzen 9 7950X CPU and 64 GB of RAM.
\Cref{fig:positive-certified-planar} shows the output on a synthetic planar
instance.

\begin{figure}[H]
\centering
\includegraphics[width=0.48\textwidth,height=0.22\textheight,keepaspectratio]{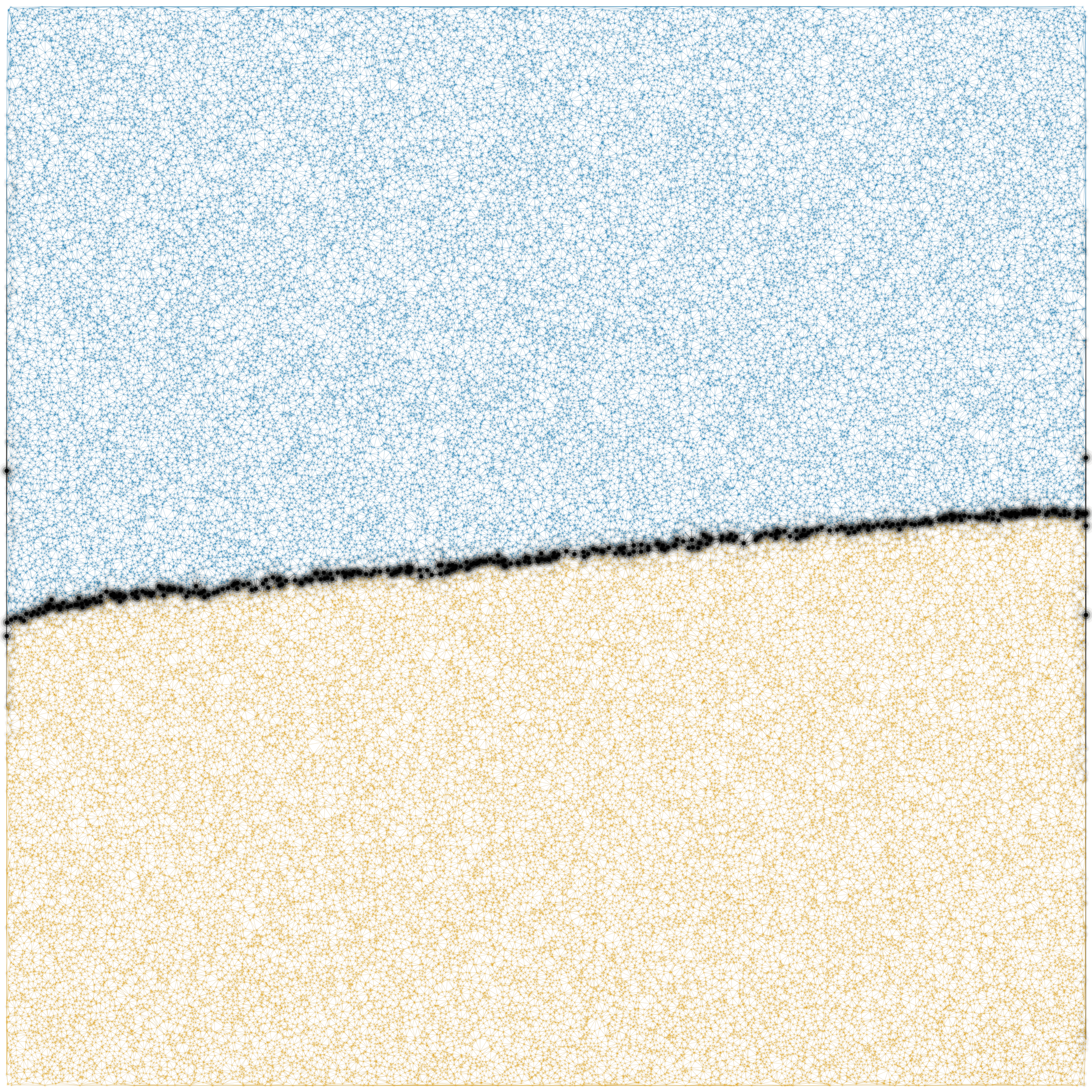}
\caption{Positive-SDP output on a 100,000-vertex synthetic planar instance.}
\label{fig:positive-certified-planar}
\end{figure}

\paragraph{DIMACS road-network instances.}
We also ran the prototype on the DIMACS shortest-path road-network instances
\cite{demetrescu2009shortest}.
Note that this dataset is of directed graphs,
and that we replace it with its undirected form.
The table reports the number of undirected edges after loading and deduplication.
The road-network separators and the corresponding one-dimensional
projections are visualized in
\cref{fig:dimacs-road-separators,fig:dimacs-road-projections}.

\begin{table}[H]
\centering
\scriptsize
\renewcommand{\arraystretch}{1.15}
\resizebox{\textwidth}{!}{%
\begin{tabular}{l r r r r}
\hline
\textbf{Region}
& \textbf{Vertices}
& \textbf{Edges}
& \textbf{Separator vertices}
& \textbf{Runtime (seconds)}\\
\hline
NY
& $264346$
& $365050$
& $16$
& $38.07$\\
BAY
& $321270$
& $397415$
& $18$
& $29.69$\\
COL
& $435666$
& $521200$
& $21$
& $23.51$\\
\hline
\end{tabular}%
}
\caption{DIMACS road-network experiment summary.
Runtime is in seconds.}
\label{tab:dimacs-road-experiments}
\end{table}

\begin{figure}[H]
\centering
\begin{tabular}{@{}ccc@{}}
\includegraphics[height=0.17\textheight,keepaspectratio]{code/python/results/visualizations/road_NY_unconstrained_positive_paper.png}
&
\includegraphics[height=0.17\textheight,keepaspectratio]{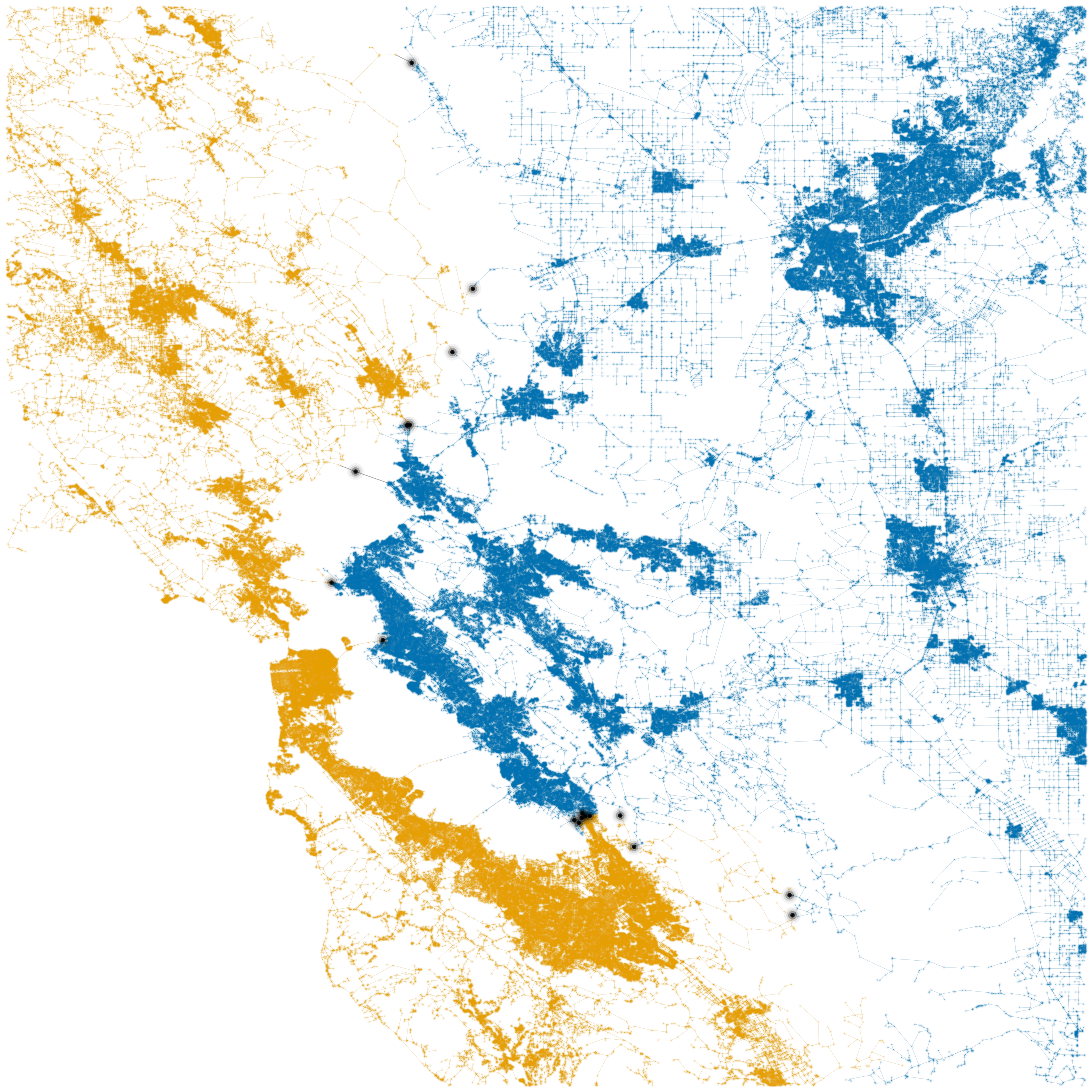}
&
\includegraphics[height=0.17\textheight,keepaspectratio]{code/python/results/visualizations/road_COL_unconstrained_positive_paper.png}
\end{tabular}
\caption{DIMACS road-network separators produced by the positive-SDP prototype.
From left to right: New York, Bay Area, and Colorado.}
\label{fig:dimacs-road-separators}
\end{figure}

\begin{figure}[H]
\centering
\begin{tabular}{@{}ccc@{}}
\includegraphics[height=0.17\textheight,keepaspectratio]{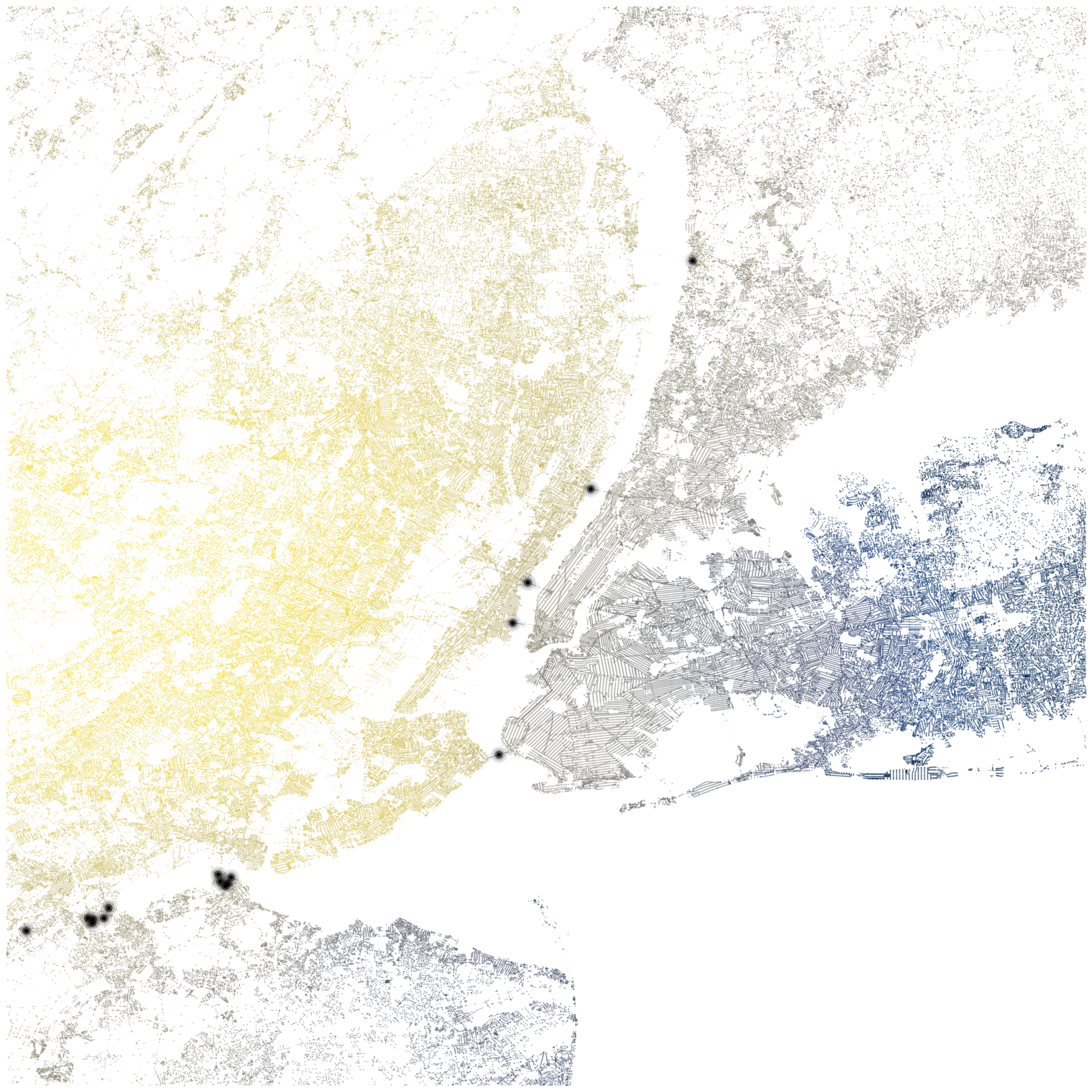}
&
\includegraphics[height=0.17\textheight,keepaspectratio]{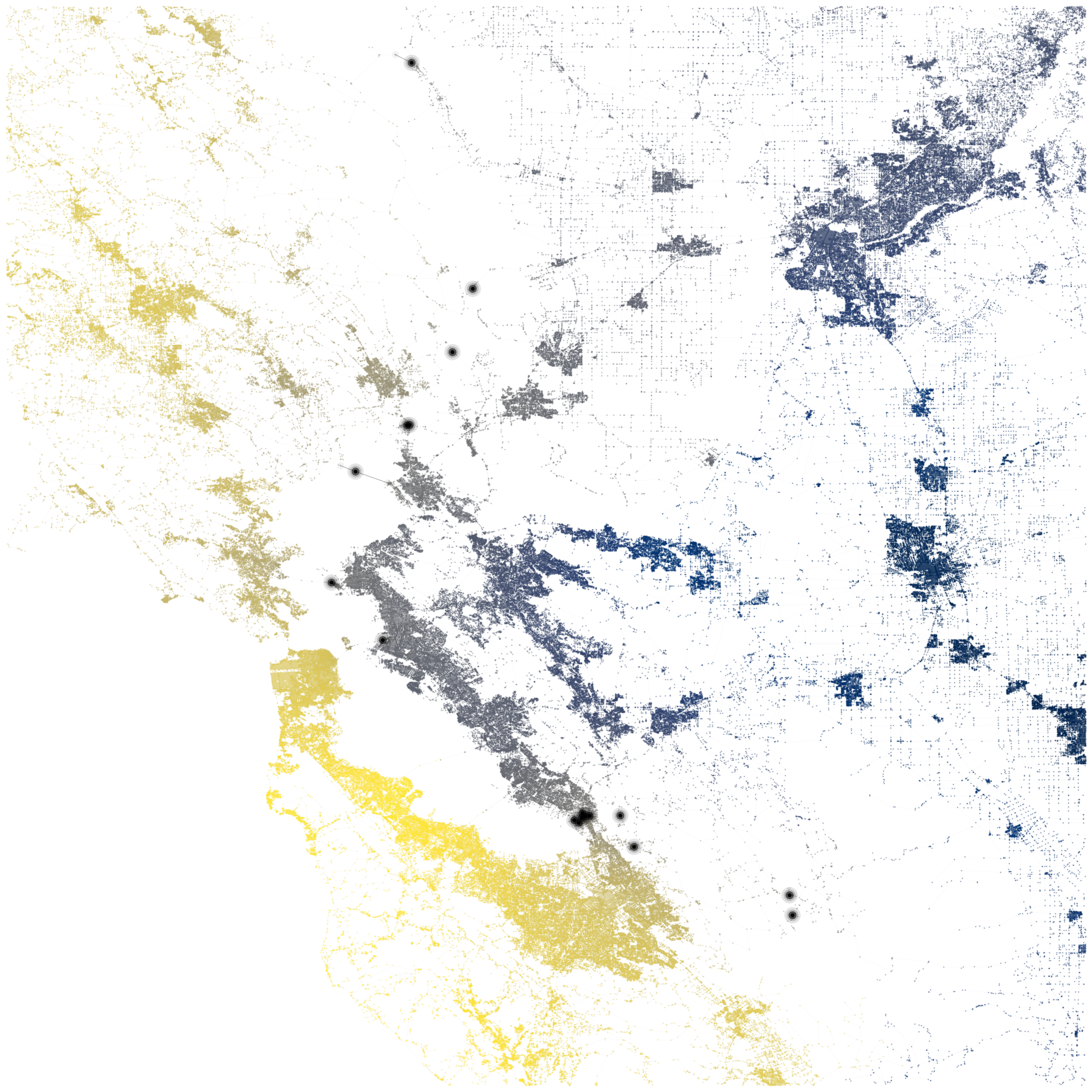}
&
\includegraphics[height=0.17\textheight,keepaspectratio]{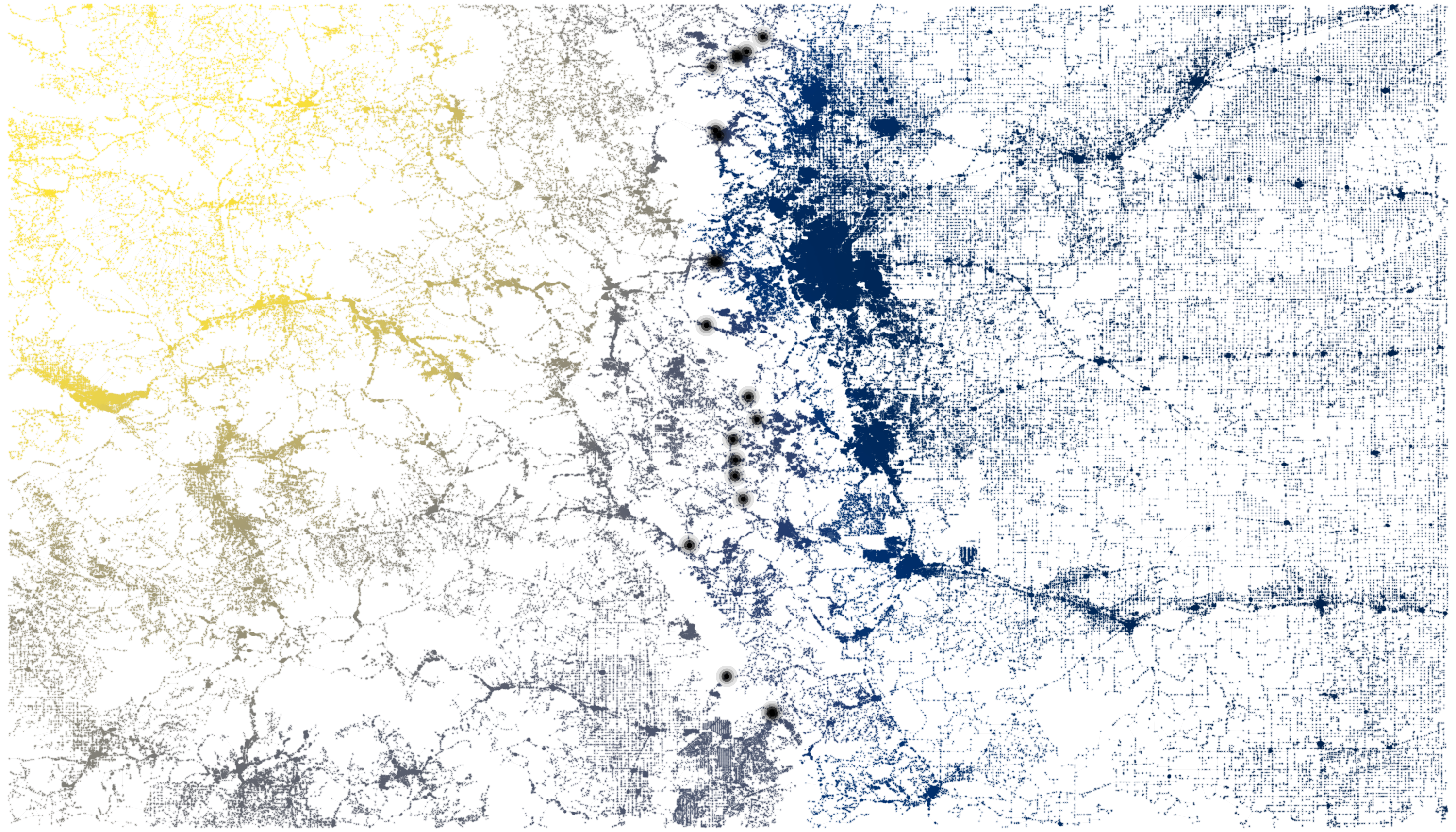}
\end{tabular}
\caption{1D projection visualizations for the DIMACS road-network runs.
From left to right: New York, Bay Area, and Colorado.}
\label{fig:dimacs-road-projections}
\end{figure}

\section{Acknowledgements}

\ifanonymous
\else
Thank you to Hung Le and Jonathan Conroy for helpful comments on presentation.

Thank you to Therese Biedl for pointing out that not every embedded graph in the torus
can be simply triangulated, which led to the $K_5$ example discussed at the start of
\cref{subsec:genus-graphs-geometric-bound}.

Thank you to Arnold Filtser for pointing to the simple algorithm of Krauthgamer and Roughgarden
for approximating padded decomposability~\cite{krauthgamer2011metric}.

Thank you also to Lap Chi Lau for his excellent course in 2022,
in which some of the preliminary results of this work began as a course project.

In parallel to this work,
and also as a part of the same course,
Kam Chuen Tung independently derived
essentially the same result as
\cref{prop:planar-lambda-2-star-restatement},
and similar bounds for $\gamma^{(1)}$
based on $L^2$-extremal spread for bounded-genus graphs and $K_h$-minor-free graphs.
These results have appeared in his recent PhD thesis~\cite[Chapter 7]{tung2025reweighted},
which also includes some interesting generalizations.

After the initial preprint of this work, Lap Chi Lau notified the author
that he was also independently aware of the bound on $\gamma^{(1)}$
for planar graphs stated in \cref{prop:planar-lambda-2-star-restatement} prior to either work.
\fi

The implementation used for this project was written
with significant LLM assistance, via a combination of Claude Code and Codex.
We take full responsibility for the claims and results present in this work.

\bibliographystyle{alpha}
\bibliography{citations}

\pagebreak
\appendix

\section{Graph Class Definitions}
\label{sec:class-defs}

There are a number of graph classes that are relevant to our results.
The most notable is the well-known class of \defn{planar graphs},
which are graphs that can be drawn on the plane without crossing edges.
Most of the other classes we will consider generalize planar graphs in some manner.

\begin{definition}
A graph $G$ is said to have \defn{genus}
$g$ if it can be embedded into an orientable surface of genus $g$,
and it cannot be embedded into an orientable surface of smaller genus.
Planar graphs are exactly those that have genus $0$.
\end{definition}

Up to homeomorphism, there is only one orientable surface of genus $g$,
for any particular $g$.
For $g=0$, it is the sphere.
The orientable surface of genus $g>0$ can be obtained
from the orientable surface of genus $g-1$ by attaching a ``handle''.
Alternatively, it can be expressed as $g$ ``doughnuts'' merged together.
See \cref{fig:surface-examples} for examples of these surfaces,
and see 
\cref{fig:surface-graph-example} for an example of a graph drawn on a genus-$3$ surface.

\begin{figure}[H]
    \centering
    \includegraphics[scale=0.09]{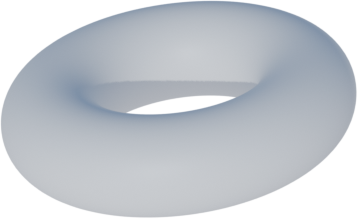}
    \hspace{1em}
    \includegraphics[scale=0.09]{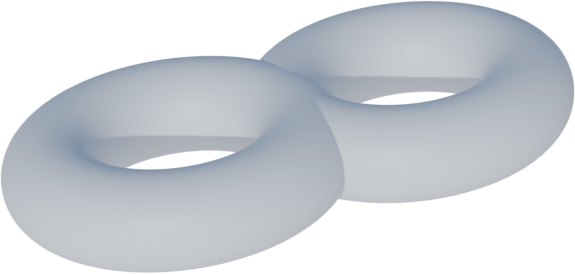}
    \hspace{1em}
    \includegraphics[scale=0.09]{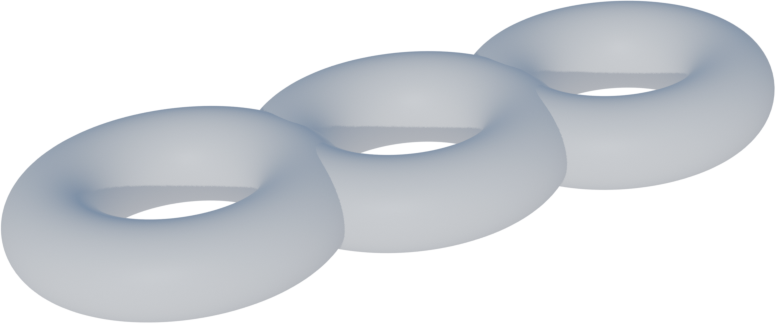}
    \hspace{1em}
    \includegraphics[scale=0.09]{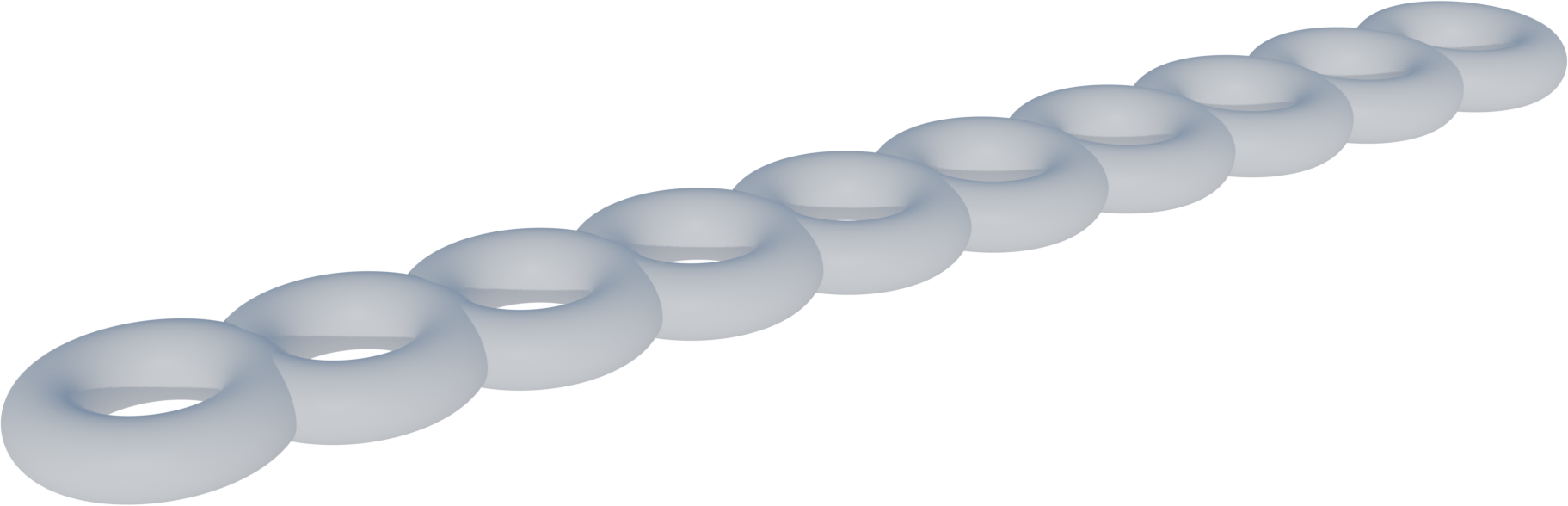}
    \caption{Examples of (left to right) a genus-$1$ surface, a genus-$2$ surface, a genus-$3$ surface,
    and a genus $10$ surface.}
    \label{fig:surface-examples}
\end{figure}

\begin{figure}[H]
    \centering
    \includegraphics[scale=0.6]{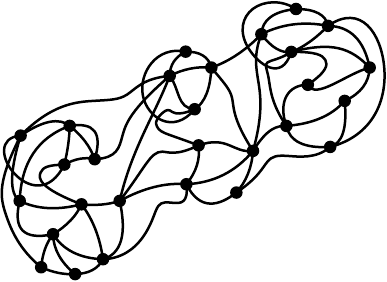}
    \hspace{2em}
    \includegraphics[scale=0.25]{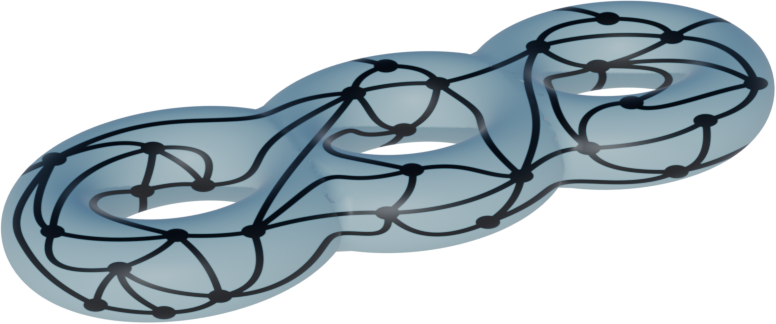}
    \caption{An example of a graph (left) and a drawing of the graph on a genus-$3$ surface (right).}
    \label{fig:surface-graph-example}
\end{figure}

\begin{definition}
For a graph $G$ and a graph $H$,
we say that $H$ is a \defn{minor} of $G$
if it can be obtained from $G$ by a sequence of
edge deletions, vertex deletions,
and edge contractions
(that is, merging the two vertices incident to an edge into one).
Equivalently, $H$ is a minor of $G$ if there are pairwise disjoint,
nonempty vertex sets $X_v\subseteq V(G)$ for $v\in V(H)$
such that every $G[X_v]$ is connected
and, for each $uv\in E(H)$,
some edge of $G$ has one endpoint in $X_u$ and the other in $X_v$.

We say a graph $G$ is
\defn{$H$-minor-free}
if it does not contain $H$ as a minor.
\end{definition}

We will also consider two much more fundamentally geometric graph classes:

\begin{definition}
A graph $G$ is said to be
\defn{the intersection graph of
$d$-dimensional balls with ply $k$}
(sometimes called a $d$-dimensional $k$-ply neighbourhood system)
if there is a collection of
$d$-dimensional balls $B_1,\dots,B_n$
with ply at most $k$ for all points almost surely,
so that each ball corresponds to a vertex, and each edge corresponds to a pair of vertices whose balls intersect.
\end{definition}

\begin{definition}
A \defn{$d$-dimensional $k$-nearest neighbour graph}
is a graph $G$ with $n$ vertices $v_1,\dots,v_n$ corresponding to points $p_1,\dots,p_n\in\RR^d$
so that an edge $v_iv_j$ exists if and only if $p_i$ is among
the $k$-nearest neighbours of $p_j$, or $p_j$ is among the $k$-nearest neighbours of $p_i$.
\end{definition}

\section{A Disproof of a Genus Claim}
\label{sec:short-disproof}

In 2010, Lee and Sidiropoulos claimed the following result~\cite[Theorem 4.2]{lee2010genus}:

\begin{description}
  \item[($\star$)]\ \ Let $G=(V,E)$ be a graph of orientable genus $g$.
Let $(V,d)$ be a metric space with the property that
$(V,\sqrt{d})$ embeds isometrically into a Hilbert space,
and let $d_G$ be the path metric arising from an edge $(u,v)$
having length $d(u,v)$. Then for every $p\geq1$ and product weight
$\omega$, there exists a mapping $f:(V,d_G)\to\RR$
with
$$
\left(\max_{x\neq y\in X}\frac{|f(x)-f(y)|}{d_G(x,y)}\right)
\cdot
\left(
  \frac
  {
    \sum_{x,y\in X}\omega(x,y)d_G(x,y)^p
  }
  {
    \sum_{x,y\in X}\omega(x,y)|f(x)-f(y)|^p
  }
\right)^{1/p}
\lesssim C(p)\cdot\sqrt{\log g}
$$
where $C(p)$ is some constant depending only on $p$.
\end{description}

Lee and Sidiropoulos deferred a proof of this claim
to an absent full version.
This would imply a polynomial-time algorithm to compute
separators of size $\OO(\sqrt{gn\log g})$,
and more generally an $\OO(\sqrt{\log g})$-approximation,
by applying it to existing techniques~\cite{FeigeHL05,FeigeHL08}.
Unfortunately, the theorem as-stated appears to be incorrect,
even for uniform $\omega$ and $p=1$,
and the existence of such algorithms remains open.
We provide a short counter-example here.
We emphasize that it is the \emph{precise} statement for which we give a counter-example.
In fact, our counter-example uses no special properties about genus.
There may be a straightforward fix for their methods as a whole.

\begin{example}
Let $\{G_n=(V_n,E_n)\}_{n\in Z}$,
for some countably-large set $Z\subset\NN$ of integers,
be an infinite collection of
$4$-regular spectral expander graphs%
~\cite[Corollary 1.4]{friedman2008proof}.
That is,
if $A_n$ is the adjacency matrix
of $G_n$,
and its eigenvalues are $4=\alpha_1(A_n)\geq\alpha_2(A_n)\geq\cdots$,
then $\alpha_i(A_n)\leq\sqrt{12}+0.1$
for all $i\geq2$.

For each $n\in Z$, define
the metric $(V_n,d_n)$ by
$$d_n(u,v)=\begin{cases}
  0 & u=v\\
  1 & u\neq v
\end{cases},$$
and define the uniform product weights
$\omega_n(u,v)=1/n^2$.
Then, with $p=1$,
for every $n\in Z$, $G_n$, $(V_n,d_n)$, and $\omega_n$
all satisfy the premises of $(\star)$,
but there is some sufficiently large $n$
for which they cannot satisfy the conclusion.
\end{example}

\begin{proof}
Denote the constants $\lambda:=\sqrt{12}+0.1$ and
$\gamma:=4-\lambda>0$.
Since $\alpha_2(A_n)\leq\lambda<4$,
the graph $G_n$ is connected~\cite[Section 2.3]{hoory2006expander}.
Let $D_n=d_{G_n}$ denote its shortest-path metric.
The hypotheses of $(\star)$
are satisfied since $(V_n,\sqrt{d_n})=(V_n,d_n)$.
Let $f_n:V_n\to\RR$ be any map
for which the denominator in the
conclusion of $(\star)$ is nonzero.

Set
$$
L_n=\max_{x\neq y}\frac{|f_n(x)-f_n(y)|}{D_n(x,y)}.
$$
The quantity in the conclusion of $(\star)$ is invariant under rescaling $f_n$,
so assume without loss of generality that $L_n=1$.
Hence, $|f_n(u)-f_n(v)|\leq1$ for every
$uv\in E_n$.
Moreover, the same quantity is translation-invariant,
so further assume that $\sum_{v\in V_n}f_n(v)=0$.

Since $G_n$ is connected and $4$-regular,
the largest eigenvalue of $A_n$ is $4$,
with eigenspace spanned by $\mathbf{1}$%
~\cite[Section 2.3]{hoory2006expander}.
Since $\sum_{v\in V_n}f_n(v)=0$,
we have $f_n\perp\mathbf{1}$.
By the Rayleigh quotient characterization
of the second eigenvalue%
~\cite[Section 5.2.3]{hoory2006expander},
$$
\alpha_2(A_n)
=
\max_{\substack{h\neq0\\ h\perp\mathbf{1}}}
\frac{h^TA_nh}{h^Th}.
$$
$\alpha_2(A_n)\leq\lambda$, so applying this to $h=f_n$ gives
$$
f_n^TA_nf_n
\leq
\lambda\sum_{v\in V_n}f_n(v)^2.
$$
Therefore,
$$
\gamma\sum_{v\in V_n}f_n(v)^2
=
(4-\lambda)\sum_{v\in V_n}f_n(v)^2
\leq
4\sum_{v\in V_n}f_n(v)^2-f_n^TA_nf_n.
$$

From here, we perform some standard transformations.
First,
$$
f_n^TA_nf_n
=
2\sum_{uv\in E_n}f_n(u)f_n(v).
$$
Second,
$$
\sum_{uv\in E_n}(f_n(u)-f_n(v))^2
=
\sum_{uv\in E_n}(f_n(u)^2+f_n(v)^2)
-
2\sum_{uv\in E_n}f_n(u)f_n(v).
$$
Third, by using $4$-regularity,
$$
\sum_{uv\in E_n}(f_n(u)^2+f_n(v)^2)
=
4\sum_{v\in V_n}f_n(v)^2.
$$
Combining these we get
$$
4\sum_{v\in V_n}f_n(v)^2-f_n^TA_nf_n
=
\sum_{uv\in E_n}(f_n(u)-f_n(v))^2.
$$

Using $|f_n(u)-f_n(v)|\leq1$
for every $uv\in E_n$,
and $|E_n|=2n$,
we get
$$
\gamma\sum_{v\in V_n}f_n(v)^2
\leq
2n.
$$
Hence, using Cauchy-Schwarz,
$$
\frac1{n^2}\sum_{x,y\in V_n}|f_n(x)-f_n(y)|
\leq
\left(
  \frac1{n^2}\sum_{x,y\in V_n}
  (f_n(x)-f_n(y))^2
\right)^{1/2}
=
\left(
  \frac2n\sum_{v\in V_n}f_n(v)^2
\right)^{1/2}
\leq
\frac2{\sqrt\gamma}.
$$

On the other hand,
every ball of radius $r$
in a $4$-regular graph has size at most
$$
1+4\sum_{j=0}^{r-1}3^j
=
2\cdot3^r-1.
$$

Let
$$
r_n=\lfloor\log_3(n/4)\rfloor.
$$
Then $3^{r_n}\leq n/4$, and hence
$$
2\cdot3^{r_n}-1
\leq
n/2-1
\leq
n/2.
$$
Thus every ball of radius $r_n$
has size at most $n/2$.
Hence, for every $x\in V_n$,
at least $n/2$ vertices $y$
satisfy $D_n(x,y)>r_n$,
so
$$
\frac1{n^2}\sum_{x,y\in V_n}D_n(x,y)
\geq
\frac{r_n}{2}
\gtrsim
\log n.
$$
Since $\omega_n(x,y)=1/n^2$,
we get
$$
L_n
\cdot
\frac{
  \sum_{x,y\in V_n}\omega_n(x,y)D_n(x,y)
}{
  \sum_{x,y\in V_n}
  \omega_n(x,y)|f_n(x)-f_n(y)|
}
\gtrsim
\log n.
$$

Since $G_n$ has $2n$ edges,
it can be embedded into a surface of orientable genus
$O(n)$
by creating a handle for each edge.
Thus, the conclusion of $(\star)$ would imply,
for some admissible $f_n$,
$$
L_n
\cdot
\frac{
  \sum_{x,y\in V_n}\omega_n(x,y)D_n(x,y)
}{
  \sum_{x,y\in V_n}
  \omega_n(x,y)|f_n(x)-f_n(y)|
}
\lesssim
C(1)\sqrt{\log(n)}.
$$
This contradicts the lower bound above
for all sufficiently large $n\in Z$.
\end{proof}

\section{Deferred Proof of Subdivision Lemma}
\label{sec:deferred-subdivision-proof}
We now prove \cref{lemma:subdivision-bound-propagation}
by adapting an argument of Kelner~\cite{kelner2004spectral,kelner2006spectral,kelner2006new}.

\begin{lemma}[Restatement of \cref{lemma:subdivision-bound-propagation}]
Let $G$ be a triangulation of genus $g$
with maximum degree $\Delta$.
Then there is some (universal) constant $c$
so that $|V(G)|\cdot\gamma^{(1)}(G)\lesssim \Delta^c\cdot|V(G^{(k)})|\cdot\gamma^{(1)}(G^{(k)})$.
\end{lemma}

The proof will be similar to, but not quite the same as, that of 
\cref{lemma:uniform-shallow-minor-application}.
The primary differences will be that the paths corresponding to edges will be randomly sampled,
and that the ``uniformity'' is only approximate.
It seems quite plausible that a more general form of 
\cref{lemma:uniform-shallow-minor-application} could be extracted from the below proof,
but the most straightforward proof method using these random paths results in an extra factor of $\Delta$.
This factor is perfectly fine for proving
\cref{lemma:subdivision-bound-propagation},
but poses issues for a generalized form of
\cref{lemma:uniform-shallow-minor-application},
so we have kept this proof and construction separate.

\begin{proof}
Denote $G=(V,E)$, $n:=|V|$, $H:=G^{(k)}=(V',E')$, and $n':=|V'|$.
For the remainder of this proof, we will use the notation $\OO_\Delta()$ and $\Theta_\Delta()$ to hide polynomial factors in $\Delta$.
Note that the maximum degree of $H$ is $\max\{6,\Delta\}$, since no new vertices of degree $>6$ are added from a hexagonal subdivision, nor does any existing vertex have its degree changed.
Each edge in $G$ is split into $2^k$ pieces in $H$,
and each triangle in $G$ is partitioned into $4^k$ triangles in $H$.
The number of triangles incident to any vertex of degree $>6$ remains constant during subdivision.
Hence, $\frac{n'}{n}\in\Theta_\Delta(4^k)$.
Let $y_H,f_H$ denote the optimal solution to $\gamma^{(1)}(H)$.
Assume without loss of generality that
$\sum_{x'\in V'}|f_H(x')|^2=1$.
This assumption will allow us to slightly simplify some later steps.

Kelner~\cite[Proof of Lemma 5.2]{kelner2004spectral,kelner2006new}
is able to show that there exist random variables $\pi_V:V\to V'$
and $\pi_E:E\to\{\text{paths through $H$}\}$
with the following properties:
\begin{enumerate}
    \item For each $uv\in E$, $\pi_E(uv)$ is a path in $H$ from $\pi_V(u)$ to $\pi_V(v)$.
    \item For each $v\in V$, $\pi_V(v)$ is a uniform distribution over its support, which we denote $p(v)$.
        Moreover, every vertex $v'\in V'$ is contained in some $p(v)$ for a $v\in V$.
    \item For each $u\neq v\in V$, $\pi_V(u)$ and $\pi_V(v)$ are independent and have disjoint supports.
    \item For each $uv\in E$, $\pi_E(uv)$ is dependent only on $\pi_V(u)$ and $\pi_V(v)$.
    \item Each path $\pi_E(e)$ (for $e\in E$) has length at most $\OO_\Delta\left(2^k\right)$.
    \item For each $v'\in V'$, let $v\in V$ be the unique vertex with $v'\in p(v)$.
        Then $\Pr[\pi_V(v)=v']=\Theta_\Delta(1/4^k)$.
    \item Each edge in $H$ appears in the image of $\pi_E$ with probability $\OO_\Delta(1/2^k)$,
        and moreover each edge (and hence also vertex endpoint) in $H$ appears in the support of $\OO_\Delta(1)$ random variables $\pi_E(e)$ for $e\in E$.
        Denote the support of edges in $\pi_E(e)$ as $p(e)$.
    \item Each vertex $v'\in V'$ appears as an endpoint in the support of some $\pi_E(uv)$
        only if it is contained in $p(u)\cup p(v)$
        (i.e., the vertices incident to elements of $p(uv)$ are a subset of $p(u)\cup p(v)$).
\end{enumerate}
Essentially, this construction is similar to our earlier uniform shallow minors,
except that the uniformity is approximate and ignores factors of $\Delta$.
Additionally, this construction also uses randomized paths
that use only a small number of vertices in expectation.
In contrast, our earlier construction for uniform shallow minors
used an upper bound that accounted for every vertex in every path.

Using these properties, we obtain that for each pair $u,v\in V$,
$$\EX\left[|f_H(\pi_V(u))-f_H(\pi_V(v))|^2\right]\in\Omega_\Delta\left(\frac1{16^k}\sum_{u'\in p(u), v'\in p(v)}|f_H(u')-f_H(v')|^2\right),$$
and if $uv\in E$ then
$$\EX\left[\sum_{u'v'\in\pi_E(uv)}\left(y_H(u')+y_H(v')\right)\right]
\in \OO_\Delta\left(\frac1{2^k}\sum_{u'v'\in p(uv)}\left(y_H(u')+y_H(v')\right)\right).$$
Hence, there is a deterministic choice $\pi_V^*,\pi_E^*$
so that
\begin{align*}
\frac{\sum_{uv\in E}\sum_{u'v'\in\pi_E^*(uv)}\left(y_H(u')+y_H(v')\right)}
{\sum_{u,v\in V}|f_H(\pi_V^*(u))-f_H(\pi_V^*(v))|^2}
&\leq
\frac{\EX\left[\sum_{uv\in E}\sum_{u'v'\in\pi_E(uv)}\left(y_H(u')+y_H(v')\right)\right]}
{\EX\left[\sum_{u,v\in V}|f_H(\pi_V(u))-f_H(\pi_V(v))|^2\right]}\\
&\in\OO_\Delta\left(
    8^k
\frac{
\sum_{uv\in E}\sum_{u',v'\in p(uv)}\left(y_H(u')+y_H(v')\right)
}
{
\sum_{u\neq v\in V}\sum_{u'\in p(u),v'\in p(v)}|f_H(u')-f_H(v')|^2
}
\right)\\
&=\OO_\Delta\left(
    8^k
\frac{
\sum_{v'\in V'}y_H(v')
}
{
\sum_{u\neq v\in V}\sum_{u'\in p(u),v'\in p(v)}|f_H(u')-f_H(v')|^2
}
\right).
\end{align*}
Let $\rho:=\max_{uv\in E}|\pi_E^*(uv)|$
be the maximum length (in terms of edges) of a sampled path,
so $\rho\in\OO_\Delta(2^k)$.
Choose $f_G(v):=f_H(\pi_V^*(v))-\frac1n\sum_{x\in V}f_H(\pi_V^*(x))$,
and $y_G(v):=2\rho
\sum_{u:uv\in E}\sum_{x\in p(v)\cap\pi_E^*(uv)}y_H(x)
$.

We start by showing that the constraints are satisfied.
By definition, $\sum_{v\in V}f_G(v)=\overline 0$.
Consider some edge $uv\in E$,
and the corresponding path $\pi_E^*(uv)$ through $H$.
By Cauchy-Schwarz and the triangle inequality (similar to the argument in
\cref{lemma:dim-red-Q-alpha}),
$$y_G(u)+y_G(v)\geq|f_G(u)-f_G(v)|^2,$$
so the constraints are satisfied.

It remains only to check the objective value.
First, note that
$$2n\sum_{x\in V}|f_G(x)|^2=\sum_{u,v\in V}|f_G(u)-f_G(v)|^2=\sum_{u,v\in V}|f_H(\pi_V^*(u))-f_H(\pi_V^*(v))|^2.$$
Moreover,
$$\sum_{u\neq v\in V}\sum_{u'\in p(u),v'\in p(v)}|f_H(u')-f_H(v')|^2$$
$$=\sum_{u',v'\in V'}|f_H(u')-f_H(v')|^2-\sum_{v\in V}\sum_{u',v'\in p(v)}|f_H(u')-f_H(v')|^2.$$
For any $v\in V$, any two $u',v'\in p(v)$, and any $w'\in V'\setminus p(v)$,
the triangle inequality gives
$|f_H(u')-f_H(v')|\leq|f_H(u')-f_H(w')|+|f_H(v')-f_H(w')|$,
so
by Jensen's inequality on the convex function $x\mapsto x^2$,
$\frac12|f_H(u')-f_H(v')|^2\leq|f_H(u')-f_H(w')|^2+|f_H(v')-f_H(w')|^2$.
In other words,
every squared distance within $p(v)$
can be bounded above by two squared distances
crossing out of $p(v)$.
So long as there are at least two vertices,
there are always more of the latter type (up to factors in $\Delta$),
so it follows that
$$\sum_{u\neq v\in V}\sum_{u'\in p(u),v'\in p(v)}|f_H(u')-f_H(v')|^2
\in\Omega_\Delta\left(\sum_{u',v'\in V'}|f_H(u')-f_H(v')|^2\right).$$

Finally,
\begin{align*}
\frac{\sum_{v\in V}y_G(v)}{\sum_{v\in V}||f_G(v)||_2^2}
&=
\frac{4n\rho
\sum_{v\in V}
\sum_{u:uv\in E}\sum_{x\in p(v)\cap\pi_E^*(uv)}y_H(x)
}{\sum_{u,v\in V}|f_H(\pi_V^*(u))-f_H(\pi_V^*(v))|^2}\\
&\in
\OO_{\Delta}\left(
n\rho\cdot
8^k
\frac{
\sum_{v'\in V'}y_H(v')
}
{
\sum_{u\neq v\in V}\sum_{u'\in p(u),v'\in p(v)}|f_H(u')-f_H(v')|^2
}
\right)\\
&=
\OO_{\Delta}\left(
n\rho\cdot
8^k
\frac{
\sum_{v'\in V'}y_H(v')
}
{
\sum_{u',v'\in V'}|f_H(u')-f_H(v')|^2
}
\right)\\
&=
\OO_{\Delta}\left(
n\rho\cdot
8^k
\frac{
\sum_{v'\in V'}y_H(v')
}
{
2n'\sum_{v'\in V'}|f_H(v')|^2
}
\right)\\
&=
\OO_{\Delta}\left(
4^k
\frac{
\sum_{v'\in V'}y_H(v')
}
{
\sum_{v'\in V'}|f_H(v')|^2
}
\right).
\end{align*}
\end{proof}

\section{Deferred Proof of Area Inequality}
\label{sec:deferred-area-proof}

In this section, we prove \cref{lemma:area-ineq}.

\begin{lemma}[Restatement of \cref{lemma:area-ineq}]
Let $d\geq2$,
and let $C$ be a sphere cap on the $d$-dimensional unit sphere embedded in $\RR^{d+1}$.
Let $f$ be its centre
and
let $r$ be the maximum Euclidean distance in $\RR^{d+1}$
from $f$ to a point in $C$.
Assume $C$ contains at most half the surface area of the sphere.
Then the $d$-dimensional
surface area of the sphere cap
satisfies
$\text{area}(C)\geq\frac{A_dr^d}{4^d}$.
\end{lemma}

The proof will use three important well-understood functions:
The Gamma function $\Gamma$, the Beta function $B$, and the regularized incomplete Beta function $I$.
We will not require the explicit representation of the Gamma function,
but we will use Gautschi's inequality:

\begin{proposition}[{Gautschi's inequality~\cite{gautschi1959some,NIST_DLMF_Gautschi}}]
For $x\geq0,s\in(0,1)$, it holds that
$$x^{1-s}<\frac{\Gamma(x+1)}{\Gamma(x+s)}<(x+1)^{1-s}.$$
\end{proposition}

We will use the explicit forms of the Beta and regularized incomplete Beta functions:

\begin{definition}
For $z_1,z_2,x>0$,
the \defn{incomplete Beta function} is defined as
$B(x;z_1,z_2):=\int_0^x t^{z_1-1}(1-t)^{z_2-1}dt$.
The \defn{Beta function} is defined as
$B(z_1,z_2):=B(1;z_1,z_2)$,
and it is also known to be equal to $\frac{\Gamma(z_1)\Gamma(z_2)}{\Gamma(z_1+z_2)}$.
The \defn{regularized incomplete Beta function}
is defined as $I_x(z_1,z_2):=\frac{B(x;z_1,z_2)}{B(z_1,z_2)}$.
\end{definition}

We also have the following important result for our purposes:

\begin{proposition}[{\cite[Equation (1)]{li2011concise}}]
For a sphere cap $C$ with height $h\leq1$
of the $d$-dimensional unit sphere,
$\text{area}(C)=\frac12 A_d\cdot I_{2h-h^2}\left(\frac d2,\frac12\right)$.
\end{proposition}

Note that the convention for dimension differs in the reference.

\begin{proof}[Proof of \cref{lemma:area-ineq}]
Let $h=\frac{r^2}2$ be the height of the sphere cap
(see \cref{fig:r-h-relationship}).
Then,
$$B\left(2h-h^2;\frac d2,\frac12\right)
=\int_0^{2h-h^2}t^{\frac d2-1}(1-t)^{-\frac12}dt
\geq\int_0^{2h-h^2}t^{\frac d2-1}dt
=\frac{2(2h-h^2)^{\frac d2}}d.
$$

Next,
$$B\left(\frac d2,\frac12\right)
=\frac{\Gamma\left(\frac d2\right)\cdot\Gamma\left(\frac12\right)}{\Gamma\left(\frac d2+\frac12\right)}
=\frac{\Gamma\left(\frac d2\right)\cdot\sqrt{\pi}}{\Gamma\left(\frac d2+\frac12\right)}
<\sqrt{\pi}\left(\frac d2-\frac12\right)^{-\frac12},
$$
where the last inequality follows from Gautschi's inequality.

Now,
$2h-h^2=h(2-h)\geq h$ (since $h\leq1$), so
$$I_{2h-h^2}\left(\frac d2,\frac12\right)\geq
\frac{
\frac{2(2h-h^2)^{\frac d2}}d
}{
\sqrt{\pi}\left(\frac d2-\frac12\right)^{-\frac12}
}
\geq
\frac{2h^{\frac d2}}{\sqrt{\pi}d}
\left(\frac d2-\frac12\right)^{\frac12}
=
\frac{2h^{\frac d2}}{\sqrt{2\pi}}
\frac{\sqrt{d-1}}d.
$$

Finally,
$$\text{area}(C)
=\frac12 A_d\cdot I_{2h-h^2}\left(\frac d2,\frac12\right)
\geq
\frac12 A_d\cdot
\frac{2h^{\frac d2}}{\sqrt{2\pi}}
\frac{\sqrt{d-1}}d
=A_d\cdot
\frac{r^d}{2^{\frac d2}\sqrt{2\pi}}
\frac{\sqrt{d-1}}d
\geq
A_d\cdot
\frac{r^d}{4^d},
$$
where the last inequality holds since $d\geq2$.
\end{proof}

\end{document}